\documentclass[pdflatex,sn-mathphys-ay]{sn-jnl}

\usepackage{graphicx}%
\usepackage{multirow}%
\usepackage{amsmath,amssymb,amsfonts}%
\usepackage{amsthm}%
\usepackage{mathrsfs}%
\usepackage[title]{appendix}%
\usepackage{xcolor}%
\usepackage{textcomp}%
\usepackage{manyfoot}%
\usepackage{booktabs}%
\usepackage{algorithm}%
\usepackage{algorithmicx}%
\usepackage{algpseudocode}%
\usepackage{listings}%
\usepackage{float,adjustbox}
\usepackage{tikz}
	\usetikzlibrary{arrows.meta}
	\usetikzlibrary{shapes}
	\usetikzlibrary{calc}
	\usetikzlibrary{positioning,intersections,fit}
	\usetikzlibrary{shapes.geometric}
	\usetikzlibrary{decorations.markings,decorations.pathreplacing}
	\usetikzlibrary{fit}
\usepackage{environ}
\usepackage{epsfig,graphicx,rotating}
\usepackage{etoolbox,booktabs}
\usepackage{subcaption}
\usepackage{arydshln}
\usepackage{enumitem}

\allowdisplaybreaks

\def\circleB{(0,0) circle (1)}
\def\circleC{(1,0) circle (1)}
\def\circleA{(0.5,1) circle (1)}

\theoremstyle{thmstyleone}%
\newtheorem{theorem}{Theorem}
\newtheorem{proposition}[theorem]{Proposition}%

\theoremstyle{thmstyletwo}%
\newtheorem{example}{Example}%
\newtheorem{remark}{Remark}%

\theoremstyle{thmstylethree}%
\newtheorem{definition}{Definition}%

\begin{document}

\title[Article Title]{A diagrammatic proof-theoretic semantics for the Greimas semiotic square}


\author*[1]{Michael D. Fowler}\email{synthifou@gmail.com}

\affil*[1]{Technische Universit\"{a}t Berlin, Strasse des 17. Juni 135. 10623 Berlin Germany}


\abstract{We develop a diagrammatic proof system for a fragment of structural semantics inspired by the Greimas semiotic square, using  spider diagrams as the underlying formalism. The basic terms are represented as diagrammatic configurations, and the relations of contradiction and implication are interpreted as transformations governed by a set of inference rules. These transformations are realised as derivations, with proof trees serving as witnesses. Our main result shows that the construction of the four meta-terms can be captured uniformly: each is derivable from a conjunctive pair of basic configurations via a fixed derivation schema composed of contour introduction and habitat transformation rules. This yields a proof-theoretic account of the combinatorial operation underlying meta-term formation, and provides a semantic interpretation of the Greimasian operation `+' as a derivational construction rather than a logical combination. We further show that diagrammatic negation in this setting is not a Boolean complement, but a restricted, zone-determined semantic counter-position, reflecting the relational character of opposition in structural semiotics. The resulting framework provides a compositional, rule-based semantics in which complex configurations are generated constructively from simpler ones. In addition to extending the expressive scope of spider diagram calculi, the system illustrates how diagrammatic reasoning can be used to formalise non-classical semantic operations within a unified inferential setting.}

%



\maketitle

\section{Introduction}

Diagrammatic reasoning systems provide formal frameworks in which logical structure is represented and manipulated visually through rule-governed transformations. Among such systems, spider diagrams offer a well-developed diagrammatic calculus for reasoning about set-theoretic and existential information, equipped with a sound and complete system of inference rules. In this paper, we develop a diagrammatic proof system based on unitary spider diagrams for representing and transforming a class of semantic configurations inspired by the semiotic square. 

The semiotic square is an analytic and heuristic tool for investigating signification in literary contexts and was developed in the early 1960s by Algirdas Julien Greimas (1917-1992) who initiated a research program to develop a structuralist theory of meaning that was based on how oppositional relations between concepts iterate within narratives. As Hawkes suggests, Greimas attempted ``to describe narrative structure in terms of an established linguistic model derived from the Saussurean notion of an underlying \textit{langue} or competence which generates a specific \textit{parole} or performance, as well as from Saussure's and Jakobson's concept of the fundamental signifying role of binary opposition.'' \citep[page 69]{Hawkes2003} 

Greimasian structural semiotics establishes the notion of a seme (denoted /seme/) as the smallest irreducible unit of meaning for which when in combination, form a sememe, or the smallest unit of meaning that associates to a morpheme or lexeme, and is thus analogous to a phoneme in phonology. For example the sememe `cat' contains the semes of /animal/, /pet/, /feline/ etc. What Greimas calls a homologation is defined as a structure in which 
\begin{quote}
two sememes $S$ and $S^{\prime}$ will be called homologous in relation to non-$S$ and non-$S^{\prime}$ if they have in common a semic content $s$ (that is to say, at least one seme) which, considered as a positive term, is present at the same time, under its negative form of non-$s$, in the sememes non-$S$ and non-$S^{\prime}$. \citep[page xxxi]{Greimas1983}
\end{quote}
As Hawkes suggests, this framing is heavily influenced by both the structural linguistics of \citet{Saussure1983} as well as how \citet{Levi-Strauss1955} identified constitutive units of myths as emerging from the condition:
\begin{equation}
\frac{A}{\mathrm{non-}A} \,\,\cong\,\, \frac{B}{\mathrm{non-}B}\,\,.
\end{equation}
Greimas and Court\'{e}s similarly consider this condition as not only relevant to myths or folktales, but more generally applicable to all semiotic domains of cultural production (i.e., texts, art, music). Homologation is thus considered a ``procedure of structuration'' \citep[page 144]{GreimasCourtes1982} which is commonly formulated as:
\begin{equation}\label{homologationEG}
A : B :: A^{\prime} : B^{\prime}.
\end{equation}
Under the analytic lens of the structural semiotics of Greimas, sign systems then become assemblages of differences where meaning arises not from inherent qualities but from relational positions within a structure. A seme itself is thus not autonomous or ``atomistic,'' but instead ``exists only because of the differential gap that opposes it to other semes.'' \citep[page 278]{GreimasCourtes1982} By this we mean that semes themselves follow the type of structure of opposition given in Equation~(\ref{homologationEG}), and form what Greimas calls semic categories \cite[page 16]{Greimas1987}. These semic categories constitute the content plane, and therefore are considered anterior to the collection of semes that form the categories themselves. But there is structure here too, which becomes fully formed and visualized through Greimas' most well-known contribution, the semiotic square. Here, \citet[page 88]{GreimasRastier1968} define a family of relations that include:  contrariety  (\raisebox{0pt}{\tikz{\draw[thick,dotted,<->] (0,0)--(0.6,0);}}), contradiction (\raisebox{0pt}{\tikz{\draw[<->] (0,0)--(0.6,0);}}) and implication (\raisebox{0pt}{\tikz{\draw[thick,dashed,->] (0,0)--(0.6,0);}}) between a seme pair $(s_{1}, s_{2})$ and their negations $(\neg s_{1}, \neg s_{2})$, and visualise these relations through the diagram given in Figure~\ref{EQSemioticSquare}.

\begin{figure}
\centering
\begin{tikzpicture}[scale=0.75]
\node at (0,0)(s1){$s_{1}$};
\node at (4,0)(s2){$s_{2}$};
\node (not-s2) at (0,-4){$\neg s_{2}$};
\node at (4,-4)(not-s1){$\neg s_{2}$};
\draw[thick,dotted,<->] (s1) -- node[above] {\scalebox{1}{$\bf{S}$}} (s2);
\draw[thick,dotted,<->]  (not-s1) -- node[below]{\scalebox{1}{$\bf\bar{S}$}}(not-s2);
\draw[thick,dashed,->] (not-s2)--node[left,xshift=-5]{\footnotesize Positive Deixis}(s1);
\draw[thick,dashed,->] (not-s1)--node[right,xshift=5]{\footnotesize Negative Deixis}(s2);
\draw[<->](s1)--(not-s1);
\draw[<->](s2)--(not-s2);
\end{tikzpicture}
\caption{The Greimas semiotic square.}
\label{EQSemioticSquare}
\end{figure}
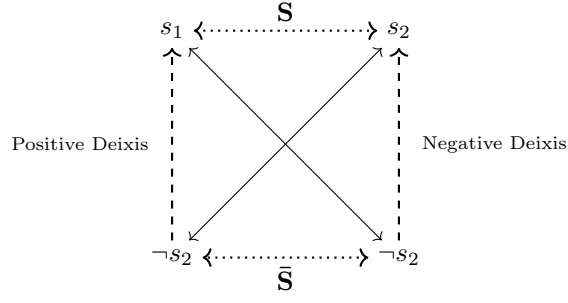

The seme pair of $\{s_{1},s_{2}\}$ is known as the complex axis, and $\{\neg s_{1}, \neg s_{2}\}$ the neutral axis. The four central terms of $\{s_{1}, s_{2},\neg s_{1}, \neg s_{2}\}$ thus allow for other meta-terms as sememes to be derived. The procedures for these constructions are found through the ``combing of the four simple terms.'' \citep[page 81]{Hebert2020} This yields 
\begin{align*}
\bf{S}=& (s_{1} + s_{2}), \quad \mathrm{complex\,\, term},\\
\bf\bar{S}=& (\neg s_{1} + \neg s_{2}), \quad \mathrm{neutral\,\, term},
\end{align*}
as well as the Positive Deixis = ($ s_{1} + \neg s_{2}$), and Negative Deixis = ($\neg s_{1} + s_{2}$). The Greimas square has become the most popular tool in the structural semiotics method for identifying and tracking of structures of homologation within narratives, with Greimas and Court\'{e}s  arguing that it represents ``a rigorous formulation of reasoning by analogy.'' \citep[page 144]{GreimasCourtes1982} 

Yet, despite its importance, a formalisation of the process by which homologations are established within the semiotic square often remains under-specified. The contribution of this paper is to extend spider diagram calculi of \citet{HowseStapletonTaylor2005} to a new semantic domain by providing a proof-theoretic reconstruction of the operations underlying the Greimas square. Specifically, spider diagrams allow for a formalisation of the under-specification found in the initial binary term-pair $s_{1}$ and $s_{2}$. Here, Court\'{e}s describes this binary pair as formed through a `contrariety' (denoted $s_{1}/s_{2}$), which corresponds to the notion that two terms
\begin{quote}
may be said to be contrary when the presence of one presupposes the presence of the other, and when the absence of one goes hand in hand with the absence of the other. More generally speaking, two terms (S1 and S2) are said to be contrary if the negation of one implies the affirmation of the other and vice versa. \citep{Courtes1991}
\end{quote}
H\'{e}bert acknowledges the subtleties of distinguishing between the relation of contrariety and contradiction in the square through defining contradiction as ``categorical'' and contrariety as ``gradual'':
\begin{quote}
. . .  in classical logic, true and false are contradictories, since not-true equals false and not-false equals true; however, rich and poor are contraries, since not-rich does not necessarily equal poor, and not-poor does not necessarily equal rich . . . Opposition can be viewed in several diﬀerent ways: as a comparative relation at the same level as alterity and identity, or as a subspecies of alterity, or as a subspecies of similarity, among others. Actually, elements set in opposition are comparable, and therefore similar: for example, day and night are opposable because both are time spans in a day (a shared property). \citep[page 32]{Hebert2020} 
\end{quote}
As such, \citet[page 4]{Rossolatos2012} identifies the $s_{1}$/$s_{2}$ binary pair as one in which ``contrariety constitutes a deflected or fuzzier form of contradiction.'' 

Unlike the logical hexagon \citep{Moretti2012,Beziau2012}, the semiotic square is not reducible to an oppositional-geometric structure. Each vertex represents semic zones inside a restricted semantic universe generated by the square, where negation is interpreted as a semic counter-position rather than a global Boolean complement. A close reading of \citet[page 50]{Greimas1987} reveals that the square functions on a hierarchy of terms (incl. seme and sememes) such that the meta-terms are created through hyponymic relations on semes. This corresponds to a Venn diagram describing the following syntagmatic universe: 

\begin{equation}\label{Ex1}
\resizebox{0.25\textwidth}{!}{
\begin{tikzpicture}[font=\Large,line width=0.5mm]
	\scope[even odd rule]
	\clip \circleB;
	\clip \circleC;
	\fill[red] \circleB ;
	\endscope
\draw \circleB node[label={[label distance=8mm]-90:$S_{1}$}] {};
\draw \circleC node [label={[label distance=8mm]-90:$S_{2}$}] {};
\draw (0.5,0) circle (2) node[label={[label distance=20mm]90:$M$}]{};
\draw (0.5,0.5) circle (3) node[label={[label distance=30mm]90:$X$}]{};
\end{tikzpicture}
}
\end{equation}

Here, according to \citet[page 338]{GreimasCourtes1982}, $X$ is a set of generic terms that are ``defined only as intersection points of various relations,'' for which ``any term of the semiotic square is an intersection point of the relations of contrariety, contradiction, and complementarity.'' Within $X$, the subset $M$ is a collection of compound terms or meta-terms that are derived from `$+$' operations (i.e. $\mathbf{S}, \mathbf{\bar{S}}$, Postivie Deixis, Negative Deixis), and $S_{i}$ are subset collections of semes. The semes are divided into two contraneous sets $S_{1}$ and $S_{2}$ that are demarcated according to their position ($1=$ left, $2=$ right) in the square rather than their lexical categorisation. The intersection of $S_{1}$ and $S_{2}$ is empty and denoted by a shaded region. 

In terms of its construction (i.e., its application to a source text), Greimas and Court\'{e}s argue that the sememes $\bf{S}$ and $\bf\bar{S}$ (cf. Figure~(\ref{EQSemioticSquare})) are assigned through the link between the ``positive'' term $s_{1}$ and its negation $\neg s_{1}$, (or duly the ``negative'' term $s_{2}$ and its negation $\neg s_{2}$). This is known as a ``structuration'' in which there is a ``reduction of parasynonomous sememic occurrences into classes, and, on the other hand, the homologation of the identified sememic categories (or sememic oppositions).'' \citep[page 318]{GreimasCourtes1982} Van Wolde argues that this process implies series of `operations':
\begin{quote}
The fundamental syntactic description shows that the articulation of meaning, as described in the taxonomic model of the semiotic square, is rendered dynamic. This happens by means of actions or operations on the semes. The first operation is negation, which runs along the axis of contradiction: by means of negation $s_{1}$ generates non-$s_{1}$ and $s_{2}$ generates non-$s_{2}$. The second operation is assertion, which focusses on the sub-contrary terms: non-$s_{1}$ remains non-$s_{1}$, non-$s_{2}$ remains non-$s_{2}$. Assertion entails implication: non-$s_{1}$ becomes $s_{2}$, and non-$s_{2}$ becomes $s_{1}$. These operations of negation, assertion and implication succeed one another, amounting to one big transformation. There are two possibilities here: either $s_{1}$ is transformed into $s_{2}$ via non-$s_{1}$, or $s_{2}$ is transformed into $s_{1}$ via non-$s_{2}$. In short, the elementary semantic structure with its four logical semic terms and six relations (fundamental semantics) is made dynamic by the operations of negation and affirmation so that there can be changes in the semantic structure (fundamental syntax) together they constitute the first stage (fundamental grammar) of the generation of meaning. \citep[page 13]{vanWolde1989}
\end{quote}
In this article we draw on van Wolde's seme-wise operations to derive the complex ($\mathfrak{c}$) and neutral ($\mathfrak{n}$) axes of the Greimas square via
\begin{equation}\label{vanWoldeEq}
s_{1}\rightarrow \neg s_{1} \rightarrow s_{2},\quad  \mathrm{and} \quad s_{2}\rightarrow \neg s_{2} \rightarrow s_{1},
\end{equation}
within the logical framework and semantics of unitary spider diagrams. This positions the semiotic square not only as an analytic tool, but as a generative semantic system whose transformations can be made explicit within a formal diagrammatic framework. In turn, this allows the generative procedures that are only implicitly present in the traditional formulation of the semiotic square to become explicit through such a formalisation. Within the framework, unitary spider diagrams furnish a semantics in which instances of $s_{1}$ and $s_{2}$ are propagated through successive operations of negation and implication, while proof trees encode the compositional structure of these transformations. Building on the van Wolde operational account, we extend this derivational perspective to the construction of the meta-terms, including the complex and neutral axes, by treating them as the result of composed diagrammatic inferences. 

\subsection{Summary}

The article proceeds as follows. In Section~\ref{unitaryspiderdiagrams}, we introduce unitary spider diagrams as a formal diagrammatic language and reinterpret their set-theoretic semantics in terms of semantic configurations, where contours represent semes and spiders encode their actualisation and virtualisation. Section~\ref{semioticsquareapplication} applies this framework to the Greimas square by defining diagrammatic counterparts of the core relations of contrariety, negation, and implication, thereby modelling the transitions between semes as transformations of diagrams. In Section~\ref{prooftreesdiagramming}, we identify a system of inference rules and show how these operations are realised as derivations in a diagrammatic proof system, with proof trees witnessing each transformation. Section~\ref{derivingmetaterms} extends this approach to the construction of the meta-terms of the semiotic square. We prove that each meta-term is derivable from a conjunctive seme-pair via a uniform inference schema, and interpret this construction as a diagrammatic semantic lift in which a new configuration is introduced at the level of $M$. This provides a formal account of the Greimasian operation ‘+’ as a derivational process rather than a logical or set-theoretic combination. We conclude by discussing the implications of this framework for the formalisation of structural semiotics and the role of diagrammatic reasoning in modelling non-classical semantic operations.

\section{Unitary Spider Diagrams}\label{unitaryspiderdiagrams}

Diagrammatic reasoning has long played a role in mathematical practice, and more recently has motivated the development of formal systems in which diagrams serve not merely as heuristic aids, but as objects of logical inference in their own right \citep{UrbasJamnikStapleton2015}. Such systems exploit the immediacy and structural transparency of diagrams to support reasoning tasks that are often more cumbersome in symbolic form. Within the domain of semiotics, Greimas similarly emphasised the importance of diagrammatic structures as heuristic and hermeneutic tools for analysing meaning. The semiotic square itself may be understood as a diagrammatic articulation of relational structure, originally inspired by forms analogous to Cayley graphs. \citep{Fowler2025} In this sense, both diagrammatic logic and structural semiotics share a common concern with how relational configurations may be made explicit and manipulable.

In this article, we bring these perspectives together by adopting unitary spider diagrams \citep{HowseStapletonTaylor2005} as a formal diagrammatic language for modelling semic relations. Spider diagrams extend Venn diagrams by incorporating existential information and supporting a well-defined system of inference rules, enabling diagrammatic transformations to be organised into proof trees \citep{DelaneyStapletonTaylorThompson2013}. While traditionally interpreted in set-theoretic terms, we instead treat spider diagrams as representations of semantic configurations, where contours correspond to semes or sememes, and spiders encode their actualisation or virtualisation. Here we drawn on what Greimas and Court\'{e}s note as the conditions in which subjects and objects in a narrative are virtualised prior to their junction, for which ``their actualization and their realization take place in accordance with the two types of characteristic relations of the function; disjunction actualizes subjects and objects, conjunction realizes them.'' \citep[page 9]{GreimasCourtes1982} This framing allows diagrammatic inference to be understood as modelling transformations within a semic structure, thereby aligning spider diagram reasoning with the generative processes underlying the semiotic square. Our definition of unitary spider diagrams is adapted from that given by \citet{HowseStapletonTaylor2005}.

\begin{definition}[Unitary spider diagram]
Let the 4-tuple $d=\langle L, Z, Z^{\ast},SI\rangle$ denote a \emph{unitary spider diagram} with labels in the set $\mathcal{L}$ for which the following holds:
\begin{enumerate}
\item $L=L(d) \in \mathbb{F}\mathcal{L}$ is a finite set of contour labels.
\item $Z=Z(d) \subseteq \{(a, L - a)\,|\, a \subseteq L\}$ is a set of \textit{zones} $(Z(d) \subseteq \mathcal{Z})$ for which
\begin{enumerate}
\item[(i).] $\forall l \in L$ there exists $(a,L - a) \in Z$ such that $l \in a$, and
\item[(ii).] $(\varnothing, L) \in Z$, such that $r=R(d)=\mathcal{P}Z - \{\varnothing\}$ is the set of \textit{regions}.
\end{enumerate}
\item $Z^{\ast} = Z^{\ast}(d) \subseteq Z$ is the set of \textit{shaded zones}. We define the set of shaded regions through $R^{\ast}=R^{\ast}(d)=\mathcal{P}Z^{\ast} - \{\varnothing\}$.
\item Let $SI=SI(d) \subset \mathbb{Z}^{+} \times R(d)$ be a finite set of \textit{spider identifiers} for which:
\begin{equation}
\forall(n_{1},r_{1}),(n_{2},r_{2})\in SI \wedge r_{1}=r_{2} \Rightarrow n_{1}=n_{2}.
\end{equation}
Given $(n,r)\in SI$, it follows that there exists $n$ spiders with \textit{habitat} $r$.
\end{enumerate}
\label{SpiderDiagramDefinition}
\end{definition}

We utilise the properties of Venn diagrams in our unitary spider diagrams such that all contours intersect. This allows for the topological qualities of enclosure, intersection and exclusion to be visualised through closed curves that indicate set boundaries. We also follow \citet{Howse2001} and use a boundary rectangle to enclose a diagram, though this boundary is not considered a contour itself. Each contour contains a label in the set $\mathcal{L}$ (where $\mathbb{F}\mathcal{L}$ is the set of all finite subsets of $\mathcal{L}$), such that the labels of a diagram $d$ are denoted $L(d)$. A basic region $r$ in a spider diagram is an area enclosed by a contour for which we define zones as regions that contain no other regions within their contours, and are given through an ordered pair $(in,\,out)$. 

Zones or regions may contain spiders (black dots) that represent existential qualifications. A spider whose habitat spans multiple regions expresses existential uncertainty across those regions. If a spider inhabits a region $r$, it asserts the existence of an element in that region. If a spider spans several regions, it asserts the existence of an element in the union of those regions, without committing to any single region individually. Thus in our reinterpretation, spiders are taken to indicate the actualisation of semes within a given semantic configuration. A spider whose habitat is confined to a single region corresponds to a determinate actualisation, whereas a spider spanning multiple regions expresses a form of virtualisation, in which the realisation of a seme is left underdetermined across several possibilities.
Consequently, the more regions a spider spans, the weaker (more permissive) the existential claim becomes. This distinction between localised and spanning spiders will be important when interpreting diagrams as statements about the actualisation and virtualisation of semes. 

\begin{example}
Consider the following unitary spider diagram $d$:
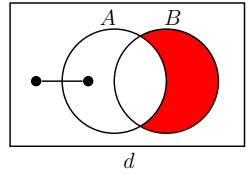
\begin{equation}\label{Ex1}
\resizebox{0.25\textwidth}{!}{
\begin{tikzpicture}
\draw[thick] (-2,-1.25) rectangle (2.5,1.5)node[below,pos=0,xshift=65]{\large$d$};
	\scope
	\clip (-1,-1) rectangle (2,1)
      	(0,0) circle (1);
	\fill[red] (1,0) circle (1);
	\endscope
\draw[thick] (0,0) circle (1)node[above left,label={[yshift=20]90:{\large$A$}}] (white) {};
\draw[thick] (1,0) circle (1)node[above right,label={[yshift=20]90:{\large$B$}}] (black) {};
\node[circle,minimum size=2mm,inner sep=0,fill=black] at (-0.5,0)(p1){};
\node[circle,minimum size=2mm,inner sep=0,fill=black] at (-1.5,0)(p2){};
\draw[thick](p1)--(p2);
\end{tikzpicture}
}
\end{equation}
We denote the set of spiders in the diagram as $S(d)$, for which there is a corresponding function in the form 
\begin{equation}
\eta:S(d) \rightarrow \mathbb{P}Z(d) - \{\varnothing\},
\end{equation}
where $\mathbb{P}$ is the powerset of a diagram's zones. Subsequently we use $\eta_{d}(s)$ to indicate the \textit{habitat} of a spider $s$ in a diagram. As found in our diagram given in Equation~(\ref{Ex1}), a single spider may have \textit{feet} that touch different zones, but not the same zone twice, for which we give the set of contours labels of $d$ as \{$A$, $B$\}. The set of zones in $d$ with labels in $\mathcal{L}$ is given through:
\begin{equation}
\mathcal{Z}=\{(a,b)\in \mathbb{F}\mathcal{L} \times \mathbb{F}\mathcal{L}: a \cap b = \varnothing\},
\end{equation}
which yields: 
\begin{equation}
Z(d)=\{(\varnothing,\{A,B\}),(\{A\},\{B\}), (\{B\}, \{A\}), (\{A,B\},\varnothing)\}.
\end{equation}
Zones can also be shaded as $Z^{\ast}(d)$, which in turn places an upper bound on set cardinality. As \citet[page 134]{GilHowseKent1999}
note, ``the semantics of a shaded zone is that the set it denotes may not contain elements other than those indicated by the spiders which touch that zone. A shaded zone which has no spiders is thus empty.'' Hence, in Equation~(\ref{Ex1}), the diagram indicates that $B - A$ (i.e., the complement of $A$ in $B$) is empty. To identify the spider in $d$, we use the spider identifier which yields: 
\begin{equation}
SI(d)=\{(1,\{(\{A\},\{B\}),(\varnothing,\{A,B\})\})\},
\end{equation}
which allows us to translate $d$ into the logical statement:
\begin{equation}
\left( \forall x (B(x) \Rightarrow A(x))\right) \wedge \left( \exists x( (A(x) \wedge \neg B(x)) \vee (\neg A(x) \wedge \neg B(x)))\right),
\end{equation}
whose semantics in natural language equates to the sentence ``$B$ is a subset of $A$, and there exists something outside $B$ (which is possibly inside $A$, or possibly outside $A$).'' 
\end{example}

\begin{remark}
Each spider diagram thus has an interpretation in the form
\begin{equation}
I=(U,\Phi),
\end{equation}
for which $U$ is the set of all elements of an interpretation (universal set) and $\Phi$ is the function that maps contour labels to subsets of $U$ as $\Phi(l)\subseteq U$. \citep[page 493]{UrbasJamnikStapleton2015} A zone in a spider diagram in the form $(in, out)$ thus represents the set:
\begin{equation}
\Phi(in, out)=\bigcap_{l \in in} \Phi(l) \cap \bigcap_{l \in out} (U - \Phi(l)),
\end{equation}
in which $\Phi(l)$ is the set assigned to the contour label $l$. This interpretation allows diagrammatic entailment to be understood as semantic entailment between truth conditions.
\end{remark}

\subsection{Application to the semiotic square}\label{semioticsquareapplication}

In order to ground our operational account of the semiotic square via the seme-wise constructors given in Equation~\ref{vanWoldeEq}, we first turn to a partial example square by \citet[page 5]{Greimas1988} that uses the seme pair /life/ vs. /death/ for his analysis of the short story ``Two friends'' (\textit{Deux amis}) by Guy de Maupassant (1850-1893).

\begin{example}\label{GreimasExample}
Consider the following semiotic square based on the contrary seme pair $s_{1}$ = /life/, $s_{2}$ = /death/:

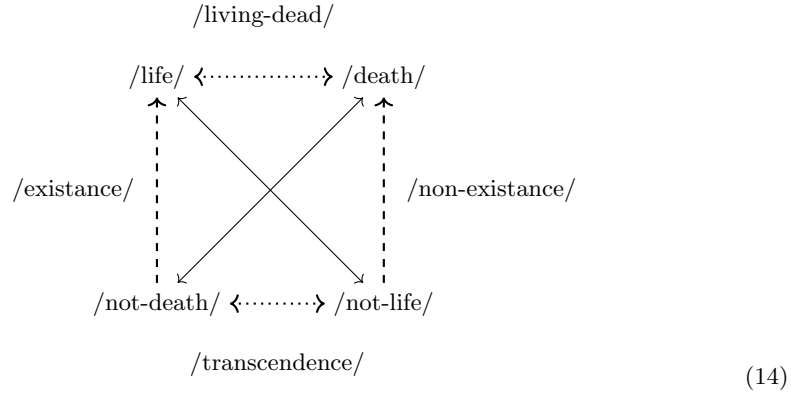
\begin{equation}
\begin{tikzpicture}[scale=0.75]
\node at (0,0)(s1){/life/};
\node at (4,0)(s2){/death/};
\node (not-s2) at (0,-4){/not-death/};
\node at (4,-4)(not-s1){/not-life/};
\draw[thick,dotted,<->] (s1) -- node[label={[label distance=4mm]90:/living-dead/}] {} (s2);
\draw[thick,dotted,<->]  (not-s1) -- node[label={[label distance=4mm]-90:/transcendence/}] {}(not-s2);
\draw[thick,dashed,->] (not-s2)--node[left,xshift=-5]{/existance/}(s1);
\draw[thick,dashed,->] (not-s1)--node[right,xshift=5]{/non-existance/}(s2);
\draw[<->](s1)--(not-s1);
\draw[<->](s2)--(not-s2);
\end{tikzpicture}
\end{equation}
Here, the meta-terms are constructed using the $`+'$ operator so that the hyponym (meta-term) at the $\mathbf{S}$ position, /living dead/ associates to figurative actors such as the Zombie or Vampire, while the meta-term at the $\mathbf{\bar{S}}$ position, /transcendence/, associates to figurative actors such as Spirits or Ghosts. At the Positive Dexies position is the meta-term /existence/ which carries the axiological value of `living' (i.e., /not-death/ + /life/) and the actualising modality of `being-able-to-do'. Duly, at the Negative Deixis position is the meta-term /non-existence/ (i.e., /not-life/ + /death/) containing the axiological value of `dying' which associates to the virtualising modality `not-being-able-to-do.' Each of the deixes is also associated to the thymic (classematic) categories of ``euphoria'' (Positive Deixis) and ``dysphoria'' (Negative Deixis), in turn valorising their diagrammatically implicated positive ($s_{1}$) or negative ($s_{2}$) terms. The neutral term $\mathbf{\bar{S}}$ is consequently classified as ``aphoric.'' \citep[page 346]{GreimasCourtes1982}
\end{example}

In terms of describing the semiotic square through unitary spider diagrams, the first question to be addressed concerns the existential status of the seme pair $(s_{1}$, $s_{2})$ as the complex axis $\mathfrak{c}$. In its simplest form, this would imply that the corresponding diagrams for each seme can be expressed as:
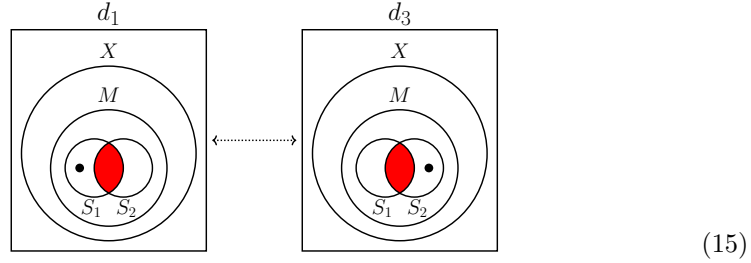
\begin{align}
\resizebox{0.5\textwidth}{!}{
\begin{tikzpicture}[line width=0.5mm,font=\Large]
\begin{scope}[local bounding box=d1]
	\scope[even odd rule]
	\clip \circleB;
	\clip \circleC;
	\fill[red] \circleB ;
	\endscope
\draw \circleB node[label={[label distance=8mm,xshift=-3,font=\huge]-90:$S_{1}$}] (S1){};
\draw \circleC node [label={[label distance=8mm,xshift=3,font=\huge]-90:$S_{2}$}] (S2){};
\draw (0.5,0) circle (2) node[label={[label distance=20mm,font=\huge]90:$M$}] (M) {};
\node[circle,minimum size=3mm,inner sep=0,fill=black] at (-0.5,0)(s1){};
\draw (0.5,0.5) circle (3) node[label={[label distance=30mm,font=\huge]90:$X$}]{};
\node[draw, rectangle, fit=(d1), inner sep=10pt, label={[font=\Huge]$d_{1}$}] (d1rec) {};
\end{scope}
\begin{scope}[local bounding box=d2,shift={(10,0)}] 
	\scope[even odd rule]
	\clip \circleB;
	\clip \circleC;
	\fill[red] \circleB ;
	\endscope
\draw \circleB node[label={[label distance=8mm,xshift=-3,font=\huge]-90:$S_{1}$}] (S1){};
\draw \circleC node [label={[label distance=8mm,xshift=3,font=\huge]-90:$S_{2}$}] (S2){};
\draw (0.5,0) circle (2) node[label={[label distance=20mm,font=\huge]90:$M$}] (M) {};
\node[circle,minimum size=3mm,inner sep=0,fill=black] at (1.5,0)(s2){};
\draw (0.5,0.5) circle (3) node[label={[label distance=30mm,font=\huge]90:$X$}]{};
\node[draw, rectangle, fit=(d2), inner sep=10pt, label={[font=\Huge]$d_{3}$}] (d2rec){};
\end{scope}
\draw[dotted,<->,shorten <=5pt, shorten >=5pt](d1rec)--(d2rec);
\end{tikzpicture}
}
\end{align}
Here, we have the left (e.g., /life/) and right (e.g., /death/) semes as the following spiders:
\begin{align}
SI(d_{1})&=\{(1,\{(\{S_{1}\},\{M\},\{X\}),(\{S_{2}\})\})\},\\
SI(d_{2})&=\{(1,\{(\{S_{2}\},\{M\},\{X\}),(\{S_{1}\})\})\}.
\end{align}

The assignment of $l \in L$ for both diagrams represents two distinct seme classes ($S_{1}, S_{2}$) that are respectively subsets of $M$ (the set of meta-terms), which in turn are subsets of $X$ (the inclusive syntagmatic set of terms or semantic universe). Spider diagrams allow us to indicate cardinality in a zone such that a spider in a non-shaded zone indicates the existence of at least one instance in that class. Hence in both $Z^{\ast}(d_{1})$ and $Z^{\ast}(d_{3})$ we use $S_{1}\cap S_{2}$ to indicate that the two semic instances are unique. By non-intersection, we mean here what Greimas alludes to in the notion of structuration as: ``the identification of the semic contents $s$ of a given inventory of occurrences, and that identification requires the `structuring' presence, that is to say, the disjoining, of the negative terms [$s_{i} \in S_{2}$] of the semic categories with the positive terms [$s_{i} \in S_{1}$] we are trying to identify.'' \citep[page 193]{Greimas1983} In terms of Example~\ref{GreimasExample} and its interpretative semantics, we assert that at the level of semes that /life/ is distinct and separate from /death/.

\begin{definition}[Diagrammatic contrariety]
Let $d_1 = \langle L, Z, Z^{\ast}, SI_1\rangle$ and $d_2 = \langle L, Z, Z^{\ast}, SI_2\rangle$ be unitary spider diagrams with
\[
L = \{S_1,S_2,M,X\},
\]
embedding the complex-axis seme pair $\mathfrak{c}$ within a semantic universe $X$, for which $\{S_{1},S_{2}\}\subseteq M \subseteq X$. We say that $d_1$ and $d_2$ are in the relation of \emph{diagrammatic contrariety}, written
\[
d_1 \;\#_{\mathrm{diag}}\; d_2,
\]
if the following conditions hold:
\begin{enumerate}
\item both diagrams contain exactly one spider,
\item the shaded zones coincide and enforce disjointness:
\[
Z^{\ast}(d_1) = Z^{\ast}(d_2) = \{(\{S_1,S_2\},\varnothing)\},
\]
so that $S_1 \cap S_2 = \varnothing$,
\item the spiders inhabit complementary single zones:
\[
\eta_{d_1}(s) = \{(\{S_1,M,X\},\{S_2\})\}, \qquad
\eta_{d_2}(s) = \{(\{S_2,M,X\},\{S_1\})\}.
\]
\end{enumerate}
\end{definition}

\begin{remark}
Diagrammatic contrariety expresses the mutual exclusion of $S_1$ and $S_2$ together with the independent actualisation of each seme in its respective diagram. This condition aligns to how \citet[page 89]{GreimasRastier1968} describe the relations found in $\mathfrak{c}$ as not only hyperonymic in regards to $\mathbf{S}$, but also presuppositional. This is expressed through the non-shading of the complementary habitat of $S_{i}$ instances compared to the intersection $S_{1}\cap S_{2}$. In contrast to diagrammatic negation (cf. Definition~\ref{Def:Negation}), which introduces indeterminacy, contrariety relates two determinate semantic configurations that are incompatible but not exhaustive.
\end{remark}

In order to represent the van Wolde composition of the square's elements and the transformation of $d_{1}$ into $d_{3}$ (e.g., as /life/ $\rightarrow$ /not-life/ $\rightarrow$ /death/) we provide the following two progressions that account for the existential states of semes given through $s_{1}\rightarrow \neg s_{1} \rightarrow s_{2}$, and $s_{2}\rightarrow \neg s_{2} \rightarrow s_{1}$:
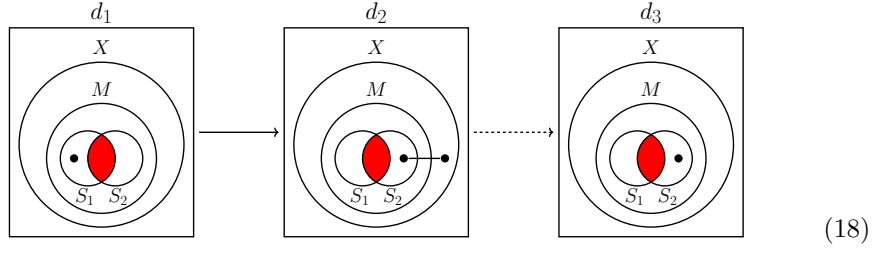
\begin{align}\label{d1d2d3}
\resizebox{0.75\textwidth}{!}{
\begin{tikzpicture}[line width=0.5mm,font=\Large]
\begin{scope}[local bounding box=d1]
	\scope[even odd rule]
	\clip \circleB;
	\clip \circleC;
	\fill[red] \circleB ;
	\endscope
\draw \circleB node[label={[label distance=8mm,xshift=-3,font=\huge]-90:$S_{1}$}] (S1){};
\draw \circleC node [label={[label distance=8mm,xshift=3,font=\huge]-90:$S_{2}$}] (S2){};
\draw (0.5,0) circle (2) node[label={[label distance=20mm,font=\huge]90:$M$}]{};
\node[circle,minimum size=3mm,inner sep=0,fill=black] at (-0.5,0)(s1){};
\draw (0.5,0.5) circle (3) node[label={[label distance=30mm,font=\huge]90:$X$}]{};
\node[draw, rectangle, fit=(d1), inner sep=10pt, label={[font=\Huge]$d_{1}$}] (d1rec) {};
\end{scope}
\begin{scope}[local bounding box=d2,shift={(10,0)}]
	\scope[even odd rule]
	\clip \circleB;
	\clip \circleC;
	\fill[red] \circleB ;
	\endscope
\draw \circleB node[label={[label distance=8mm,xshift=-3,font=\huge]-90:$S_{1}$}] (S1){};
\draw \circleC node [label={[label distance=8mm,xshift=3,font=\huge]-90:$S_{2}$}] (S2){};
\draw (0.5,0) circle (2) node[label={[label distance=20mm,font=\huge]90:$M$}]{};
\node[circle,minimum size=3mm,inner sep=0,fill=black] at (1.5,0)(s2){};
\node[circle,minimum size=3mm,inner sep=0,fill=black] at (3,0)(x1){};
\draw(s2)--(x1);
\draw (0.5,0.5) circle (3) node[label={[label distance=30mm,font=\huge]90:$X$}]{};
\node[draw, rectangle, fit=(d2), inner sep=10pt, label={[font=\Huge]$d_{2}$}] (d2rec)  {};
\end{scope}
\begin{scope}[local bounding box=d3,shift={(20,0)}]
	\scope[even odd rule]
	\clip \circleB;
	\clip \circleC;
	\fill[red] \circleB ;
	\endscope
\draw \circleB node[label={[label distance=8mm,xshift=-3,font=\huge]-90:$S_{1}$}] (S1){};
\draw \circleC node [label={[label distance=8mm,xshift=3,font=\huge]-90:$S_{2}$}] (S2){};
\draw (0.5,0) circle (2) node[label={[label distance=20mm,font=\huge]90:$M$}]{};
\node[circle,minimum size=3mm,inner sep=0,fill=black] at (1.5,0)(s2){};
\draw (0.5,0.5) circle (3) node[label={[label distance=30mm,font=\huge]90:$X$}]{};
\node[draw, rectangle, fit=(d3), inner sep=10pt, label={[font=\Huge]$d_{3}$}] (d3rec) {};
\end{scope}
\draw[->,shorten >=5pt,shorten <=5pt](d1rec)--(d2rec);
\draw[->,shorten >=5pt,shorten <=5pt,dashed](d2rec)--(d3rec);
\end{tikzpicture}
} 
\end{align}
\begin{align}\label{d3d4d1}
\resizebox{0.75\textwidth}{!}
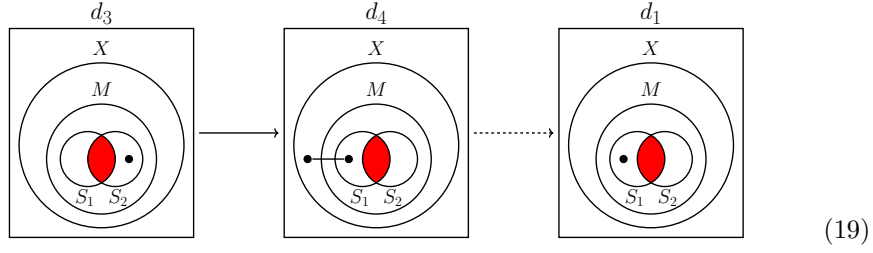
{
\begin{tikzpicture}[line width=0.5mm,font=\Large]
\begin{scope}[local bounding box=d1]
	\scope[even odd rule]
	\clip \circleB;
	\clip \circleC;
	\fill[red] \circleB ;
	\endscope
\draw \circleB node[label={[label distance=8mm,xshift=-3,font=\huge]-90:$S_{1}$}] (S1){};
\draw \circleC node [label={[label distance=8mm,xshift=3,font=\huge]-90:$S_{2}$}] (S2){};
\draw (0.5,0) circle (2) node[label={[label distance=20mm,font=\huge]90:$M$}]{};
\node[circle,minimum size=3mm,inner sep=0,fill=black] at (1.5,0)(s2){};
\draw (0.5,0.5) circle (3) node[label={[label distance=30mm,font=\huge]90:$X$}]{};
\node[draw, rectangle, fit=(d1), inner sep=10pt, label={[font=\Huge]$d_{3}$}] (d1rec) {};
\end{scope}
\begin{scope}[local bounding box=d2,shift={(10,0)}]
	\scope[even odd rule]
	\clip \circleB;
	\clip \circleC;
	\fill[red] \circleB ;
	\endscope
\draw \circleB node[label={[label distance=8mm,xshift=-3,font=\huge]-90:$S_{1}$}] (S1){};
\draw \circleC node [label={[label distance=8mm,xshift=3,font=\huge]-90:$S_{2}$}] (S2){};
\draw (0.5,0) circle (2) node[label={[label distance=20mm,font=\huge]90:$M$}]{};
\node[circle,minimum size=3mm,inner sep=0,fill=black] at (-0.5,0)(s1){};
\node[circle,minimum size=3mm,inner sep=0,fill=black] at (-2,0)(x1){};
\draw(s1)--(x1);
\draw (0.5,0.5) circle (3) node[label={[label distance=30mm,font=\huge]90:$X$}]{};
\node[draw, rectangle, fit=(d2), inner sep=10pt, label={[font=\Huge]$d_{4}$}] (d2rec)  {};
\end{scope}
\begin{scope}[local bounding box=d3,shift={(20,0)}]
	\scope[even odd rule]
	\clip \circleB;
	\clip \circleC;
	\fill[red] \circleB ;
	\endscope
\draw \circleB node[label={[label distance=8mm,xshift=-3,font=\huge]-90:$S_{1}$}] (S1){};
\draw \circleC node [label={[label distance=8mm,xshift=3,font=\huge]-90:$S_{2}$}] (S2){};
\draw (0.5,0) circle (2) node[label={[label distance=20mm,font=\huge]90:$M$}]{};
\node[circle,minimum size=3mm,inner sep=0,fill=black] at (-0.5,0)(s1){};
\draw (0.5,0.5) circle (3) node[label={[label distance=30mm,font=\huge]90:$X$}]{};
\node[draw, rectangle, fit=(d3), inner sep=10pt, label={[font=\Huge]$d_{1}$}] (d3rec) {};
\end{scope}
\draw[->,shorten >=5pt,shorten <=5pt](d1rec)--(d2rec);
\draw[->,shorten >=5pt,shorten <=5pt,dashed](d2rec)--(d3rec);
\end{tikzpicture}
} 
\end{align}

We provide arrows between diagrams following after Greimas (cf. Figure~(\ref{EQSemioticSquare})). We note here that in the case of $d_{2}$ and $d_{4}$ as the negations of the seme instances in $d_{1}$ and $d_{3}$ respectively, we approach the diagrammatic description by defining a spider habitat in $d_{2}$ as 
\begin{equation}
SI(d_{2})=\{(1,\{(\{S_{2},M,X\},\{S_{1}\}),(\{X\},\{M, S_{1},S_{2}\})\})\},
\end{equation}
and similarly for $d_{4}$ as
\begin{equation}
SI(d_{4})=\{(1,\{(\{S_{1},M,X\},\{S_{2}\}),(\{X\},\{M, S_{1},S_{2}\})\})\}.
\end{equation}
It is important to note that the diagrammatic representation of $\neg s_{1}$ given by $d_{2}$ does not correspond to Boolean negation in classical logic or that found in the square of opposition of Aristotle. \citep{RodriguesDias2023} Rather, it captures the existential realisation of terms outside $S_{1}$ under conditions of semantic indeterminacy. In the structural semiotic view, contradiction is not purely truth-functional, but relational and generative, and allows for intermediate or virtualised semantic states. In terms of the Greimas square given in Example~\ref{GreimasExample}, this provides an interpretive semantics in which $\neg s_{1}$ as /not-life/ is either a term found in the set $S_{2}$ or some other term within the semantic universe $X$. The OR condition virtualises the instance of $\neg s_{1}$ and thereby provides an elementary axiological structure on the term at that position in the square. \citep[page 21]{GreimasCourtes1982} This is also true of the dual case of $\neg s_{2}$ as /not-death/. 

\begin{definition}[Diagrammatic negation]\label{Def:Negation}
Let $d=\langle L,Z,Z^{\ast},SI\rangle$ be a unitary spider diagram with $L(d)=\{S_{1},S_{2},M,X\}$. For a designated seme $S_i\in L(d)$, with complementary index $j\neq i$ and $\{i,j\}=\{1,2\}$, the \emph{diagrammatic negation} of $S_i$, denoted $\neg_{\mathrm{diag}}(S_i)$, is a unitary spider diagram defined as follows: 
\begin{enumerate}
\item the intersection zone is shaded:
\[
Z^{\ast}(d)=\{(\{S_1,S_2\},\varnothing)\},
\]
so that $S_1\cap S_2=\varnothing$,
\item $d$ contains exactly one spider;
\item the habitat of that spider is the union of the two zones not contained in $S_i$, namely
\[
\eta(s)=\{(\{S_j,M,X\},\{S_i\}),(\{X\},\{M,S_i,S_j\})\};
\]
equivalently,
\[
SI(d)=\left\{\left(1,\{(\{S_j,M,X\},\{S_i\}),(\{X\},\{M,S_i,S_j\})\}\right)\right\}.
\]
\end{enumerate}
\end{definition}

\begin{proposition}[Non-Boolean character of diagrammatic negation]\label{Prop:GreimasNegation}
Let $L=\{S_1,S_2,M,X\}$ with $S_1,S_2\subseteq M\subseteq X$ and $S_1\cap S_2=\varnothing$. For $i,j\in\{1,2\}$ with $i\neq j$, the habitat of $\neg_{\mathrm{diag}}(S_i)$ is not the complement of $S_i$ in $X$. Rather,
\[
\eta(\neg_{\mathrm{diag}}(S_i))=(S_j\setminus S_i)\;\cup\;\bigl(X\setminus(M\cup S_1\cup S_2)\bigr).
\]
Consequently,
\[
\eta(\neg_{\mathrm{diag}}(S_i)) \subsetneq X\setminus S_i
\]
whenever \(M\setminus(S_1\cup S_2)\neq\varnothing\).
\end{proposition}

\begin{remark}
Proposition~\ref{Prop:GreimasNegation} formalises the non-Boolean character of negation in the semiotic square. In classical logic, negation is interpreted as complement in a universal domain, so that $\neg S_i$ corresponds to $X \setminus S_i$. By contrast, in structural semiotics, negation is defined relationally within a system of oppositions: the term $\neg s_i$ is not an arbitrary element outside $s_i$, but a counter-position determined by its relation to the contrary term $s_j$ and the structure of the semantic universe. The diagrammatic construction reflects this distinction. The habitat of $\neg_{\mathrm{diag}}(S_i)$ is restricted to zones that are admissible relative to the semic opposition $(S_1,S_2)$ and the meta-level domain $M$, and therefore excludes regions that would be included under a Boolean complement. In particular, elements in $M \setminus (S_1 \cup S_2)$ are not part of the negation, even though they lie outside $S_i$ in the global sense. Thus, diagrammatic negation captures a structured form of semantic exclusion, rather than a truth-functional operation on sets.
\end{remark}

As we see in the Greimas square from Example~\ref{GreimasExample}, the diagrammatic negation of $S_i$ (e.g., /life/) represents the existential actualisation of an instance lying outside $S_i$ (i.e., /not-life/, which depending on the textualization could be: /purgatory/, or /hell/, or /stasis/ etc.), while leaving its realisation underdetermined between the contrary seme $S_j$ and the exterior region outside both $S_1$, $S_2$, and $M$ (i.e., some other yet to be categorised term within the semantic universe). This hyperotactic condition may be expressed as:
\[
(S_1 \cap S_2 = \varnothing) \;\wedge\;\exists x \big(x \in (S_j \setminus S_i)\;\vee\; x \in (X \setminus (S_1 \cup S_2 \cup M))\big).
\]
Here, the disjunction reflects the union of admissible habitats of the spider, rather than a truth-functional negation. As such, we follow the observation by H\'{e}bert that in structural semiotics the operation of contradiction focusses on the establishment of ``comparative relations.'' \citep[page 32]{Hebert2020} Thus, in diagrams of type $\neg_{\mathrm{diag}}(S_i)$ such as in Example~\ref{GreimasExample}, the spider encodes /not-life/ (resp. /not-death/) as a virtualised semantic configuration in which there may exist many other terms of the same semic category. Hence, diagrammatic negation reinforces the hyperotactic relation among terms, and is local and contrastive.

\begin{definition}[Diagrammatic implication]\label{Def:Implication}
Let $d_1 = \langle L, Z, Z^{\ast}, SI_1\rangle$ and $d_2 = \langle L, Z, Z^{\ast}, SI_2\rangle$ be unitary spider diagrams with
\[
L = \{S_1,S_2,M,X\}.
\]
Let $S_i \in L$ be a designated seme with complementary index $j \neq i$, where $\{i,j\}=\{1,2\}$. We say that $d_2$ is the \emph{diagrammatic implication} of $d_1$, written
\[
d_1 \Rightarrow_{\mathrm{diag}} d_2,
\]
if the following conditions hold:
\begin{enumerate}
\item $d_1 = \neg_{\mathrm{diag}}(S_i)$ is the diagrammatic negation of $S_i$,
\item $d_2$ is a unitary diagram containing exactly one spider whose habitat is the single zone
\[
\eta(s) = \{(\{S_j\},\{S_i\})\},
\]
that is, the spider is fully contained in $S_j \setminus S_i$,
\item the shaded zones are preserved:
\[
Z^{\ast}(d_1) = Z^{\ast}(d_2) = \{(\{S_1,S_2\},\varnothing)\},
\]
\item the transformation from $d_1$ to $d_2$ is monotonic with respect to habitat inclusion, i.e.
\[
\eta_{d_2}(s) \subseteq \eta_{d_1}(s),
\]
with strict inclusion corresponding to a proper refinement.
\end{enumerate}
\end{definition}

\begin{remark}
We see that diagrammatic implication corresponds to the restriction of a spider's habitat from a set of zones to a proper subset, resolving semantic indeterminacy into determinate actualisation. Diagrammatic implication is thus dual to diagrammatic negation: whereas negation expands the habitat of a spider to introduce semantic indeterminacy, implication restricts the habitat to resolve that indeterminacy---i.e. the OR condition collapses. In terms of the square given in Example~\ref{GreimasExample}, given $d_{4}=\neg_{\mathrm{diag}}(S_{2})$, then:
\[
d_{4} \Rightarrow_{\mathrm{diag}} d_{1} \cong /\text{not-death}/ \rightarrow /\mathrm{life}/. 
\]
Within an interpretive semantics of this square, the implication between /not-death/ and /life/ establishes the axiological value of /living/ in the Positive Deixis: i.e., the indeterminate status of the term /not-death/ is resolved as /life/ via the thymic category of euphoria. In this sense, implication may be understood as the diagrammatic realisation of a previously virtualised semantic configuration. \citet[page 255]{GreimasCourtes1982} define a realization as an element of the 3-tuple virtualisation/actualisation/realisation for which a modality such as `being-able-to-do' or `knowing-how-to-do' transforms a disjunction into a conjunction through a performance. As \citet[page 186]{Hebert2020} suggests ``the shift from negative competence to positive competence is a movement from the nonexistence of an action to its potentiality (ontological status: `possible'), whereas performance entails moving from the possibility of an action to its realization (ontological status: ` real').''
\end{remark}

Our formalisations of diagrammatic negation ($\neg_{\mathrm{diag}}$) and diagrammatic implication ($\Rightarrow_{\mathrm{diag}}$) may also be understood as sequences of transformations generated by inference rules on spider diagrams. This means that the core relation of diagrammatic contrariety ($\#_{\mathrm{diag}}$) between instances in $S_{1}$ and $S_{2}$ becomes the static counterpart to the dynamic operations of Greimasian negation and implication. While the latter describe transformations between semantic configurations, contrariety specifies a structural relation between configurations in which each seme is actualised in isolation from the other. In the next section, we make the transformations of $\neg_{\mathrm{diag}}$ and $\Rightarrow_{\mathrm{diag}}$ precise by identifying the rules that witness each transition, thereby interpreting the arrows of the semiotic square as derivations in a diagrammatic proof system.

\section{Diagrammatic proof trees witnessing semic entailments}\label{prooftreesdiagramming}

Intuitively, Greimasian contradiction corresponds to inference rules that expand the habitat of spiders (introducing semantic indeterminacy), while Greimasian implication corresponds to rules that restrict or refine habitats (resolving indeterminacy). Thus, the diagrammatic progressions given in Equations~(\ref{d1d2d3})--(\ref{d3d4d1}) may be understood as derivations in a formal system of spider diagram transformations. In this section, we make this correspondence explicit by identifying a set of inference rules, denoted $\mathcal{R}$, that follow those of \citet{HowseStapletonTaylor2005,UrbasJamnikStapletonFlower2012} which witness each transition, possibly after a finite sequence of structural normalisation steps ensuring that the diagrams satisfy the requirements of the inference rules. In this way, the arrows of the semiotic square---namely `contradiction' (\raisebox{0pt}{\tikz{\draw[->] (0,0)--(0.6,0);}}) and `implication' (\raisebox{0pt}{\tikz{\draw[thick,dashed,->] (0,0)--(0.6,0);}})---are interpreted as composable derivations in a diagrammatic proof system.

\begin{definition}[$\alpha$-diagram]\label{def:AlphaDiagram}
A unitary spider diagram is an $\alpha$-diagram if and only if the habitat of every spider is a single zone.
\end{definition}

\begin{definition}[Combine diagrams]\label{def:Combine}
Let $d_{0}$ and $d_{1}$ be unitary $\alpha$-diagrams such that $Z(d_{0})=Z(d_{1})$ or $d_{0}=\bot$ or $d_{1}=\top$. We denote their combination as $d^{\ast}=d_{0} \ast d_{1}$, such that a unitary $\alpha$-diagram is defined as follows:
\begin{enumerate}
\item if $d_{0}=\bot$, or $d_{1}=\bot$, then $d^{\ast}=\bot$,
\item If there is a zone that is shaded in one diagram but contains more spiders in the other diagram then $d_{0} \wedge d_{1}$ is unsatisfiable and $d_{\ast}=\bot$. That is to say that if there exists $z \in Z(d_{i}), i = 0, 1$, such that $z \in Z^{\ast}(d_{j})$, and $S(\{z\},d_{i})-S(\{z\},d_{j})\neq \varnothing$, where $j=1-i$, then $d^{\ast}=\bot$,
\item Otherwise $d^{\ast}$ is a unitary $\alpha$-diagram such that:
	\begin{itemize}
	\item the set of zones of the combined diagram is the union of the zone sets of the original diagrams: $Z(d^{\ast})=Z(d_{0}\cup Z(d_{1}))$;
	\item shaded zones in the combined diagram are shaded in one (or both) of the original diagrams: $Z^{\ast}(d^{\ast})=Z^{\ast}(d_{0}\cup Z^{\ast}(d_{1}))$;
	\item the number of spiders in any zone in the combined diagram is the maximum number of spiders inhabiting that zone in the original diagrams: $\forall z \in Z(d^{\ast}) \bullet S(\{z\},d^{\ast})=S(\{z\},d_{0})\cup S(\{z\},d_{1})$.
	\end{itemize}
\end{enumerate}
\end{definition}

\begin{definition}[Conjunction elimination]
Let $d_0$ and $d_1$ be unitary spider diagrams such that their combination $d = d_0 \ast d_1$ is defined (cf. Definition~\ref{def:Combine}). Then from $d$ we may derive either $d_0$ or $d_1$. That is, the following inference is admissible:
\[
\frac{d_0 \ast d_1}{d_i}
\quad \text{for } i \in \{0,1\}.
\]
The derived diagram $d_i$ is obtained by restricting the spider set of $d$ to those spiders whose habitats originate in $d_i$, while preserving:
\begin{enumerate}
    \item the set of contours: $L(d_i) = L(d)$,
    \item the set of zones: $Z(d_i) = Z(d)$,
    \item the shading: $Z^*(d_i) = Z^*(d)$.
\end{enumerate}
\end{definition}

\begin{definition}[Split a spider]
Let $d$ be a unitary diagram and let $r$, $r_{1}$ and $r_{2}$ be regions of $d$ such that $r = r_{1} \cup r_{2}$ and $r_{1} \cap r_{2} = \varnothing$. Let $s_{n}(r)$ be a spider in $d$ (with habitat $r$). Let $d_{1}$ and $d_{2}$ be unitary diagrams such that:
\begin{enumerate}
\item $Z(d) = Z(d_{1})=Z(d_{2})$,
\item $Z^{\ast}(d) = Z^{\ast}(d_{1})=Z^{\ast}(d_{2})$,
\item there exists spider $s_{1} \in S(d_{1})$ and $s_{2} \in S(d_{2})$ such that $\eta(s_{1})=r_{1} \wedge \eta(s_{2})=r_{2}$, and $S(d)-\{s_{n}(r)\}=S(d_{1})=\{s_{1}\}=S(d_{2})-\{s_{2}\}$.
\end{enumerate}
Then we may replace $d$ with $d_{1} \vee d_{2}$ and vice versa.
\end{definition}

\begin{definition}[Add feet to a spider]
Let $d$ be a unitary diagram which contains a spider $s_{n}(r)$ whose habitat $r$ does not contain a zone $z$ of $d$. Let $d^{\prime}$ be a unitary diagram such that:
\begin{enumerate}
\item $Z(d) = Z(d^{\prime})$,
\item $Z^{\ast}(d) = Z^{\ast}(d^{\prime})$,
\item there exists spider $s^{\ast} \in S(d^{\prime})$ such that $\eta(s^{\ast})=r \cup \{z\}$ amd $S(d^{\prime})-\{s^{\ast}\}=S(d^{\prime})-\{s_{n}(r)\}$.
\end{enumerate}
Then $d$ can be replaced with $d^{\prime}$.
\end{definition}

\begin{definition}[Erase a spider]
Let $d$ be a unitary diagram such that there exists $(n, r) \in SI(d)$ with $r \cap Z^{\ast}(d) = \varnothing$. Let $d^{\prime})$ be a unitary diagram such that:
\begin{enumerate}
\item $Z(d) = Z(d^{\prime})$,
\item $Z^{\ast}(d) = Z^{\ast}(d^{\prime})$,
\item and $S(d)-\{s_{n}(r)\}=S(d^{\prime})$.
\end{enumerate}
Then $d$ can be replaced with $d^{\prime}$.
\end{definition}

\begin{definition}[Copy a spider]
Let $d_{1}$ and $d_{2}$ be unitary spider diagrams with corresponding regions, $r_{1}$ and $r_{2}$ such that:
\begin{enumerate}
\item $r_{1}$ contains no shaded zones
\item in $d_{1}$, all of the spiders that have a foot in $r_{1}$ are also found in $S(r_{1},d_{1})$,
\item there exists a habitat preserving injective but not surjective map $\xi$ from $S(r_{1},d_{1})$ to $S(r_{2},d_{2})$ such that $\eta_{d_{1}}(s_{1})$ corresponds to $\eta_{d_{2}}(\xi(s_{1}))$
\end{enumerate}
Choose a spider, $s$, that is in $S(r_2,d_2)$ but is not in the image of $\xi$ such that there exists a region, $r^{\prime}$ in $d_1$ that corresponds to $\eta_{d_{2}}(s)$. Let $d^{\prime}$ be a copy of $d_1$ except $d^{\prime}$ contains $s$ with habitat $r^{\prime}$. Then $d_1 \wedge d_{2}$ is logically equivalent to $d^{\prime}_{1}\wedge d_{2}$
\end{definition}

\begin{definition}[Idempotency of $\vee$]
Let $d$ be a spider diagram. We can replace $d$ by $d \vee d$ and vice versa.
\end{definition}

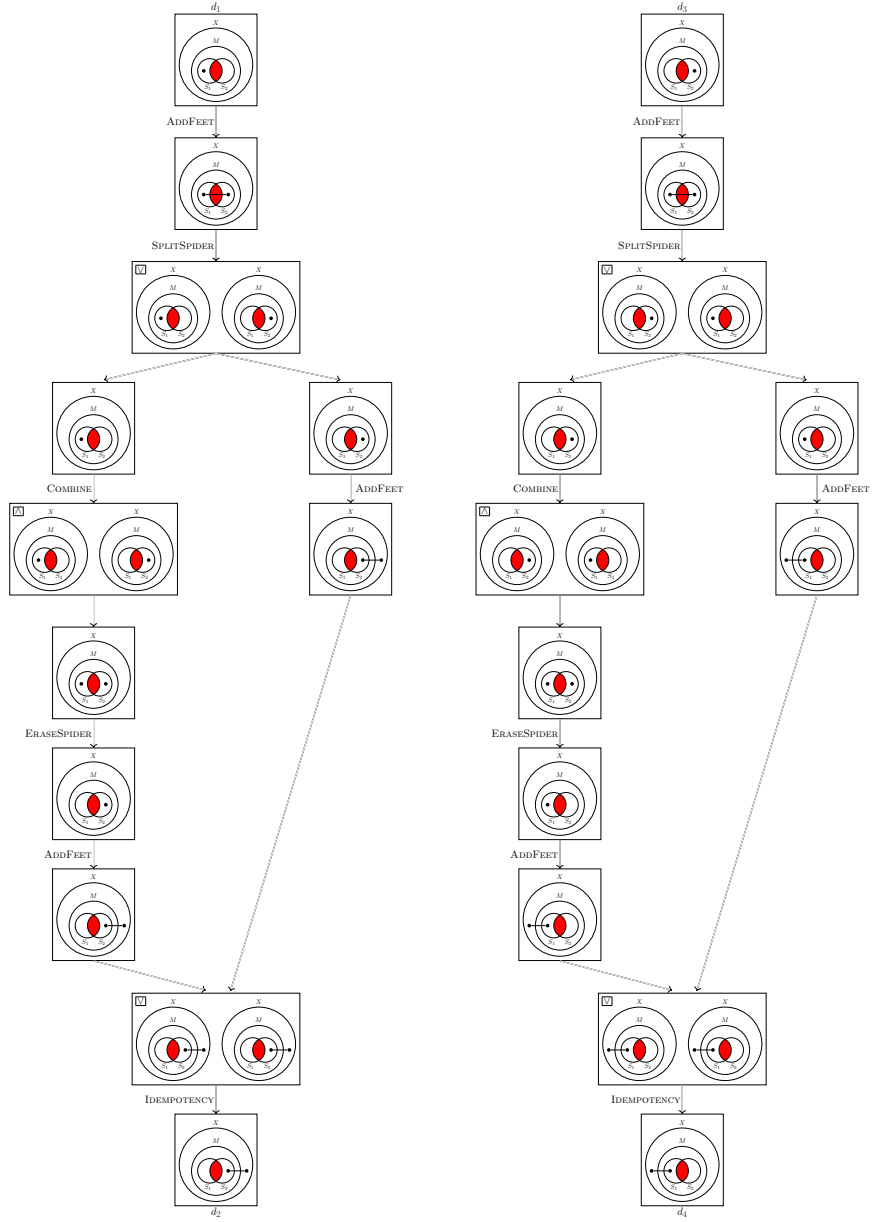
\begin{figure}
\centering
\subfloat[b][Proof tree $T_1$ deriving $d_2 = \neg_{\mathrm{diag}}(S_1)$ from $d_1$.\label{ProofTree1}]{%
\begin{adjustbox}{width=0.4\linewidth}
\begin{tikzpicture}[line width=0.5mm,font=\Large]
\def\circleB{(0,0) circle (1)}
\def\circleC{(1,0) circle (1)}
\def\circleA{(0.5,1) circle (1)}
\begin{scope}[local bounding box=d1]
	\scope[even odd rule]
	\clip \circleB;
	\clip \circleC;
	\fill[red] \circleB ;
	\endscope
\draw \circleB node[label={[label distance=8mm,xshift=-5]-90:$S_{1}$}] (S1){};
\draw \circleC node [label={[label distance=8mm,xshift=5]-90:$S_{2}$}] (S2){};
\draw (0.5,0) circle (2) node[label={[label distance=20mm]90:$M$}] (M) {};
\node[circle,minimum size=3mm,inner sep=0,fill=black] at (-0.5,0)(s1){};
\draw (0.5,0.5) circle (3) node[label={[label distance=30mm]90:$X$}]{};
\node[draw, rectangle, fit=(d1), inner sep=10pt, label={[font=\Huge]$d_{1}$}] (d1BOX) {};
\end{scope}
\begin{scope}[local bounding box=d2,shift={($(d1BOX) + (-0.5,-11)$)}]
	\scope[even odd rule]
	\clip \circleB;
	\clip \circleC;
	\fill[red] \circleB ;
	\endscope
\draw \circleB node[label={[label distance=8mm,xshift=-5]-90:$S_{1}$}] (S1){};
\draw \circleC node [label={[label distance=8mm,xshift=5]-90:$S_{2}$}] (S2){};
\draw (0.5,0) circle (2) node[label={[label distance=20mm]90:$M$}] (M) {};
\node[circle,minimum size=3mm,inner sep=0,fill=black] at (-0.5,0)(s1){};
\node[circle,minimum size=3mm,inner sep=0,fill=black] at (1.5,0)(s2){};
\draw (s1)--(s2);
\draw (0.5,0.5) circle (3) node[label={[label distance=30mm]90:$X$}]{};
\node[draw, rectangle, fit=(d2), inner sep=10pt] (d2BOX) {};
\end{scope}
\begin{scope}[local bounding box=d3,shift={($(d2BOX) + (-4,-11)$)}] 
	\scope[even odd rule]
	\clip \circleB;
	\clip \circleC;
	\fill[red] \circleB ;
	\endscope
\draw \circleB node[label={[label distance=8mm,xshift=-5]-90:$S_{1}$}] (S1){};
\draw \circleC node [label={[label distance=8mm,xshift=5]-90:$S_{2}$}] (S2){};
\draw (0.5,0) circle (2) node[label={[label distance=20mm]90:$M$}] (M) {};
\node[circle,minimum size=3mm,inner sep=0,fill=black] at (-0.5,0)(s1){};
\draw (0.5,0.5) circle (3) node[label={[label distance=30mm]90:$X$}]{};
\begin{scope}[xshift=70mm]
	\scope[even odd rule]
	\clip \circleB;
	\clip \circleC;
	\fill[red] \circleB ;
	\endscope
\draw \circleB node[label={[label distance=8mm,xshift=-5]-90:$S_{1}$}] (S1){};
\draw \circleC node [label={[label distance=8mm,xshift=5]-90:$S_{2}$}] (S2){};
\draw (0.5,0) circle (2) node[label={[label distance=20mm]90:$M$}] (M) {};
\node[circle,minimum size=3mm,inner sep=0,fill=black] at (1.5,0)(s2){};
\draw (0.5,0.5) circle (3) node[label={[label distance=30mm]90:$X$}]{};
\node[draw, rectangle, fit=(d3), inner sep=10pt] (d3BOX) {};
\node[anchor=north west,xshift=5,yshift=-5] at (d3BOX.north west) {\Huge$\boxed\vee$};
\end{scope}
\end{scope}
\begin{scope}[local bounding box=d4,shift={($(d3BOX.south) + (-10.5,-7)$)}]
	\scope[even odd rule]
	\clip \circleB;
	\clip \circleC;
	\fill[red] \circleB ;
	\endscope
\draw \circleB node[label={[label distance=8mm,xshift=-5]-90:$S_{1}$}] (S1){};
\draw \circleC node [label={[label distance=8mm,xshift=5]-90:$S_{2}$}] (S2){};
\draw (0.5,0) circle (2) node[label={[label distance=20mm]90:$M$}] (M) {};
\node[circle,minimum size=3mm,inner sep=0,fill=black] at (-0.5,0)(s1){};
\draw (0.5,0.5) circle (3) node[label={[label distance=30mm]90:$X$}]{};
\node[draw, rectangle, fit=(d4), inner sep=10pt] (d4BOX) {};
\end{scope}
\begin{scope}[local bounding box=d5,shift={($(d3BOX.south) + (10.5,-7)$)}]
	\scope[even odd rule]
	\clip \circleB;
	\clip \circleC;
	\fill[red] \circleB ;
	\endscope
\draw \circleB node[label={[label distance=8mm,xshift=-5]-90:$S_{1}$}] (S1){};
\draw \circleC node [label={[label distance=8mm,xshift=5]-90:$S_{2}$}] (S2){};
\draw (0.5,0) circle (2) node[label={[label distance=20mm]90:$M$}] (M) {};
\node[circle,minimum size=3mm,inner sep=0,fill=black] at (1.5,0)(s2){};
\draw (0.5,0.5) circle (3) node[label={[label distance=30mm]90:$X$}]{};
\node[draw, rectangle, fit=(d5), inner sep=10pt] (d5BOX) {};
\end{scope}
\begin{scope}[local bounding box=d6,shift={($(d4BOX.south) + (-4,-7)$)}]
	\scope[even odd rule]
	\clip \circleB;
	\clip \circleC;
	\fill[red] \circleB ;
	\endscope
\draw \circleB node[label={[label distance=8mm,xshift=-5]-90:$S_{1}$}] (S1){};
\draw \circleC node [label={[label distance=8mm,xshift=5]-90:$S_{2}$}] (S2){};
\draw (0.5,0) circle (2) node[label={[label distance=20mm]90:$M$}] (M) {};
\node[circle,minimum size=3mm,inner sep=0,fill=black] at (-0.5,0)(s1){};
\draw (0.5,0.5) circle (3) node[label={[label distance=30mm]90:$X$}]{};
\begin{scope}[xshift=70mm]
	\scope[even odd rule]
	\clip \circleB;
	\clip \circleC;
	\fill[red] \circleB ;
	\endscope
\draw \circleB node[label={[label distance=8mm,xshift=-5]-90:$S_{1}$}] (S1){};
\draw \circleC node [label={[label distance=8mm,xshift=5]-90:$S_{2}$}] (S2){};
\draw (0.5,0) circle (2) node[label={[label distance=20mm]90:$M$}] (M) {};
\node[circle,minimum size=3mm,inner sep=0,fill=black] at (1.5,0)(s2){};
\draw (0.5,0.5) circle (3) node[label={[label distance=30mm]90:$X$}]{};
\node[draw, rectangle, fit=(d6), inner sep=10pt] (d6BOX) {};
\node[anchor=north west,xshift=5,yshift=-5] at (d6BOX.north west) {\Huge$\boxed\wedge$};
\end{scope}
\end{scope}
\begin{scope}[local bounding box=d7,shift={($(d6BOX) + (-0.5,-11)$)}]
	\scope[even odd rule]
	\clip \circleB;
	\clip \circleC;
	\fill[red] \circleB ;
	\endscope
\draw \circleB node[label={[label distance=8mm,xshift=-5]-90:$S_{1}$}] (S1){};
\draw \circleC node [label={[label distance=8mm,xshift=5]-90:$S_{2}$}] (S2){};
\draw (0.5,0) circle (2) node[label={[label distance=20mm]90:$M$}] (M) {};
\node[circle,minimum size=3mm,inner sep=0,fill=black] at (1.5,0)(s2){};
\node[circle,minimum size=3mm,inner sep=0,fill=black] at (-0.5,0)(s1){};
\draw (0.5,0.5) circle (3) node[label={[label distance=30mm]90:$X$}]{};
\node[draw, rectangle, fit=(d7), inner sep=10pt] (d7BOX) {};
\end{scope}
\begin{scope}[local bounding box=d8,shift={($(d7BOX.south) + (-0.5,-7)$)}]
	\scope[even odd rule]
	\clip \circleB;
	\clip \circleC;
	\fill[red] \circleB ;
	\endscope
\draw \circleB node[label={[label distance=8mm,xshift=-5]-90:$S_{1}$}] (S1){};
\draw \circleC node [label={[label distance=8mm,xshift=5]-90:$S_{2}$}] (S2){};
\draw (0.5,0) circle (2) node[label={[label distance=20mm]90:$M$}] (M) {};
\node[circle,minimum size=3mm,inner sep=0,fill=black] at (1.5,0)(s2){};
\draw (0.5,0.5) circle (3) node[label={[label distance=30mm]90:$X$}]{};
\node[draw, rectangle, fit=(d8), inner sep=10pt] (d8BOX) {};
\end{scope}
\begin{scope}[local bounding box=d9,shift={($(d8BOX.south) + (-0.5,-7)$)}]
	\scope[even odd rule]
	\clip \circleB;
	\clip \circleC;
	\fill[red] \circleB ;
	\endscope
\draw \circleB node[label={[label distance=8mm,xshift=-5]-90:$S_{1}$}] (S1){};
\draw \circleC node [label={[label distance=8mm,xshift=5]-90:$S_{2}$}] (S2){};
\draw (0.5,0) circle (2) node[label={[label distance=20mm]90:$M$}] (M) {};
\node[circle,minimum size=3mm,inner sep=0,fill=black] at (1.5,0)(s2){};
\node[circle,minimum size=3mm,inner sep=0,fill=black] at (3,0)(x1){};
\draw(s2)--(x1);
\draw (0.5,0.5) circle (3) node[label={[label distance=30mm]90:$X$}]{};
\node[draw, rectangle, fit=(d9), inner sep=10pt] (d9BOX) {};
\end{scope}
\begin{scope}[local bounding box=d10,shift={($(d5BOX.south) + (-0.5,-7)$)}]
	\scope[even odd rule]
	\clip \circleB;
	\clip \circleC;
	\fill[red] \circleB ;
	\endscope
\draw \circleB node[label={[label distance=8mm,xshift=-5]-90:$S_{1}$}] (S1){};
\draw \circleC node [label={[label distance=8mm,xshift=5]-90:$S_{2}$}] (S2){};
\draw (0.5,0) circle (2) node[label={[label distance=20mm]90:$M$}] (M) {};
\node[circle,minimum size=3mm,inner sep=0,fill=black] at (1.5,0)(s2){};
\node[circle,minimum size=3mm,inner sep=0,fill=black] at (3,0)(x1){};
\draw(s2)--(x1);
\draw (0.5,0.5) circle (3) node[label={[label distance=30mm]90:$X$}]{};
\node[draw, rectangle, fit=(d10), inner sep=10pt] (d10BOX) {};
\end{scope}
\begin{scope}[local bounding box=d11,shift={($(d2BOX.south) + (3,-67)$)}] 
	\scope[even odd rule]
	\clip \circleB;
	\clip \circleC;
	\fill[red] \circleB ;
	\endscope
\draw \circleB node[label={[label distance=8mm,xshift=-5]-90:$S_{1}$}] (S1){};
\draw \circleC node [label={[label distance=8mm,xshift=5]-90:$S_{2}$}] (S2){};
\draw (0.5,0) circle (2) node[label={[label distance=20mm]90:$M$}] (M) {};
\node[circle,minimum size=3mm,inner sep=0,fill=black] at (1.5,0)(s2){};
\node[circle,minimum size=3mm,inner sep=0,fill=black] at (3,0)(x1){};
\draw(s2)--(x1);
\draw (0.5,0.5) circle (3) node[label={[label distance=30mm]90:$X$}]{};
\begin{scope}[xshift=-70mm]
	\scope[even odd rule]
	\clip \circleB;
	\clip \circleC;
	\fill[red] \circleB ;
	\endscope
\draw \circleB node[label={[label distance=8mm,xshift=-5]-90:$S_{1}$}] (S1){};
\draw \circleC node [label={[label distance=8mm,xshift=5]-90:$S_{2}$}] (S2){};
\draw (0.5,0) circle (2) node[label={[label distance=20mm]90:$M$}] (M) {};
\node[circle,minimum size=3mm,inner sep=0,fill=black] at (1.5,0)(s2){};
\node[circle,minimum size=3mm,inner sep=0,fill=black] at (3,0)(x1){};
\draw(s2)--(x1);
\draw (0.5,0.5) circle (3) node[label={[label distance=30mm]90:$X$}]{};
\node[draw, rectangle, fit=(d11), inner sep=10pt] (d11BOX) {};
\node[anchor=north west,xshift=5,yshift=-5] at (d11BOX.north west) {\Huge$\boxed\vee$};
\end{scope}
\end{scope}
\begin{scope}[local bounding box=d12,shift={($(d11BOX.south) + (-0.5,-7)$)}]
	\scope[even odd rule]
	\clip \circleB;
	\clip \circleC;
	\fill[red] \circleB ;
	\endscope
\draw \circleB node[label={[label distance=8mm,xshift=-5]-90:$S_{1}$}] (S1){};
\draw \circleC node [label={[label distance=8mm,xshift=5]-90:$S_{2}$}] (S2){};
\draw (0.5,0) circle (2) node[label={[label distance=20mm]90:$M$}] (M) {};
\node[circle,minimum size=3mm,inner sep=0,fill=black] at (1.5,0)(s2){};
\node[circle,minimum size=3mm,inner sep=0,fill=black] at (3,0)(x1){};
\draw(s2)--(x1);
\draw (0.5,0.5) circle (3) node[label={[label distance=30mm]90:$X$}]{};
\node[draw, rectangle, fit=(d12), inner sep=10pt,label={[font=\Huge]-90:$d_{2}$}] (d12BOX) {};
\end{scope}
\draw[thick,->,double equal sign distance](d1BOX.south)--node[right,label={[font=\Huge]180:\scshape{AddFeet}}]{}(d2BOX.north);
\draw[thick,->,double equal sign distance](d2BOX.south)--node[right,label={[font=\Huge]180:\scshape{SplitSpider}}]{}(d3BOX.north);
\draw[thick,->,double equal sign distance,shorten >=25pt](d3BOX.south)--(d4BOX.north);
\draw[thick,->,double equal sign distance,shorten >=25pt](d3BOX.south)--(d5BOX.north);
\draw[thick,->,double equal sign distance](d4BOX.south)--node[right,label={[font=\Huge]180:\scshape{Combine}}]{}(d6BOX.north);
\draw[thick,->,double equal sign distance](d6BOX.south)--(d7BOX.north);
\draw[thick,->,double equal sign distance](d7BOX.south)--node[right,label={[font=\Huge]180:\scshape{EraseSpider}}]{}(d8BOX.north);
\draw[thick,->,double equal sign distance](d8BOX.south)--node[right,label={[font=\Huge]180:\scshape{AddFeet}}]{}(d9BOX.north);
\draw[thick,->,double equal sign distance](d5BOX.south)--node[right,label={[font=\Huge]0:\scshape{AddFeet}}]{}(d10BOX.north);
\draw[thick,->,double equal sign distance,shorten >=5pt](d10BOX.south)--(d11BOX);
\draw[thick,->,double equal sign distance,shorten >=25pt](d9BOX.south)--(d11BOX.north);
\draw[thick,->,double equal sign distance](d11BOX.south)--node[right,label={[font=\Huge]180:\scshape{Idempotency}}]{}(d12BOX.north);
\end{tikzpicture}
\end{adjustbox}
}
\hspace{20pt}
\subfloat[b][Proof tree $T_2$ deriving $d_4 = \neg_{\mathrm{diag}}(S_2)$ from $d_3$.\label{ProofTree2}]{%
\begin{adjustbox}{width=0.4\linewidth}
\begin{tikzpicture}[line width=0.5mm,font=\Large]
\def\circleB{(0,0) circle (1)}
\def\circleC{(1,0) circle (1)}
\def\circleA{(0.5,1) circle (1)}
\begin{scope}[local bounding box=d1]
	\scope[even odd rule]
	\clip \circleB;
	\clip \circleC;
	\fill[red] \circleB ;
	\endscope
\draw \circleB node[label={[label distance=8mm,xshift=-5]-90:$S_{1}$}] (S1){};
\draw \circleC node [label={[label distance=8mm,xshift=5]-90:$S_{2}$}] (S2){};
\draw (0.5,0) circle (2) node[label={[label distance=20mm]90:$M$}] (M) {};
\node[circle,minimum size=3mm,inner sep=0,fill=black] at (1.5,0)(s1){};
\draw (0.5,0.5) circle (3) node[label={[label distance=30mm]90:$X$}]{};
\node[draw, rectangle, fit=(d1), inner sep=10pt, label={[font=\Huge]$d_{3}$}] (d1BOX) {};
\end{scope}
\begin{scope}[local bounding box=d2,shift={($(d1BOX) + (-0.5,-11)$)}]
	\scope[even odd rule]
	\clip \circleB;
	\clip \circleC;
	\fill[red] \circleB ;
	\endscope
\draw \circleB node[label={[label distance=8mm,xshift=-5]-90:$S_{1}$}] (S1){};
\draw \circleC node [label={[label distance=8mm,xshift=5]-90:$S_{2}$}] (S2){};
\draw (0.5,0) circle (2) node[label={[label distance=20mm]90:$M$}] (M) {};
\node[circle,minimum size=3mm,inner sep=0,fill=black] at (1.5,0)(s1){};
\node[circle,minimum size=3mm,inner sep=0,fill=black] at (-0.5,0)(s2){};
\draw (s1)--(s2);
\draw (0.5,0.5) circle (3) node[label={[label distance=30mm]90:$X$}]{};
\node[draw, rectangle, fit=(d2), inner sep=10pt] (d2BOX) {};
\end{scope}
\begin{scope}[local bounding box=d3,shift={($(d2BOX) + (-4,-11)$)}] 
	\scope[even odd rule]
	\clip \circleB;
	\clip \circleC;
	\fill[red] \circleB ;
	\endscope
\draw \circleB node[label={[label distance=8mm,xshift=-5]-90:$S_{1}$}] (S1){};
\draw \circleC node [label={[label distance=8mm,xshift=5]-90:$S_{2}$}] (S2){};
\draw (0.5,0) circle (2) node[label={[label distance=20mm]90:$M$}] (M) {};
\node[circle,minimum size=3mm,inner sep=0,fill=black] at (1.5,0)(s1){};
\draw (0.5,0.5) circle (3) node[label={[label distance=30mm]90:$X$}]{};
\begin{scope}[xshift=70mm]
	\scope[even odd rule]
	\clip \circleB;
	\clip \circleC;
	\fill[red] \circleB ;
	\endscope
\draw \circleB node[label={[label distance=8mm,xshift=-5]-90:$S_{1}$}] (S1){};
\draw \circleC node [label={[label distance=8mm,xshift=5]-90:$S_{2}$}] (S2){};
\draw (0.5,0) circle (2) node[label={[label distance=20mm]90:$M$}] (M) {};
\node[circle,minimum size=3mm,inner sep=0,fill=black] at (-0.5,0)(s2){};
\draw (0.5,0.5) circle (3) node[label={[label distance=30mm]90:$X$}]{};
\node[draw, rectangle, fit=(d3), inner sep=10pt] (d3BOX) {};
\node[anchor=north west,xshift=5,yshift=-5] at (d3BOX.north west) {\Huge$\boxed\vee$};
\end{scope}
\end{scope}
\begin{scope}[local bounding box=d4,shift={($(d3BOX.south) + (-10.5,-7)$)}]
	\scope[even odd rule]
	\clip \circleB;
	\clip \circleC;
	\fill[red] \circleB ;
	\endscope
\draw \circleB node[label={[label distance=8mm,xshift=-5]-90:$S_{1}$}] (S1){};
\draw \circleC node [label={[label distance=8mm,xshift=5]-90:$S_{2}$}] (S2){};
\draw (0.5,0) circle (2) node[label={[label distance=20mm]90:$M$}] (M) {};
\node[circle,minimum size=3mm,inner sep=0,fill=black] at (1.5,0)(s1){};
\draw (0.5,0.5) circle (3) node[label={[label distance=30mm]90:$X$}]{};
\node[draw, rectangle, fit=(d4), inner sep=10pt] (d4BOX) {};
\end{scope}
\begin{scope}[local bounding box=d5,shift={($(d3BOX.south) + (10.5,-7)$)}]
	\scope[even odd rule]
	\clip \circleB;
	\clip \circleC;
	\fill[red] \circleB ;
	\endscope
\draw \circleB node[label={[label distance=8mm,xshift=-5]-90:$S_{1}$}] (S1){};
\draw \circleC node [label={[label distance=8mm,xshift=5]-90:$S_{2}$}] (S2){};
\draw (0.5,0) circle (2) node[label={[label distance=20mm]90:$M$}] (M) {};
\node[circle,minimum size=3mm,inner sep=0,fill=black] at (-0.5,0)(s2){};
\draw (0.5,0.5) circle (3) node[label={[label distance=30mm]90:$X$}]{};
\node[draw, rectangle, fit=(d5), inner sep=10pt] (d5BOX) {};
\end{scope}
\begin{scope}[local bounding box=d6,shift={($(d4BOX.south) + (-4,-7)$)}]
	\scope[even odd rule]
	\clip \circleB;
	\clip \circleC;
	\fill[red] \circleB ;
	\endscope
\draw \circleB node[label={[label distance=8mm,xshift=-5]-90:$S_{1}$}] (S1){};
\draw \circleC node [label={[label distance=8mm,xshift=5]-90:$S_{2}$}] (S2){};
\draw (0.5,0) circle (2) node[label={[label distance=20mm]90:$M$}] (M) {};
\node[circle,minimum size=3mm,inner sep=0,fill=black] at (1.5,0)(s1){};
\draw (0.5,0.5) circle (3) node[label={[label distance=30mm]90:$X$}]{};
\begin{scope}[xshift=70mm]
	\scope[even odd rule]
	\clip \circleB;
	\clip \circleC;
	\fill[red] \circleB ;
	\endscope
\draw \circleB node[label={[label distance=8mm,xshift=-5]-90:$S_{1}$}] (S1){};
\draw \circleC node [label={[label distance=8mm,xshift=5]-90:$S_{2}$}] (S2){};
\draw (0.5,0) circle (2) node[label={[label distance=20mm]90:$M$}] (M) {};
\node[circle,minimum size=3mm,inner sep=0,fill=black] at (-0.5,0)(s2){};
\draw (0.5,0.5) circle (3) node[label={[label distance=30mm]90:$X$}]{};
\node[draw, rectangle, fit=(d6), inner sep=10pt] (d6BOX) {};
\node[anchor=north west,xshift=5,yshift=-5] at (d6BOX.north west) {\Huge$\boxed\wedge$};
\end{scope}
\end{scope}
\begin{scope}[local bounding box=d7,shift={($(d6BOX) + (-0.5,-11)$)}]
	\scope[even odd rule]
	\clip \circleB;
	\clip \circleC;
	\fill[red] \circleB ;
	\endscope
\draw \circleB node[label={[label distance=8mm,xshift=-5]-90:$S_{1}$}] (S1){};
\draw \circleC node [label={[label distance=8mm,xshift=5]-90:$S_{2}$}] (S2){};
\draw (0.5,0) circle (2) node[label={[label distance=20mm]90:$M$}] (M) {};
\node[circle,minimum size=3mm,inner sep=0,fill=black] at (-0.5,0)(s2){};
\node[circle,minimum size=3mm,inner sep=0,fill=black] at (1.5,0)(s1){};
\draw (0.5,0.5) circle (3) node[label={[label distance=30mm]90:$X$}]{};
\node[draw, rectangle, fit=(d7), inner sep=10pt] (d7BOX) {};
\end{scope}
\begin{scope}[local bounding box=d8,shift={($(d7BOX.south) + (-0.5,-7)$)}]
	\scope[even odd rule]
	\clip \circleB;
	\clip \circleC;
	\fill[red] \circleB ;
	\endscope
\draw \circleB node[label={[label distance=8mm,xshift=-5]-90:$S_{1}$}] (S1){};
\draw \circleC node [label={[label distance=8mm,xshift=5]-90:$S_{2}$}] (S2){};
\draw (0.5,0) circle (2) node[label={[label distance=20mm]90:$M$}] (M) {};
\node[circle,minimum size=3mm,inner sep=0,fill=black] at (-0.5,0)(s2){};
\draw (0.5,0.5) circle (3) node[label={[label distance=30mm]90:$X$}]{};
\node[draw, rectangle, fit=(d8), inner sep=10pt] (d8BOX) {};
\end{scope}
\begin{scope}[local bounding box=d9,shift={($(d8BOX.south) + (-0.5,-7)$)}]
	\scope[even odd rule]
	\clip \circleB;
	\clip \circleC;
	\fill[red] \circleB ;
	\endscope
\draw \circleB node[label={[label distance=8mm,xshift=-5]-90:$S_{1}$}] (S1){};
\draw \circleC node [label={[label distance=8mm,xshift=5]-90:$S_{2}$}] (S2){};
\draw (0.5,0) circle (2) node[label={[label distance=20mm]90:$M$}] (M) {};
\node[circle,minimum size=3mm,inner sep=0,fill=black] at (-0.5,0)(s2){};
\node[circle,minimum size=3mm,inner sep=0,fill=black] at (-2,0)(x1){};
\draw(s2)--(x1);
\draw (0.5,0.5) circle (3) node[label={[label distance=30mm]90:$X$}]{};
\node[draw, rectangle, fit=(d9), inner sep=10pt] (d9BOX) {};
\end{scope}
\begin{scope}[local bounding box=d10,shift={($(d5BOX.south) + (-0.5,-7)$)}]
	\scope[even odd rule]
	\clip \circleB;
	\clip \circleC;
	\fill[red] \circleB ;
	\endscope
\draw \circleB node[label={[label distance=8mm,xshift=-5]-90:$S_{1}$}] (S1){};
\draw \circleC node [label={[label distance=8mm,xshift=5]-90:$S_{2}$}] (S2){};
\draw (0.5,0) circle (2) node[label={[label distance=20mm]90:$M$}] (M) {};
\node[circle,minimum size=3mm,inner sep=0,fill=black] at (-0.5,0)(s2){};
\node[circle,minimum size=3mm,inner sep=0,fill=black] at (-2,0)(x1){};
\draw(s2)--(x1);
\draw (0.5,0.5) circle (3) node[label={[label distance=30mm]90:$X$}]{};
\node[draw, rectangle, fit=(d10), inner sep=10pt] (d10BOX) {};
\end{scope}
\begin{scope}[local bounding box=d11,shift={($(d2BOX.south) + (3,-67)$)}] 
	\scope[even odd rule]
	\clip \circleB;
	\clip \circleC;
	\fill[red] \circleB ;
	\endscope
\draw \circleB node[label={[label distance=8mm,xshift=-5]-90:$S_{1}$}] (S1){};
\draw \circleC node [label={[label distance=8mm,xshift=5]-90:$S_{2}$}] (S2){};
\draw (0.5,0) circle (2) node[label={[label distance=20mm]90:$M$}] (M) {};
\node[circle,minimum size=3mm,inner sep=0,fill=black] at (-0.5,0)(s2){};
\node[circle,minimum size=3mm,inner sep=0,fill=black] at (-2,0)(x1){};
\draw(s2)--(x1);
\draw (0.5,0.5) circle (3) node[label={[label distance=30mm]90:$X$}]{};
\begin{scope}[xshift=-70mm]
	\scope[even odd rule]
	\clip \circleB;
	\clip \circleC;
	\fill[red] \circleB ;
	\endscope
\draw \circleB node[label={[label distance=8mm,xshift=-5]-90:$S_{1}$}] (S1){};
\draw \circleC node [label={[label distance=8mm,xshift=5]-90:$S_{2}$}] (S2){};
\draw (0.5,0) circle (2) node[label={[label distance=20mm]90:$M$}] (M) {};
\node[circle,minimum size=3mm,inner sep=0,fill=black] at (-0.5,0)(s2){};
\node[circle,minimum size=3mm,inner sep=0,fill=black] at (-2,0)(x1){};
\draw(s2)--(x1);
\draw (0.5,0.5) circle (3) node[label={[label distance=30mm]90:$X$}]{};
\node[draw, rectangle, fit=(d11), inner sep=10pt] (d11BOX) {};
\node[anchor=north west,xshift=5,yshift=-5] at (d11BOX.north west) {\Huge$\boxed\vee$};
\end{scope}
\end{scope}
\begin{scope}[local bounding box=d12,shift={($(d11BOX.south) + (-0.5,-7)$)}]
	\scope[even odd rule]
	\clip \circleB;
	\clip \circleC;
	\fill[red] \circleB ;
	\endscope
\draw \circleB node[label={[label distance=8mm,xshift=-5]-90:$S_{1}$}] (S1){};
\draw \circleC node [label={[label distance=8mm,xshift=5]-90:$S_{2}$}] (S2){};
\draw (0.5,0) circle (2) node[label={[label distance=20mm]90:$M$}] (M) {};
\node[circle,minimum size=3mm,inner sep=0,fill=black] at (-0.5,0)(s2){};
\node[circle,minimum size=3mm,inner sep=0,fill=black] at (-2,0)(x1){};
\draw(s2)--(x1);
\draw (0.5,0.5) circle (3) node[label={[label distance=30mm]90:$X$}]{};
\node[draw, rectangle, fit=(d12), inner sep=10pt,label={[font=\Huge]-90:$d_{4}$}] (d12BOX) {};
\end{scope}
\draw[thick,->,double equal sign distance](d1BOX.south)--node[right,label={[font=\Huge]180:\scshape{AddFeet}}]{}(d2BOX.north);
\draw[thick,->,double equal sign distance](d2BOX.south)--node[right,label={[font=\Huge]180:\scshape{SplitSpider}}]{}(d3BOX.north);
\draw[thick,->,double equal sign distance,shorten >=25pt](d3BOX.south)--(d4BOX.north);
\draw[thick,->,double equal sign distance,shorten >=25pt](d3BOX.south)--(d5BOX.north);
\draw[thick,->,double equal sign distance](d4BOX.south)--node[right,label={[font=\Huge]180:\scshape{Combine}}]{}(d6BOX.north);
\draw[thick,->,double equal sign distance](d6BOX.south)--(d7BOX.north);
\draw[thick,->,double equal sign distance](d7BOX.south)--node[right,label={[font=\Huge]180:\scshape{EraseSpider}}]{}(d8BOX.north);
\draw[thick,->,double equal sign distance](d8BOX.south)--node[right,label={[font=\Huge]180:\scshape{AddFeet}}]{}(d9BOX.north);
\draw[thick,->,double equal sign distance](d5BOX.south)--node[right,label={[font=\Huge]0:\scshape{AddFeet}}]{}(d10BOX.north);
\draw[thick,->,double equal sign distance,shorten >=5pt](d10BOX.south)--(d11BOX);
\draw[thick,->,double equal sign distance,shorten >=25pt](d9BOX.south)--(d11BOX.north);
\draw[thick,->,double equal sign distance](d11BOX.south)--node[right,label={[font=\Huge]180:\scshape{Idempotency}}]{}(d12BOX.north);
\end{tikzpicture}
\end{adjustbox}
}
\caption{Proof trees witnessing diagrammatic negation on the complex axis $\mathfrak{c}$. Each tree corresponds to a derivation in the inference system $\mathcal{R}$, transforming an initial diagram representing an actualised seme into its diagrammatic negation.} 
\label{ProofTrees1-2}
\end{figure}

Consider the proof trees $T_{1}$ and $T_{2}$ we provide in Figure~\ref{ProofTrees1-2} that witness Greimasian negation as $\neg_{\mathrm{diag}}(S_{i})$. In Figure~\ref{ProofTree1} we derive $\neg s_{1}$ through $d_2 = \neg_{\mathrm{diag}}(S_1)$ from $d_1$ ($\cong s_{1}$) via a sequence of inference rules that includes: \textsc{AddFeet}, \textsc{SplitSpider}, \textsc{Combine}, \textsc{EraseSpider}, \textsc{Idempotency}. These inferences rules correspond to the expansion of the spider habitat from a determinate to an indeterminate configuration. Duly, in $T_2$ (Figure~\ref{ProofTree2}) we derive $\neg s_{2}$ through $d_4 = \neg_{\mathrm{diag}}(S_2)$ from $d_3$, in a symmetric manner as a diagrammatic negation on the second seme $s_{2}$ of the complex axis $\mathfrak{c}$. Each proof tree $T_i$ witnesses an instance of diagrammatic negation in the sense of Definition~\ref{Def:Negation}, for which we write 
\begin{equation}
T_1 : d_1 \vdash_{\mathcal{R}} d_2,
\end{equation}
to denote that $d_2$ is derivable from $d_1$ under the inference system $\mathcal{R}$, and duly,
\begin{equation}
T_2 : d_3 \vdash_{\mathcal{R}} d_4,
\end{equation}
to denote the derivation of $d_{4}$ from $d_{3}$ under $\mathcal{R}$.

\begin{figure}
\centering
\subfloat[b][Proof tree $T_3: d_5 \vdash_{\mathcal{R}} d_6$.\label{ProofTree3}]{%
\begin{adjustbox}{width=0.4\linewidth}
\begin{tikzpicture}[line width=0.5mm,font=\Large]
\def\circleB{(0,0) circle (1)}
\def\circleC{(1,0) circle (1)}
\def\circleA{(0.5,1) circle (1)}
%
\begin{scope}[local bounding box=d2,shift={($(d1BOX) + (-0.5,-11)$)}]
	\scope[even odd rule]
	\clip \circleB;
	\clip \circleC;
	\fill[red] \circleB ;
	\endscope
\draw \circleB node[label={[label distance=8mm,xshift=-5]-90:$S_{1}$}] (S1){};
\draw \circleC node [label={[label distance=8mm,xshift=5]-90:$S_{2}$}] (S2){};
\draw (0.5,0) circle (2) node[label={[label distance=20mm]90:$M$}] (M) {};
\node[circle,minimum size=3mm,inner sep=0,fill=black] at (3,0)(x1){};
\node[circle,minimum size=3mm,inner sep=0,fill=black] at (1.5,0)(s2){};
\draw (x1)--(s2);
\draw (0.5,0.5) circle (3) node[label={[label distance=30mm]90:$X$}]{};
\node[draw, rectangle, fit=(d2), inner sep=10pt,label={[font=\Huge]$d_{5}$}] (d2BOX) {};
\end{scope}
\begin{scope}[local bounding box=d3,shift={($(d2BOX) + (-4,-11)$)}] 
	\scope[even odd rule]
	\clip \circleB;
	\clip \circleC;
	\fill[red] \circleB ;
	\endscope
\draw \circleB node[label={[label distance=8mm,xshift=-5]-90:$S_{1}$}] (S1){};
\draw \circleC node [label={[label distance=8mm,xshift=5]-90:$S_{2}$}] (S2){};
\draw (0.5,0) circle (2) node[label={[label distance=20mm]90:$M$}] (M) {};
\node[circle,minimum size=3mm,inner sep=0,fill=black] at (1.5,0)(s1){};
\draw (0.5,0.5) circle (3) node[label={[label distance=30mm]90:$X$}]{};
\begin{scope}[xshift=70mm]
	\scope[even odd rule]
	\clip \circleB;
	\clip \circleC;
	\fill[red] \circleB ;
	\endscope
\draw \circleB node[label={[label distance=8mm,xshift=-5]-90:$S_{1}$}] (S1){};
\draw \circleC node [label={[label distance=8mm,xshift=5]-90:$S_{2}$}] (S2){};
\draw (0.5,0) circle (2) node[label={[label distance=20mm]90:$M$}] (M) {};
\node[circle,minimum size=3mm,inner sep=0,fill=black] at (3,0)(x1){};
\draw (0.5,0.5) circle (3) node[label={[label distance=30mm]90:$X$}]{};
\node[draw, rectangle, fit=(d3), inner sep=10pt] (d3BOX) {};
\node[anchor=north west,xshift=5,yshift=-5] at (d3BOX.north west) {\Huge$\boxed\vee$};
\end{scope}
\end{scope}
\begin{scope}[local bounding box=d4,shift={($(d3BOX.south) + (10.5,-7)$)}]
	\scope[even odd rule]
	\clip \circleB;
	\clip \circleC;
	\fill[red] \circleB ;
	\endscope
\draw \circleB node[label={[label distance=8mm,xshift=-5]-90:$S_{1}$}] (S1){};
\draw \circleC node [label={[label distance=8mm,xshift=5]-90:$S_{2}$}] (S2){};
\draw (0.5,0) circle (2) node[label={[label distance=20mm]90:$M$}] (M) {};
\node[circle,minimum size=3mm,inner sep=0,fill=black] at (3,0)(x1){};
\draw (0.5,0.5) circle (3) node[label={[label distance=30mm]90:$X$}]{};
\node[draw, rectangle, fit=(d4), inner sep=10pt] (d4BOX) {};
\end{scope}
\begin{scope}[local bounding box=d5,shift={($(d3BOX.south) + (-10.5,-7)$)}]
	\scope[even odd rule]
	\clip \circleB;
	\clip \circleC;
	\fill[red] \circleB ;
	\endscope
\draw \circleB node[label={[label distance=8mm,xshift=-5]-90:$S_{1}$}] (S1){};
\draw \circleC node [label={[label distance=8mm,xshift=5]-90:$S_{2}$}] (S2){};
\draw (0.5,0) circle (2) node[label={[label distance=20mm]90:$M$}] (M) {};
\node[circle,minimum size=3mm,inner sep=0,fill=black] at (1.5,0)(s1){};
\draw (0.5,0.5) circle (3) node[label={[label distance=30mm]90:$X$}]{};
\node[draw, rectangle, fit=(d5), inner sep=10pt] (d5BOX) {};
\end{scope}
\begin{scope}[local bounding box=d6,shift={($(d4BOX.south) + (-4,-7)$)}]
	\scope[even odd rule]
	\clip \circleB;
	\clip \circleC;
	\fill[red] \circleB ;
	\endscope
\draw \circleB node[label={[label distance=8mm,xshift=-5]-90:$S_{1}$}] (S1){};
\draw \circleC node [label={[label distance=8mm,xshift=5]-90:$S_{2}$}] (S2){};
\draw (0.5,0) circle (2) node[label={[label distance=20mm]90:$M$}] (M) {};
\node[circle,minimum size=3mm,inner sep=0,fill=black] at (1.5,0)(x1){};
\draw (0.5,0.5) circle (3) node[label={[label distance=30mm]90:$X$}]{};
\begin{scope}[xshift=70mm]
	\scope[even odd rule]
	\clip \circleB;
	\clip \circleC;
	\fill[red] \circleB ;
	\endscope
\draw \circleB node[label={[label distance=8mm,xshift=-5]-90:$S_{1}$}] (S1){};
\draw \circleC node [label={[label distance=8mm,xshift=5]-90:$S_{2}$}] (S2){};
\draw (0.5,0) circle (2) node[label={[label distance=20mm]90:$M$}] (M) {};
\node[circle,minimum size=3mm,inner sep=0,fill=black] at (3,0)(x1){};
\draw (0.5,0.5) circle (3) node[label={[label distance=30mm]90:$X$}]{};
\node[draw, rectangle, fit=(d6), inner sep=10pt] (d6BOX) {};
\node[anchor=north west,xshift=5,yshift=-5] at (d6BOX.north west) {\Huge$\boxed\wedge$};
\end{scope}
\end{scope}
\begin{scope}[local bounding box=d7,shift={($(d6BOX) + (-0.5,-11)$)}]
	\scope[even odd rule]
	\clip \circleB;
	\clip \circleC;
	\fill[red] \circleB ;
	\endscope
\draw \circleB node[label={[label distance=8mm,xshift=-5]-90:$S_{1}$}] (S1){};
\draw \circleC node [label={[label distance=8mm,xshift=5]-90:$S_{2}$}] (S2){};
\draw (0.5,0) circle (2) node[label={[label distance=20mm]90:$M$}] (M) {};
\node[circle,minimum size=3mm,inner sep=0,fill=black] at (3,0)(x1){};
\node[circle,minimum size=3mm,inner sep=0,fill=black] at (1.5,0)(s1){};
\draw (0.5,0.5) circle (3) node[label={[label distance=30mm]90:$X$}]{};
\node[draw, rectangle, fit=(d7), inner sep=10pt] (d7BOX) {};
\end{scope}
\begin{scope}[local bounding box=d8,shift={($(d7BOX.south) + (-0.5,-7)$)}]
	\scope[even odd rule]
	\clip \circleB;
	\clip \circleC;
	\fill[red] \circleB ;
	\endscope
\draw \circleB node[label={[label distance=8mm,xshift=-5]-90:$S_{1}$}] (S1){};
\draw \circleC node [label={[label distance=8mm,xshift=5]-90:$S_{2}$}] (S2){};
\draw (0.5,0) circle (2) node[label={[label distance=20mm]90:$M$}] (M) {};
\node[circle,minimum size=3mm,inner sep=0,fill=black] at (1.5,0)(s1){};
\draw (0.5,0.5) circle (3) node[label={[label distance=30mm]90:$X$}]{};
\node[draw, rectangle, fit=(d8), inner sep=10pt] (d8BOX) {};
\end{scope}
\begin{scope}[local bounding box=d11,shift={($(d2BOX) + (3,-61)$)}] 
	\scope[even odd rule]
	\clip \circleB;
	\clip \circleC;
	\fill[red] \circleB ;
	\endscope
\draw \circleB node[label={[label distance=8mm,xshift=-5]-90:$S_{1}$}] (S1){};
\draw \circleC node [label={[label distance=8mm,xshift=5]-90:$S_{2}$}] (S2){};
\draw (0.5,0) circle (2) node[label={[label distance=20mm]90:$M$}] (M) {};
\node[circle,minimum size=3mm,inner sep=0,fill=black] at (1.5,0)(s1){};
\draw (0.5,0.5) circle (3) node[label={[label distance=30mm]90:$X$}]{};
\begin{scope}[xshift=-70mm]
	\scope[even odd rule]
	\clip \circleB;
	\clip \circleC;
	\fill[red] \circleB ;
	\endscope
\draw \circleB node[label={[label distance=8mm,xshift=-5]-90:$S_{1}$}] (S1){};
\draw \circleC node [label={[label distance=8mm,xshift=5]-90:$S_{2}$}] (S2){};
\draw (0.5,0) circle (2) node[label={[label distance=20mm]90:$M$}] (M) {};
\node[circle,minimum size=3mm,inner sep=0,fill=black] at (1.5,0)(s1){};
\draw (0.5,0.5) circle (3) node[label={[label distance=30mm]90:$X$}]{};
\node[draw, rectangle, fit=(d11), inner sep=10pt] (d11BOX) {};
\node[anchor=north west,xshift=5,yshift=-5] at (d11BOX.north west) {\Huge$\boxed\vee$};
\end{scope}
\end{scope}
\begin{scope}[local bounding box=d12,shift={($(d11BOX.south) + (-0.5,-7)$)}]
	\scope[even odd rule]
	\clip \circleB;
	\clip \circleC;
	\fill[red] \circleB ;
	\endscope
\draw \circleB node[label={[label distance=8mm,xshift=-5]-90:$S_{1}$}] (S1){};
\draw \circleC node [label={[label distance=8mm,xshift=5]-90:$S_{2}$}] (S2){};
\draw (0.5,0) circle (2) node[label={[label distance=20mm]90:$M$}] (M) {};
\node[circle,minimum size=3mm,inner sep=0,fill=black] at (1.5,0)(s1){};
\draw (0.5,0.5) circle (3) node[label={[label distance=30mm]90:$X$}]{};
\node[draw, rectangle, fit=(d12), inner sep=10pt,label={[font=\Huge]-90:$d_{6}$}] (d12BOX) {};
\end{scope}
\draw[thick,->,double equal sign distance](d2BOX.south)--node[right,label={[font=\Huge]180:\scshape{SplitSpider}}]{}(d3BOX.north);
\draw[thick,->,double equal sign distance,shorten >=25pt](d3BOX.south)--(d4BOX.north);
\draw[thick,->,double equal sign distance,shorten >=25pt](d3BOX.south)--(d5BOX.north);
\draw[thick,->,double equal sign distance](d4BOX.south)--node[right,label={[font=\Huge]0:\scshape{Combine}}]{}(d6BOX.north);
\draw[thick,->,double equal sign distance](d6BOX.south)--(d7BOX.north);
\draw[thick,->,double equal sign distance](d7BOX.south)--node[right,label={[font=\Huge]0:\scshape{EraseSpider}}]{}(d8BOX.north);
\draw[thick,->,double equal sign distance](d5BOX.south)--(d11BOX.north);
\draw[thick,->,double equal sign distance,shorten >=25pt](d8BOX.south)--(d11BOX.80);
\draw[thick,->,double equal sign distance](d11BOX.south)--node[right,label={[font=\Huge]180:\scshape{Idempotency}}]{}(d12BOX.north);
\end{tikzpicture}
\end{adjustbox}
}
\hspace{55pt}
\subfloat[b][Proof tree $T_4: d_7 \vdash_{\mathcal{R}} d_8$.\label{ProofTree4}]{%
\begin{adjustbox}{width=0.4\linewidth}
\begin{tikzpicture}[line width=0.5mm,font=\Large]
\def\circleB{(0,0) circle (1)}
\def\circleC{(1,0) circle (1)}
\def\circleA{(0.5,1) circle (1)}
%
\begin{scope}[local bounding box=d2,shift={($(d1BOX) + (-0.5,-11)$)}]
	\scope[even odd rule]
	\clip \circleB;
	\clip \circleC;
	\fill[red] \circleB ;
	\endscope
\draw \circleB node[label={[label distance=8mm,xshift=-5]-90:$S_{1}$}] (S1){};
\draw \circleC node [label={[label distance=8mm,xshift=5]-90:$S_{2}$}] (S2){};
\draw (0.5,0) circle (2) node[label={[label distance=20mm]90:$M$}] (M) {};
\node[circle,minimum size=3mm,inner sep=0,fill=black] at (-2,0)(x1){};
\node[circle,minimum size=3mm,inner sep=0,fill=black] at (-0.5,0)(s2){};
\draw (x1)--(s2);
\draw (0.5,0.5) circle (3) node[label={[label distance=30mm]90:$X$}]{};
\node[draw, rectangle, fit=(d2), inner sep=10pt,label={[font=\Huge]$d_{7}$}] (d2BOX) {};
\end{scope}
\begin{scope}[local bounding box=d3,shift={($(d2BOX) + (-4,-11)$)}] 
	\scope[even odd rule]
	\clip \circleB;
	\clip \circleC;
	\fill[red] \circleB ;
	\endscope
\draw \circleB node[label={[label distance=8mm,xshift=-5]-90:$S_{1}$}] (S1){};
\draw \circleC node [label={[label distance=8mm,xshift=5]-90:$S_{2}$}] (S2){};
\draw (0.5,0) circle (2) node[label={[label distance=20mm]90:$M$}] (M) {};
\node[circle,minimum size=3mm,inner sep=0,fill=black] at (-0.5,0)(s1){};
\draw (0.5,0.5) circle (3) node[label={[label distance=30mm]90:$X$}]{};
\begin{scope}[xshift=70mm]
	\scope[even odd rule]
	\clip \circleB;
	\clip \circleC;
	\fill[red] \circleB ;
	\endscope
\draw \circleB node[label={[label distance=8mm,xshift=-5]-90:$S_{1}$}] (S1){};
\draw \circleC node [label={[label distance=8mm,xshift=5]-90:$S_{2}$}] (S2){};
\draw (0.5,0) circle (2) node[label={[label distance=20mm]90:$M$}] (M) {};
\node[circle,minimum size=3mm,inner sep=0,fill=black] at (-2,0)(x1){};
\draw (0.5,0.5) circle (3) node[label={[label distance=30mm]90:$X$}]{};
\node[draw, rectangle, fit=(d3), inner sep=10pt] (d3BOX) {};
\node[anchor=north west,xshift=5,yshift=-5] at (d3BOX.north west) {\Huge$\boxed\vee$};
\end{scope}
\end{scope}
\begin{scope}[local bounding box=d4,shift={($(d3BOX.south) + (10.5,-7)$)}]
	\scope[even odd rule]
	\clip \circleB;
	\clip \circleC;
	\fill[red] \circleB ;
	\endscope
\draw \circleB node[label={[label distance=8mm,xshift=-5]-90:$S_{1}$}] (S1){};
\draw \circleC node [label={[label distance=8mm,xshift=5]-90:$S_{2}$}] (S2){};
\draw (0.5,0) circle (2) node[label={[label distance=20mm]90:$M$}] (M) {};
\node[circle,minimum size=3mm,inner sep=0,fill=black] at (-2,0)(x1){};
\draw (0.5,0.5) circle (3) node[label={[label distance=30mm]90:$X$}]{};
\node[draw, rectangle, fit=(d4), inner sep=10pt] (d4BOX) {};
\end{scope}
\begin{scope}[local bounding box=d5,shift={($(d3BOX.south) + (-10.5,-7)$)}]
	\scope[even odd rule]
	\clip \circleB;
	\clip \circleC;
	\fill[red] \circleB ;
	\endscope
\draw \circleB node[label={[label distance=8mm,xshift=-5]-90:$S_{1}$}] (S1){};
\draw \circleC node [label={[label distance=8mm,xshift=5]-90:$S_{2}$}] (S2){};
\draw (0.5,0) circle (2) node[label={[label distance=20mm]90:$M$}] (M) {};
\node[circle,minimum size=3mm,inner sep=0,fill=black] at (-0.5,0)(s1){};
\draw (0.5,0.5) circle (3) node[label={[label distance=30mm]90:$X$}]{};
\node[draw, rectangle, fit=(d5), inner sep=10pt] (d5BOX) {};
\end{scope}
\begin{scope}[local bounding box=d6,shift={($(d4BOX.south) + (-4,-7)$)}]
	\scope[even odd rule]
	\clip \circleB;
	\clip \circleC;
	\fill[red] \circleB ;
	\endscope
\draw \circleB node[label={[label distance=8mm,xshift=-5]-90:$S_{1}$}] (S1){};
\draw \circleC node [label={[label distance=8mm,xshift=5]-90:$S_{2}$}] (S2){};
\draw (0.5,0) circle (2) node[label={[label distance=20mm]90:$M$}] (M) {};
\node[circle,minimum size=3mm,inner sep=0,fill=black] at (-0.5,0)(s1){};
\draw (0.5,0.5) circle (3) node[label={[label distance=30mm]90:$X$}]{};
\begin{scope}[xshift=70mm]
	\scope[even odd rule]
	\clip \circleB;
	\clip \circleC;
	\fill[red] \circleB ;
	\endscope
\draw \circleB node[label={[label distance=8mm,xshift=-5]-90:$S_{1}$}] (S1){};
\draw \circleC node [label={[label distance=8mm,xshift=5]-90:$S_{2}$}] (S2){};
\draw (0.5,0) circle (2) node[label={[label distance=20mm]90:$M$}] (M) {};
\node[circle,minimum size=3mm,inner sep=0,fill=black] at (-2,0)(x1){};
\draw (0.5,0.5) circle (3) node[label={[label distance=30mm]90:$X$}]{};
\node[draw, rectangle, fit=(d6), inner sep=10pt] (d6BOX) {};
\node[anchor=north west,xshift=5,yshift=-5] at (d6BOX.north west) {\Huge$\boxed\wedge$};
\end{scope}
\end{scope}
\begin{scope}[local bounding box=d7,shift={($(d6BOX) + (-0.5,-11)$)}]
	\scope[even odd rule]
	\clip \circleB;
	\clip \circleC;
	\fill[red] \circleB ;
	\endscope
\draw \circleB node[label={[label distance=8mm,xshift=-5]-90:$S_{1}$}] (S1){};
\draw \circleC node [label={[label distance=8mm,xshift=5]-90:$S_{2}$}] (S2){};
\draw (0.5,0) circle (2) node[label={[label distance=20mm]90:$M$}] (M) {};
\node[circle,minimum size=3mm,inner sep=0,fill=black] at (-2,0)(x1){};
\node[circle,minimum size=3mm,inner sep=0,fill=black] at (-0.5,0)(s1){};
\draw (0.5,0.5) circle (3) node[label={[label distance=30mm]90:$X$}]{};
\node[draw, rectangle, fit=(d7), inner sep=10pt] (d7BOX) {};
\end{scope}
\begin{scope}[local bounding box=d8,shift={($(d7BOX.south) + (-0.5,-7)$)}]
	\scope[even odd rule]
	\clip \circleB;
	\clip \circleC;
	\fill[red] \circleB ;
	\endscope
\draw \circleB node[label={[label distance=8mm,xshift=-5]-90:$S_{1}$}] (S1){};
\draw \circleC node [label={[label distance=8mm,xshift=5]-90:$S_{2}$}] (S2){};
\draw (0.5,0) circle (2) node[label={[label distance=20mm]90:$M$}] (M) {};
\node[circle,minimum size=3mm,inner sep=0,fill=black] at (-0.5,0)(s1){};
\draw (0.5,0.5) circle (3) node[label={[label distance=30mm]90:$X$}]{};
\node[draw, rectangle, fit=(d8), inner sep=10pt] (d8BOX) {};
\end{scope}
\begin{scope}[local bounding box=d11,shift={($(d2BOX) + (3,-61)$)}] 
	\scope[even odd rule]
	\clip \circleB;
	\clip \circleC;
	\fill[red] \circleB ;
	\endscope
\draw \circleB node[label={[label distance=8mm,xshift=-5]-90:$S_{1}$}] (S1){};
\draw \circleC node [label={[label distance=8mm,xshift=5]-90:$S_{2}$}] (S2){};
\draw (0.5,0) circle (2) node[label={[label distance=20mm]90:$M$}] (M) {};
\node[circle,minimum size=3mm,inner sep=0,fill=black] at (-0.5,0)(s1){};
\draw (0.5,0.5) circle (3) node[label={[label distance=30mm]90:$X$}]{};
\begin{scope}[xshift=-70mm]
	\scope[even odd rule]
	\clip \circleB;
	\clip \circleC;
	\fill[red] \circleB ;
	\endscope
\draw \circleB node[label={[label distance=8mm,xshift=-5]-90:$S_{1}$}] (S1){};
\draw \circleC node [label={[label distance=8mm,xshift=5]-90:$S_{2}$}] (S2){};
\draw (0.5,0) circle (2) node[label={[label distance=20mm]90:$M$}] (M) {};
\node[circle,minimum size=3mm,inner sep=0,fill=black] at (-0.5,0)(s1){};
\draw (0.5,0.5) circle (3) node[label={[label distance=30mm]90:$X$}]{};
\node[draw, rectangle, fit=(d11), inner sep=10pt] (d11BOX) {};
\node[anchor=north west,xshift=5,yshift=-5] at (d11BOX.north west) {\Huge$\boxed\vee$};
\end{scope}
\end{scope}
\begin{scope}[local bounding box=d12,shift={($(d11BOX.south) + (-0.5,-7)$)}]
	\scope[even odd rule]
	\clip \circleB;
	\clip \circleC;
	\fill[red] \circleB ;
	\endscope
\draw \circleB node[label={[label distance=8mm,xshift=-5]-90:$S_{1}$}] (S1){};
\draw \circleC node [label={[label distance=8mm,xshift=5]-90:$S_{2}$}] (S2){};
\draw (0.5,0) circle (2) node[label={[label distance=20mm]90:$M$}] (M) {};
\node[circle,minimum size=3mm,inner sep=0,fill=black] at (-0.5,0)(s1){};
\draw (0.5,0.5) circle (3) node[label={[label distance=30mm]90:$X$}]{};
\node[draw, rectangle, fit=(d12), inner sep=10pt,label={[font=\Huge]-90:$d_{8}$}] (d12BOX) {};
\end{scope}
\draw[thick,->,double equal sign distance](d2BOX.south)--node[right,label={[font=\Huge]180:\scshape{SplitSpider}}]{}(d3BOX.north);
\draw[thick,->,double equal sign distance,shorten >=25pt](d3BOX.south)--(d4BOX.north);
\draw[thick,->,double equal sign distance,shorten >=25pt](d3BOX.south)--(d5BOX.north);
\draw[thick,->,double equal sign distance](d4BOX.south)--node[right,label={[font=\Huge]0:\scshape{Combine}}]{}(d6BOX.north);
\draw[thick,->,double equal sign distance](d6BOX.south)--(d7BOX.north);
\draw[thick,->,double equal sign distance](d7BOX.south)--node[right,label={[font=\Huge]0:\scshape{EraseSpider}}]{}(d8BOX.north);
\draw[thick,->,double equal sign distance](d5BOX.south)--(d11BOX.north);
\draw[thick,->,double equal sign distance,shorten >=25pt](d8BOX.south)--(d11BOX.80);
\draw[thick,->,double equal sign distance](d11BOX.south)--node[right,label={[font=\Huge]180:\scshape{Idempotency}}]{}(d12BOX.north);
\end{tikzpicture}
\end{adjustbox}
}
\caption{Proof trees witnessing diagrammatic implication. Each tree corresponds to a derivation in the inference system $\mathcal{R}$, transforming an initial diagram representing a virtualised seme into its diagrammatic implication.} 
\label{ProofTrees3-4}
\end{figure}
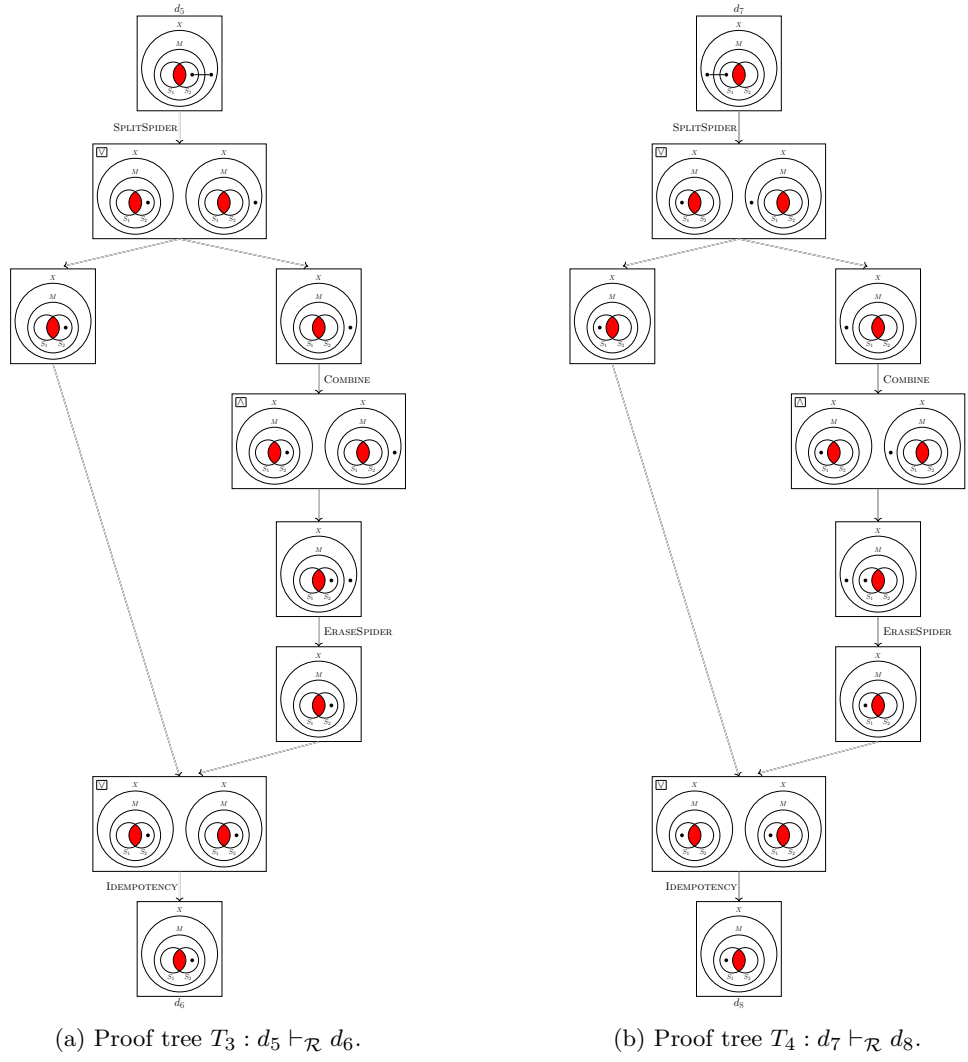

In Figure~\ref{ProofTrees3-4} we use the inference system $\mathcal{R}$ to transform the diagrammatic negations of the elements of the negative axis as $\mathfrak{n}=(\neg s_{1}, \neg s_{2})$ into the realised seme pair of the complex axis $\mathfrak{c}=(s_{1}, s_{2})$. We write
\begin{align}
T_{3}: d_{5} \vdash_{\mathcal{R}} d_{6},\\
T_{4}: d_{7} \vdash_{\mathcal{R}} d_{8},
\end{align}
to denote the derivation of $d_{6}$ from $d_{5}$ through the diagrammatic implication $d_{5}\Rightarrow_{\mathrm{diag}} d_{6}$ for which $d_{5}=\neg_{\mathrm{diag}}(S_{1})$. Here, we distinguish between the semantic relation $d_1 \Rightarrow_{\mathrm{diag}} d_2$ and its proof-theoretic realisation $d_1 \vdash_{\mathcal{R}} d_2$, where the latter witnesses the former via a derivation in $\mathcal{R}$. This is dual to the case of $d_{7}$ and $d_{8}$. 

Within an interpretive semantics of the Greimas square we gave in Example~\ref{GreimasExample}, the proof trees 
\[
T_{1}\cong \text{/life/}\rightarrow\text{/not-life/}, \quad T_{2}\cong \text{/death/}\rightarrow\text{/not-death/},
\]
and
\[
T_{3}\cong \text{/not-life/}\rightarrow\text{/death/}, \quad T_{4}\cong \text{/not-death/}\rightarrow\text{/life/},
\]
provide the basis for syntagmatic units to emerge at the level of the meta-terms. Here, the process of transformation of states of the initial seme pair associate to a figurativization \citep[page 119]{GreimasCourtes1982} in Greimas's source text in which objects and actants generate narrative trajectories according to conjunctions and disjunctions. Thus, like the steps taken in each application of inference rules to derive a proof-theoretic semantics, Greimas similarly argues that the semiotic square is a descriptor of an

\begin{quote}
autonomy of two distinct levels of semiotic representation, a logico-semantic level where logical operations which can account for the manipulation of contents are found, and a discoursive level where these same logical operations, once they are converted, can receive actantial formulations as part of a surface narrative grammar and processual and aspectual representations as part of the semantic level. \citep[page 6]{Greimas1988}
\end{quote}
 
But we also note here that the proof trees are not merely illustrative of a logico-semantic structure, but witness the definitional clauses of diagrammatic negation and implication given in Definitions  \ref{Def:Negation} and \ref{Def:Implication}. In particular, the inference rules in $\mathcal{R}$ are sufficient to generate the required habitat expansions and refinements that characterise these operations. Thus, diagrammatic negation and implication admit a constructive realisation within the proof system, with each instance corresponding to a derivation in $\mathcal{R}$.

\subsection{Deriving the meta-terms}\label{derivingmetaterms}

As we discussed earlier, in the traditional formulation of the semiotic square, the construction of meta-terms as sememes is given through 
\begin{align*}
\mathbf{S} &= s_1 + s_2,\\
\mathbf{\bar{S}}&=\neg s_{1} + \neg s_{2},\\
\mathrm{Positive \,Deixis} &= s_{1} + \neg s_{2},\\
\mathrm{Negative \,Deixis} &= s_{2} + \neg s_{1},
\end{align*}
where $\mathbf{S}$ is called the complex term and $\mathbf{\bar{S}}$ is called the neutral term. These constructions rely on an informal combinatorial operation usually denoted as `+.' In Greimas and Y\"{u}cel's analysis of the oeuvre of writer Georges Bernanos (1888-1948), they frame the construction of the complex term $\mathbf{S}$ (and therefore its contrary $\mathbf{\bar{S}}$) as that of an ``isotopy,'' \cite[page 260]{Greimas1983} or reoccurring pattern. In this article, we treat each of the meta-terms as sememes, so that their general construction may point to how isotopies function in a source text through the recurrence of shared semantic traits (classemes). Here, Greimas and Y\"{u}cel also note that ``the equilibrium of the two elements of the structure [$s_{1} + s_{2}$] is precarious, and it leans at times on the negative side, at other times on the positive side, thus instituting the dominance of one of the two elements.'' \citep[page 260-261]{Greimas1983} 

In our following Proposition~\ref{The:MetaTerms} we frame the `+' operation of meta-term generation as a composition of diagrammatic derivations in the inference system $\mathcal{R}$. Although our approach formalises the syntactic structure of the meta-terms, it concurrently accommodates what Greimas and Y\"{u}cel identify as ``structural apprehensions'' that emerge in both abstract and figurative manifestations of  sememes and isotopies.

\begin{figure}
\centering
\subfloat[$T_5: d_{s_{1}\wedge s_{2}} \vdash_{\mathcal{R}} d^{+}_{s_{1} s_{2}}$.\label{ProofTree5}]{%
\begin{adjustbox}{width=0.4\linewidth}
\begin{tikzpicture}[line width=0.5mm,font=\Large]
\def\circleB{(0,0) circle (1)}
\def\circleC{(1,0) circle (1)}
\def\circleA{(0.5,1) circle (1)}
\tikzset{fontscale/.style = {font=\relsize{#1}}}
%
\begin{scope}[local bounding box=d2] 
	\scope[even odd rule]
	\clip \circleB;
	\clip \circleC;
	\fill[red] \circleB ;
	\endscope
\draw \circleB node[label={[label distance=8mm,xshift=-5]-90:$S_{1}$}] (S1){};
\draw \circleC node [label={[label distance=8mm,xshift=5]-90:$S_{2}$}] (S2){};
\draw (0.5,0) circle (2) node[label={[label distance=20mm]90:$M$}] (M) {};
\node[circle,minimum size=3mm,inner sep=0,fill=black] at (-0.5,0)(s1){};
\draw (0.5,0.5) circle (3) node[label={[label distance=30mm]90:$X$}]{};
\begin{scope}[xshift=70mm]
	\scope[even odd rule]
	\clip \circleB;
	\clip \circleC;
	\fill[red] \circleB ;
	\endscope
\draw \circleB node[label={[label distance=8mm,xshift=-5]-90:$S_{1}$}] (S1){};
\draw \circleC node [label={[label distance=8mm,xshift=5]-90:$S_{2}$}] (S2){};
\draw (0.5,0) circle (2) node[label={[label distance=20mm]90:$M$}] (M) {};
\node[circle,minimum size=3mm,inner sep=0,fill=black] at (1.5,0)(s2){};
\draw (0.5,0.5) circle (3) node[label={[label distance=30mm]90:$X$}]{};
\node[draw, rectangle, fit=(d2), inner sep=10pt,label={[font=\Huge]90:$d^{+}_{s_{1} \wedge s_{2}}$}] (d2BOX) {};
\node[anchor=north west,xshift=5,yshift=-5] at (d2BOX.north west) {\Huge$\boxed\wedge$};
\end{scope}
\end{scope}
\begin{scope}[local bounding box=d3,shift={($(d2BOX.south) + (-0.5,-7)$)}]
	\scope[even odd rule]
	\clip \circleB;
	\clip \circleC;
	\fill[red] \circleB ;
	\endscope
\draw \circleB node[label={[label distance=8mm,xshift=-5]-90:$S_{1}$}] (S1){};
\draw \circleC node [label={[label distance=8mm,xshift=5]-90:$S_{2}$}] (S2){};
\draw (0.5,0) circle (2) node[label={[label distance=20mm]90:$M$}] (M) {};
\node[circle,minimum size=3mm,inner sep=0,fill=black] at (1.5,0)(s2){};
\node[circle,minimum size=3mm,inner sep=0,fill=black] at (-0.5,0)(s1){};
\draw (0.5,0.5) circle (3) node[label={[label distance=30mm]90:$X$}]{};
\node[draw, rectangle, fit=(d3), inner sep=10pt] (d3BOX) {};
\end{scope}
\begin{scope}[local bounding box=d4,shift={($(d3BOX.south) + (-0.5,-7)$)}]
	\scope[even odd rule]
	\clip \circleB;
	\clip \circleC;
	\fill[red] \circleB ;
	\endscope
\draw \circleB node[label={[label distance=8mm,xshift=-5]-90:$S_{1}$}] (S1){};
\draw \circleC node [label={[label distance=8mm,xshift=5]-90:$S_{2}$}] (S2){};
\draw (0.5,0) circle (2) node[label={[label distance=20mm]90:$M$}] (M) {};
\node[circle,minimum size=3mm,inner sep=0,fill=black] at (1.5,0)(s2){};
\node[circle,minimum size=3mm,inner sep=0,fill=black] at (-0.5,0)(s1){};
\node[circle,minimum size=3mm,inner sep=0,fill=black] at (0.5,1.5)(m1){};
\draw(m1)--(s1);
\draw (0.5,0.5) circle (3) node[label={[label distance=30mm]90:$X$}]{};
\node[draw, rectangle, fit=(d4), inner sep=10pt] (d4BOX) {};
\end{scope}
\begin{scope}[local bounding box=d5,shift={($(d4BOX.south) + (3,-7)$)}] 
	\scope[even odd rule]
	\clip \circleB;
	\clip \circleC;
	\fill[red] \circleB ;
	\endscope
\draw \circleB node[label={[label distance=8mm,xshift=-5]-90:$S_{1}$}] (S1){};
\draw \circleC node [label={[label distance=8mm,xshift=5]-90:$S_{2}$}] (S2){};
\draw (0.5,0) circle (2) node[label={[label distance=20mm]90:$M$}] (M) {};
\node[circle,minimum size=3mm,inner sep=0,fill=black] at (0.5,1.5)(m1){};
\draw (0.5,0.5) circle (3) node[label={[label distance=30mm]90:$X$}]{};
\begin{scope}[xshift=-70mm]
	\scope[even odd rule]
	\clip \circleB;
	\clip \circleC;
	\fill[red] \circleB ;
	\endscope
\draw \circleB node[label={[label distance=8mm,xshift=-5]-90:$S_{1}$}] (S1){};
\draw \circleC node [label={[label distance=8mm,xshift=5]-90:$S_{2}$}] (S2){};
\draw (0.5,0) circle (2) node[label={[label distance=20mm]90:$M$}] (M) {};
\node[circle,minimum size=3mm,inner sep=0,fill=black] at (1.5,0)(s2){};
\node[circle,minimum size=3mm,inner sep=0,fill=black] at (-0.5,0)(s1){};
\draw (0.5,0.5) circle (3) node[label={[label distance=30mm]90:$X$}]{};
\node[draw, rectangle, fit=(d5), inner sep=10pt] (d5BOX) {};
\node[anchor=north west,xshift=5,yshift=-5] at (d5BOX.north west) {\Huge$\boxed\vee$};
\end{scope}
\end{scope}
\begin{scope}[local bounding box=d6,shift={($(d5BOX.south) + (-10.5,-7)$)}]
	\scope[even odd rule]
	\clip \circleB;
	\clip \circleC;
	\fill[red] \circleB ;
	\endscope
\draw \circleB node[label={[label distance=8mm,xshift=-5]-90:$S_{1}$}] (S1){};
\draw \circleC node [label={[label distance=8mm,xshift=5]-90:$S_{2}$}] (S2){};
\draw (0.5,0) circle (2) node[label={[label distance=20mm]90:$M$}] (M) {};
\node[circle,minimum size=3mm,inner sep=0,fill=black] at (1.5,0)(s2){};
\node[circle,minimum size=3mm,inner sep=0,fill=black] at (-0.5,0)(s1){};
\draw (0.5,0.5) circle (3) node[label={[label distance=30mm]90:$X$}]{};
\node[draw, rectangle, fit=(d6), inner sep=10pt] (d6BOX) {};
\end{scope}
\begin{scope}[local bounding box=d7,shift={($(d5BOX.south) + (10,-7)$)}]
	\scope[even odd rule]
	\clip \circleB;
	\clip \circleC;
	\fill[red] \circleB ;
	\endscope
\draw \circleB node[label={[label distance=8mm,xshift=-5]-90:$S_{1}$}] (S1){};
\draw \circleC node [label={[label distance=8mm,xshift=5]-90:$S_{2}$}] (S2){};
\draw (0.5,0) circle (2) node[label={[label distance=20mm]90:$M$}] (M) {};
\node[circle,minimum size=3mm,inner sep=0,fill=black] at (0.5,1.5)(m1){};
\draw (0.5,0.5) circle (3) node[label={[label distance=30mm]90:$X$}]{};
\node[draw, rectangle, fit=(d7), inner sep=10pt] (d7BOX) {};
\end{scope}
\begin{scope}[local bounding box=d8,shift={($(d6BOX.south) + (-4,-7)$)}] 
	\scope[even odd rule]
	\clip \circleB;
	\clip \circleC;
	\fill[red] \circleB ;
	\endscope
\draw \circleB node[label={[label distance=8mm,xshift=-5]-90:$S_{1}$}] (S1){};
\draw \circleC node [label={[label distance=8mm,xshift=5]-90:$S_{2}$}] (S2){};
\draw (0.5,0) circle (2) node[label={[label distance=20mm]90:$M$}] (M) {};
\node[circle,minimum size=3mm,inner sep=0,fill=black] at (0.5,1.5)(m1){};
\node[circle,minimum size=3mm,inner sep=0,fill=black] at (1.5,0)(s2){};
\node[circle,minimum size=3mm,inner sep=0,fill=black] at (-0.5,0)(s1){};
\draw (0.5,0.5) circle (3) node[label={[label distance=30mm]90:$X$}]{};
\begin{scope}[xshift=70mm]
	\scope[even odd rule]
	\clip \circleB;
	\clip \circleC;
	\fill[red] \circleB ;
	\endscope
\draw \circleB node[label={[label distance=8mm,xshift=-5]-90:$S_{1}$}] (S1){};
\draw \circleC node [label={[label distance=8mm,xshift=5]-90:$S_{2}$}] (S2){};
\draw (0.5,0) circle (2) node[label={[label distance=20mm]90:$M$}] (M) {};
\node[circle,minimum size=3mm,inner sep=0,fill=black] at (0.5,1.5)(m1){};
\draw (0.5,0.5) circle (3) node[label={[label distance=30mm]90:$X$}]{};
\node[draw, rectangle, fit=(d8), inner sep=10pt] (d8BOX) {};
\node[anchor=north west,xshift=5,yshift=-5] at (d8BOX.north west) {\Huge$\boxed\wedge$};
\end{scope}
\end{scope}
\begin{scope}[local bounding box=d9,shift={($(d7BOX.south) + (-4,-7)$)}] 
	\scope[even odd rule]
	\clip \circleB;
	\clip \circleC;
	\fill[red] \circleB ;
	\endscope
\draw \circleB node[label={[label distance=8mm,xshift=-5]-90:$S_{1}$}] (S1){};
\draw \circleC node [label={[label distance=8mm,xshift=5]-90:$S_{2}$}] (S2){};
\draw (0.5,0) circle (2) node[label={[label distance=20mm]90:$M$}] (M) {};
\node[circle,minimum size=3mm,inner sep=0,fill=black] at (0.5,1.5)(m1){};
\node[circle,minimum size=3mm,inner sep=0,fill=black] at (1.5,0)(s2){};
\node[circle,minimum size=3mm,inner sep=0,fill=black] at (-0.5,0)(s1){};
\draw (0.5,0.5) circle (3) node[label={[label distance=30mm]90:$X$}]{};
\begin{scope}[xshift=70mm]
	\scope[even odd rule]
	\clip \circleB;
	\clip \circleC;
	\fill[red] \circleB ;
	\endscope
\draw \circleB node[label={[label distance=8mm,xshift=-5]-90:$S_{1}$}] (S1){};
\draw \circleC node [label={[label distance=8mm,xshift=5]-90:$S_{2}$}] (S2){};
\draw (0.5,0) circle (2) node[label={[label distance=20mm]90:$M$}] (M) {};
\node[circle,minimum size=3mm,inner sep=0,fill=black] at (0.5,1.5)(m1){};
\draw (0.5,0.5) circle (3) node[label={[label distance=30mm]90:$X$}]{};
\node[draw, rectangle, fit=(d9), inner sep=10pt] (d9BOX) {};
\node[anchor=north west,xshift=5,yshift=-5] at (d9BOX.north west) {\Huge$\boxed\wedge$};
\end{scope}
\end{scope}
\begin{scope}[local bounding box=d10,shift={($(d8BOX.south) + (-0.5,-7)$)}]
	\scope[even odd rule]
	\clip \circleB;
	\clip \circleC;
	\fill[red] \circleB ;
	\endscope
\draw \circleB node[label={[label distance=8mm,xshift=-5]-90:$S_{1}$}] (S1){};
\draw \circleC node [label={[label distance=8mm,xshift=5]-90:$S_{2}$}] (S2){};
\draw (0.5,0) circle (2) node[label={[label distance=20mm]90:$M$}] (M) {};
\node[circle,minimum size=3mm,inner sep=0,fill=black] at (0.5,1.5)(m1){};
\node[circle,minimum size=3mm,inner sep=0,fill=black] at (1.5,0)(s2){};
\node[circle,minimum size=3mm,inner sep=0,fill=black] at (-0.5,0)(s1){};
\draw (0.5,0.5) circle (3) node[label={[label distance=30mm]90:$X$}]{};
\node[draw, rectangle, fit=(d10), inner sep=10pt] (d10BOX) {};
\end{scope}
\begin{scope}[local bounding box=d11,shift={($(d9BOX.south) + (-0.5,-7)$)}]
	\scope[even odd rule]
	\clip \circleB;
	\clip \circleC;
	\fill[red] \circleB ;
	\endscope
\draw \circleB node[label={[label distance=8mm,xshift=-5]-90:$S_{1}$}] (S1){};
\draw \circleC node [label={[label distance=8mm,xshift=5]-90:$S_{2}$}] (S2){};
\draw (0.5,0) circle (2) node[label={[label distance=20mm]90:$M$}] (M) {};
\node[circle,minimum size=3mm,inner sep=0,fill=black] at (0.5,1.5)(m1){};
\node[circle,minimum size=3mm,inner sep=0,fill=black] at (1.5,0)(s2){};
\node[circle,minimum size=3mm,inner sep=0,fill=black] at (-0.5,0)(s1){};
\draw (0.5,0.5) circle (3) node[label={[label distance=30mm]90:$X$}]{};
\node[draw, rectangle, fit=(d11), inner sep=10pt] (d11BOX) {};
\end{scope}
\begin{scope}[local bounding box=d12,shift={($(d11BOX.south) + (-7.25,-7)$)}] 
	\scope[even odd rule]
	\clip \circleB;
	\clip \circleC;
	\fill[red] \circleB ;
	\endscope
\draw \circleB node[label={[label distance=8mm,xshift=-5]-90:$S_{1}$}] (S1){};
\draw \circleC node [label={[label distance=8mm,xshift=5]-90:$S_{2}$}] (S2){};
\draw (0.5,0) circle (2) node[label={[label distance=20mm]90:$M$}] (M) {};
\node[circle,minimum size=3mm,inner sep=0,fill=black] at (0.5,1.5)(m1){};
\node[circle,minimum size=3mm,inner sep=0,fill=black] at (1.5,0)(s2){};
\node[circle,minimum size=3mm,inner sep=0,fill=black] at (-0.5,0)(s1){};
\draw (0.5,0.5) circle (3) node[label={[label distance=30mm]90:$X$}]{};
\begin{scope}[xshift=-70mm]
	\scope[even odd rule]
	\clip \circleB;
	\clip \circleC;
	\fill[red] \circleB ;
	\endscope
\draw \circleB node[label={[label distance=8mm,xshift=-5]-90:$S_{1}$}] (S1){};
\draw \circleC node [label={[label distance=8mm,xshift=5]-90:$S_{2}$}] (S2){};
\draw (0.5,0) circle (2) node[label={[label distance=20mm]90:$M$}] (M) {};
\node[circle,minimum size=3mm,inner sep=0,fill=black] at (0.5,1.5)(m1){};
\node[circle,minimum size=3mm,inner sep=0,fill=black] at (1.5,0)(s2){};
\node[circle,minimum size=3mm,inner sep=0,fill=black] at (-0.5,0)(s1){};
\draw (0.5,0.5) circle (3) node[label={[label distance=30mm]90:$X$}]{};
\node[draw, rectangle, fit=(d12), inner sep=10pt] (d12BOX) {};
\node[anchor=north west,xshift=5,yshift=-5] at (d12BOX.north west) {\Huge$\boxed\vee$};
\end{scope}
\end{scope}
\begin{scope}[local bounding box=d13,shift={($(d12BOX.south) + (-0.5,-7)$)}]
	\scope[even odd rule]
	\clip \circleB;
	\clip \circleC;
	\fill[red] \circleB ;
	\endscope
\draw \circleB node[label={[label distance=8mm,xshift=-5]-90:$S_{1}$}] (S1){};
\draw \circleC node [label={[label distance=8mm,xshift=5]-90:$S_{2}$}] (S2){};
\draw (0.5,0) circle (2) node[label={[label distance=20mm]90:$M$}] (M) {};
\node[circle,minimum size=3mm,inner sep=0,fill=black] at (0.5,1.5)(m1){};
\node[circle,minimum size=3mm,inner sep=0,fill=black] at (1.5,0)(s2){};
\node[circle,minimum size=3mm,inner sep=0,fill=black] at (-0.5,0)(s1){};
\draw (0.5,0.5) circle (3) node[label={[label distance=30mm]90:$X$}]{};
\node[draw, rectangle, fit=(d13), inner sep=10pt,label={[font=\Huge]-90:$d^{+}_{ s_{1} s_{2}}$}] (d13BOX) {};
\end{scope}
\draw[thick,->,double equal sign distance](d2BOX)--node[right,label={[align=center,font=\Huge]180:\scshape{Combine}}]{}(d3BOX);
\draw[thick,->,double equal sign distance](d3BOX)--node[right,label={[align=center,font=\Huge]180:\scshape{AddFeet}}]{}(d4BOX);
\draw[thick,->,double equal sign distance](d4BOX)--node[right,label={[align=center,font=\Huge]180:\scshape{SplitSpider}}]{}(d5BOX);
\draw[thick,->,double equal sign distance,shorten >=20pt](d5BOX.south)--(d6BOX.north);
\draw[thick,->,double equal sign distance,shorten >=20pt](d5BOX.south)--(d7BOX.north);
\draw[thick,->,double equal sign distance](d6BOX)--node[right,label={[align=center,font=\Huge]180:\scshape{CopySpider}}]{}(d8BOX);
\draw[thick,->,double equal sign distance](d7BOX)--node[right,label={[align=center,font=\Huge]0:\scshape{CopySpider}}]{}(d9BOX);
\draw[thick,->,double equal sign distance](d8BOX)--node[right,label={[align=center,font=\Huge]180:\scshape{Combine}}]{}(d10BOX);
\draw[thick,->,double equal sign distance](d9BOX)--node[right,label={[align=center,font=\Huge]0:\scshape{Combine}}]{}(d11BOX);
\draw[thick,->,double equal sign distance,shorten >=20pt](d10BOX.south)--(d12BOX.north);
\draw[thick,->,double equal sign distance,shorten >=20pt](d11BOX.south)--(d12BOX.north);
\draw[thick,->,double equal sign distance](d12BOX)--node[right,label={[align=center,font=\Huge]0:\scshape{Idempotency}}]{}(d13BOX);
\end{tikzpicture}
\end{adjustbox}
}
\subfloat[$T_6: d_{\neg s_{1}\wedge \neg s_{2}} \vdash_{\mathcal{R}} d^{+}_{\neg s_{1}\neg s_{2}}$.\label{ProofTree6}]{%
\begin{adjustbox}{width=0.5\linewidth}
\begin{tikzpicture}[line width=0.5mm,font=\Large]
\def\circleB{(0,0) circle (1)}
\def\circleC{(1,0) circle (1)}
\def\circleA{(0.5,1) circle (1)}
\tikzset{fontscale/.style = {font=\relsize{#1}}}
\begin{scope}[local bounding box=d1,shift={(0,0)}]
	\scope[even odd rule]
	\clip \circleB;
	\clip \circleC;
	\fill[red] \circleB ;
	\endscope
\draw \circleB node[label={[label distance=8mm,xshift=-5]-90:$S_{1}$}] (S1){};
\draw \circleC node [label={[label distance=8mm,xshift=5]-90:$S_{2}$}] (S2){};
\draw (0.5,0) circle (2) node[label={[label distance=20mm]90:$M$}] (M) {};
\node[circle,minimum size=3mm,inner sep=0,fill=black] at (-0.5,0)(s1){};
\node[circle,minimum size=3mm,inner sep=0,fill=black] at (-2,0)(x1){};
\draw(x1)--(s1);
\draw (0.5,0.5) circle (3) node[label={[label distance=30mm]90:$X$}]{};
\begin{scope}[xshift=70mm]
	\scope[even odd rule]
	\clip \circleB;
	\clip \circleC;
	\fill[red] \circleB ;
	\endscope
\draw \circleB node[label={[label distance=8mm,xshift=-5]-90:$S_{1}$}] (S1){};
\draw \circleC node [label={[label distance=8mm,xshift=5]-90:$S_{2}$}] (S2){};
\draw (0.5,0) circle (2) node[label={[label distance=20mm]90:$M$}] (M) {};
\node[circle,minimum size=3mm,inner sep=0,fill=black] at (1.5,0)(s2){};
\node[circle,minimum size=3mm,inner sep=0,fill=black] at (3,0)(x1){};
\draw(x1)--(s2);
\draw (0.5,0.5) circle (3) node[label={[label distance=30mm]90:$X$}]{};
\node[draw, rectangle, fit=(d1), inner sep=10pt,label={[font=\Huge]90:$d^{+}_{\neg s_{1}\wedge\neg s_{2}}$}] (d1BOX) {};
\node[anchor=north west,xshift=5,yshift=-5] at (d1BOX.north west) {\Huge$\boxed\wedge$};
\end{scope}
\end{scope}
\begin{scope}[local bounding box=d2,shift={($(d1BOX.south) + (-4,-7)$)}] 
	\scope[even odd rule]
	\clip \circleB;
	\clip \circleC;
	\fill[red] \circleB ;
	\endscope
\draw \circleB node[label={[label distance=8mm,xshift=-5]-90:$S_{1}$}] (S1){};
\draw \circleC node [label={[label distance=8mm,xshift=5]-90:$S_{2}$}] (S2){};
\draw (0.5,0) circle (2) node[label={[label distance=20mm]90:$M$}] (M) {};
\node[circle,minimum size=3mm,inner sep=0,fill=black] at (-0.5,0)(s1){};
\node[circle,minimum size=3mm,inner sep=0,fill=black] at (-2,0)(x1){};
\draw(x1)--(s1);
\node[circle,minimum size=3mm,inner sep=0,fill=black] at (1.5,0)(s2){};
\node[circle,minimum size=3mm,inner sep=0,fill=black] at (3,0)(x2){};
\draw(x2)--(s2);
\draw (0.5,0.5) circle (3) node[label={[label distance=30mm]90:$X$}]{};
\begin{scope}[xshift=70mm]
	\scope[even odd rule]
	\clip \circleB;
	\clip \circleC;
	\fill[red] \circleB ;
	\endscope
\draw \circleB node[label={[label distance=8mm,xshift=-5]-90:$S_{1}$}] (S1){};
\draw \circleC node [label={[label distance=8mm,xshift=5]-90:$S_{2}$}] (S2){};
\draw (0.5,0) circle (2) node[label={[label distance=20mm]90:$M$}] (M) {};
\node[circle,minimum size=3mm,inner sep=0,fill=black] at (1.5,0)(s2){};
\node[circle,minimum size=3mm,inner sep=0,fill=black] at (3,0)(x1){};
\draw(x1)--(s2);
\draw (0.5,0.5) circle (3) node[label={[label distance=30mm]90:$X$}]{};
\node[draw, rectangle, fit=(d2), inner sep=10pt] (d2BOX) {};
\node[anchor=north west,xshift=5,yshift=-5] at (d2BOX.north west) {\Huge$\boxed\wedge$};
\end{scope}
\end{scope}
\begin{scope}[local bounding box=d3,shift={($(d2BOX.south) + (-0.5,-7)$)}]
	\scope[even odd rule]
	\clip \circleB;
	\clip \circleC;
	\fill[red] \circleB ;
	\endscope
\draw \circleB node[label={[label distance=8mm,xshift=-5]-90:$S_{1}$}] (S1){};
\draw \circleC node [label={[label distance=8mm,xshift=5]-90:$S_{2}$}] (S2){};
\draw (0.5,0) circle (2) node[label={[label distance=20mm]90:$M$}] (M) {};
\node[circle,minimum size=3mm,inner sep=0,fill=black] at (1.5,0)(s2){};
\node[circle,minimum size=3mm,inner sep=0,fill=black] at (-0.5,0)(s1){};
\node[circle,minimum size=3mm,inner sep=0,fill=black] at (-2,0)(x1){};
\node[circle,minimum size=3mm,inner sep=0,fill=black] at (3,0)(x2){};
\draw(x2)--(s2);
\draw(x1)--(s1);
\draw (0.5,0.5) circle (3) node[label={[label distance=30mm]90:$X$}]{};
\node[draw, rectangle, fit=(d3), inner sep=10pt] (d3BOX) {};
\end{scope}
\begin{scope}[local bounding box=d4,shift={($(d3BOX.south) + (-0.5,-7)$)}]
	\scope[even odd rule]
	\clip \circleB;
	\clip \circleC;
	\fill[red] \circleB ;
	\endscope
\draw \circleB node[label={[label distance=8mm,xshift=-5]-90:$S_{1}$}] (S1){};
\draw \circleC node [label={[label distance=8mm,xshift=5]-90:$S_{2}$}] (S2){};
\draw (0.5,0) circle (2) node[label={[label distance=20mm]90:$M$}] (M) {};
\node[circle,minimum size=3mm,inner sep=0,fill=black] at (1.5,0)(s2){};
\node[circle,minimum size=3mm,inner sep=0,fill=black] at (-0.5,0)(s1){};
\node[circle,minimum size=3mm,inner sep=0,fill=black] at (0.5,1.5)(m1){};
\node[circle,minimum size=3mm,inner sep=0,fill=black] at (-2,0)(x1){};
\node[circle,minimum size=3mm,inner sep=0,fill=black] at (3,0)(x2){};
\draw(x2)--(s2);
\draw(x1)--(s1);
\draw(m1)--(s1);
\draw (0.5,0.5) circle (3) node[label={[label distance=30mm]90:$X$}]{};
\node[draw, rectangle, fit=(d4), inner sep=10pt] (d4BOX) {};
\end{scope}
\begin{scope}[local bounding box=d5,shift={($(d4BOX.south) + (3,-7)$)}] 
	\scope[even odd rule]
	\clip \circleB;
	\clip \circleC;
	\fill[red] \circleB ;
	\endscope
\draw \circleB node[label={[label distance=8mm,xshift=-5]-90:$S_{1}$}] (S1){};
\draw \circleC node [label={[label distance=8mm,xshift=5]-90:$S_{2}$}] (S2){};
\draw (0.5,0) circle (2) node[label={[label distance=20mm]90:$M$}] (M) {};
\node[circle,minimum size=3mm,inner sep=0,fill=black] at (0.5,1.5)(m1){};
\draw (0.5,0.5) circle (3) node[label={[label distance=30mm]90:$X$}]{};
\begin{scope}[xshift=-70mm]
	\scope[even odd rule]
	\clip \circleB;
	\clip \circleC;
	\fill[red] \circleB ;
	\endscope
\draw \circleB node[label={[label distance=8mm,xshift=-5]-90:$S_{1}$}] (S1){};
\draw \circleC node [label={[label distance=8mm,xshift=5]-90:$S_{2}$}] (S2){};
\draw (0.5,0) circle (2) node[label={[label distance=20mm]90:$M$}] (M) {};
\node[circle,minimum size=3mm,inner sep=0,fill=black] at (1.5,0)(s2){};
\node[circle,minimum size=3mm,inner sep=0,fill=black] at (-0.5,0)(s1){};
\node[circle,minimum size=3mm,inner sep=0,fill=black] at (-2,0)(x1){};
\node[circle,minimum size=3mm,inner sep=0,fill=black] at (3,0)(x2){};
\draw(x2)--(s2);
\draw(x1)--(s1);
\draw (0.5,0.5) circle (3) node[label={[label distance=30mm]90:$X$}]{};
\node[draw, rectangle, fit=(d5), inner sep=10pt] (d5BOX) {};
\node[anchor=north west,xshift=5,yshift=-5] at (d5BOX.north west) {\Huge$\boxed\vee$};
\end{scope}
\end{scope}
\begin{scope}[local bounding box=d6,shift={($(d5BOX.south) + (-10.5,-7)$)}]
	\scope[even odd rule]
	\clip \circleB;
	\clip \circleC;
	\fill[red] \circleB ;
	\endscope
\draw \circleB node[label={[label distance=8mm,xshift=-5]-90:$S_{1}$}] (S1){};
\draw \circleC node [label={[label distance=8mm,xshift=5]-90:$S_{2}$}] (S2){};
\draw (0.5,0) circle (2) node[label={[label distance=20mm]90:$M$}] (M) {};
\node[circle,minimum size=3mm,inner sep=0,fill=black] at (1.5,0)(s2){};
\node[circle,minimum size=3mm,inner sep=0,fill=black] at (-0.5,0)(s1){};
\node[circle,minimum size=3mm,inner sep=0,fill=black] at (-2,0)(x1){};
\node[circle,minimum size=3mm,inner sep=0,fill=black] at (3,0)(x2){};
\draw(x2)--(s2);
\draw(x1)--(s1);
\draw (0.5,0.5) circle (3) node[label={[label distance=30mm]90:$X$}]{};
\node[draw, rectangle, fit=(d6), inner sep=10pt] (d6BOX) {};
\end{scope}
\begin{scope}[local bounding box=d7,shift={($(d5BOX.south) + (10,-7)$)}]
	\scope[even odd rule]
	\clip \circleB;
	\clip \circleC;
	\fill[red] \circleB ;
	\endscope
\draw \circleB node[label={[label distance=8mm,xshift=-5]-90:$S_{1}$}] (S1){};
\draw \circleC node [label={[label distance=8mm,xshift=5]-90:$S_{2}$}] (S2){};
\draw (0.5,0) circle (2) node[label={[label distance=20mm]90:$M$}] (M) {};
\node[circle,minimum size=3mm,inner sep=0,fill=black] at (0.5,1.5)(m1){};
\draw (0.5,0.5) circle (3) node[label={[label distance=30mm]90:$X$}]{};
\node[draw, rectangle, fit=(d7), inner sep=10pt] (d7BOX) {};
\end{scope}
\begin{scope}[local bounding box=d8,shift={($(d6BOX.south) + (-4,-7)$)}] 
	\scope[even odd rule]
	\clip \circleB;
	\clip \circleC;
	\fill[red] \circleB ;
	\endscope
\draw \circleB node[label={[label distance=8mm,xshift=-5]-90:$S_{1}$}] (S1){};
\draw \circleC node [label={[label distance=8mm,xshift=5]-90:$S_{2}$}] (S2){};
\draw (0.5,0) circle (2) node[label={[label distance=20mm]90:$M$}] (M) {};
\node[circle,minimum size=3mm,inner sep=0,fill=black] at (0.5,1.5)(m1){};
\node[circle,minimum size=3mm,inner sep=0,fill=black] at (-0.5,0)(s1){};
\node[circle,minimum size=3mm,inner sep=0,fill=black] at (-2,0)(x1){};
\draw(x1)--(s1);
\node[circle,minimum size=3mm,inner sep=0,fill=black] at (1.5,0)(s2){};
\node[circle,minimum size=3mm,inner sep=0,fill=black] at (3,0)(x2){};
\draw(x2)--(s2);
\draw (0.5,0.5) circle (3) node[label={[label distance=30mm]90:$X$}]{};
\begin{scope}[xshift=70mm]
	\scope[even odd rule]
	\clip \circleB;
	\clip \circleC;
	\fill[red] \circleB ;
	\endscope
\draw \circleB node[label={[label distance=8mm,xshift=-5]-90:$S_{1}$}] (S1){};
\draw \circleC node [label={[label distance=8mm,xshift=5]-90:$S_{2}$}] (S2){};
\draw (0.5,0) circle (2) node[label={[label distance=20mm]90:$M$}] (M) {};
\node[circle,minimum size=3mm,inner sep=0,fill=black] at (0.5,1.5)(m1){};
\draw (0.5,0.5) circle (3) node[label={[label distance=30mm]90:$X$}]{};
\node[draw, rectangle, fit=(d8), inner sep=10pt] (d8BOX) {};
\node[anchor=north west,xshift=5,yshift=-5] at (d8BOX.north west) {\Huge$\boxed\wedge$};
\end{scope}
\end{scope}
\begin{scope}[local bounding box=d9,shift={($(d7BOX.south) + (-4,-7)$)}] 
	\scope[even odd rule]
	\clip \circleB;
	\clip \circleC;
	\fill[red] \circleB ;
	\endscope
\draw \circleB node[label={[label distance=8mm,xshift=-5]-90:$S_{1}$}] (S1){};
\draw \circleC node [label={[label distance=8mm,xshift=5]-90:$S_{2}$}] (S2){};
\draw (0.5,0) circle (2) node[label={[label distance=20mm]90:$M$}] (M) {};
\node[circle,minimum size=3mm,inner sep=0,fill=black] at (0.5,1.5)(m1){};
\node[circle,minimum size=3mm,inner sep=0,fill=black] at (-0.5,0)(s1){};
\node[circle,minimum size=3mm,inner sep=0,fill=black] at (-2,0)(x1){};
\draw(x1)--(s1);
\node[circle,minimum size=3mm,inner sep=0,fill=black] at (1.5,0)(s2){};
\node[circle,minimum size=3mm,inner sep=0,fill=black] at (3,0)(x2){};
\draw(x2)--(s2);
\draw (0.5,0.5) circle (3) node[label={[label distance=30mm]90:$X$}]{};
\begin{scope}[xshift=70mm]
	\scope[even odd rule]
	\clip \circleB;
	\clip \circleC;
	\fill[red] \circleB ;
	\endscope
\draw \circleB node[label={[label distance=8mm,xshift=-5]-90:$S_{1}$}] (S1){};
\draw \circleC node [label={[label distance=8mm,xshift=5]-90:$S_{2}$}] (S2){};
\draw (0.5,0) circle (2) node[label={[label distance=20mm]90:$M$}] (M) {};
\node[circle,minimum size=3mm,inner sep=0,fill=black] at (0.5,1.5)(m1){};
\draw (0.5,0.5) circle (3) node[label={[label distance=30mm]90:$X$}]{};
\node[draw, rectangle, fit=(d9), inner sep=10pt] (d9BOX) {};
\node[anchor=north west,xshift=5,yshift=-5] at (d9BOX.north west) {\Huge$\boxed\wedge$};
\end{scope}
\end{scope}
\begin{scope}[local bounding box=d10,shift={($(d8BOX.south) + (-0.5,-7)$)}]
	\scope[even odd rule]
	\clip \circleB;
	\clip \circleC;
	\fill[red] \circleB ;
	\endscope
\draw \circleB node[label={[label distance=8mm,xshift=-5]-90:$S_{1}$}] (S1){};
\draw \circleC node [label={[label distance=8mm,xshift=5]-90:$S_{2}$}] (S2){};
\draw (0.5,0) circle (2) node[label={[label distance=20mm]90:$M$}] (M) {};
\node[circle,minimum size=3mm,inner sep=0,fill=black] at (0.5,1.5)(m1){};
\node[circle,minimum size=3mm,inner sep=0,fill=black] at (1.5,0)(s2){};
\node[circle,minimum size=3mm,inner sep=0,fill=black] at (-0.5,0)(s1){};
\node[circle,minimum size=3mm,inner sep=0,fill=black] at (-2,0)(x1){};
\node[circle,minimum size=3mm,inner sep=0,fill=black] at (3,0)(x2){};
\draw(x2)--(s2);
\draw(x1)--(s1);
\draw (0.5,0.5) circle (3) node[label={[label distance=30mm]90:$X$}]{};
\node[draw, rectangle, fit=(d10), inner sep=10pt] (d10BOX) {};
\end{scope}
\begin{scope}[local bounding box=d11,shift={($(d9BOX.south) + (-0.5,-7)$)}]
	\scope[even odd rule]
	\clip \circleB;
	\clip \circleC;
	\fill[red] \circleB ;
	\endscope
\draw \circleB node[label={[label distance=8mm,xshift=-5]-90:$S_{1}$}] (S1){};
\draw \circleC node [label={[label distance=8mm,xshift=5]-90:$S_{2}$}] (S2){};
\draw (0.5,0) circle (2) node[label={[label distance=20mm]90:$M$}] (M) {};
\node[circle,minimum size=3mm,inner sep=0,fill=black] at (0.5,1.5)(m1){};
\node[circle,minimum size=3mm,inner sep=0,fill=black] at (1.5,0)(s2){};
\node[circle,minimum size=3mm,inner sep=0,fill=black] at (-0.5,0)(s1){};
\node[circle,minimum size=3mm,inner sep=0,fill=black] at (-2,0)(x1){};
\node[circle,minimum size=3mm,inner sep=0,fill=black] at (3,0)(x2){};
\draw(x2)--(s2);
\draw(x1)--(s1);
\draw (0.5,0.5) circle (3) node[label={[label distance=30mm]90:$X$}]{};
\node[draw, rectangle, fit=(d11), inner sep=10pt] (d11BOX) {};
\end{scope}
\begin{scope}[local bounding box=d12,shift={($(d11BOX.south) + (-7.25,-7)$)}] 
	\scope[even odd rule]
	\clip \circleB;
	\clip \circleC;
	\fill[red] \circleB ;
	\endscope
\draw \circleB node[label={[label distance=8mm,xshift=-5]-90:$S_{1}$}] (S1){};
\draw \circleC node [label={[label distance=8mm,xshift=5]-90:$S_{2}$}] (S2){};
\draw (0.5,0) circle (2) node[label={[label distance=20mm]90:$M$}] (M) {};
\node[circle,minimum size=3mm,inner sep=0,fill=black] at (0.5,1.5)(m1){};
\node[circle,minimum size=3mm,inner sep=0,fill=black] at (1.5,0)(s2){};
\node[circle,minimum size=3mm,inner sep=0,fill=black] at (-0.5,0)(s1){};
\node[circle,minimum size=3mm,inner sep=0,fill=black] at (-2,0)(x1){};
\node[circle,minimum size=3mm,inner sep=0,fill=black] at (3,0)(x2){};
\draw(x2)--(s2);
\draw(x1)--(s1);
\draw (0.5,0.5) circle (3) node[label={[label distance=30mm]90:$X$}]{};
\begin{scope}[xshift=-70mm]
	\scope[even odd rule]
	\clip \circleB;
	\clip \circleC;
	\fill[red] \circleB ;
	\endscope
\draw \circleB node[label={[label distance=8mm,xshift=-5]-90:$S_{1}$}] (S1){};
\draw \circleC node [label={[label distance=8mm,xshift=5]-90:$S_{2}$}] (S2){};
\draw (0.5,0) circle (2) node[label={[label distance=20mm]90:$M$}] (M) {};
\node[circle,minimum size=3mm,inner sep=0,fill=black] at (0.5,1.5)(m1){};
\node[circle,minimum size=3mm,inner sep=0,fill=black] at (1.5,0)(s2){};
\node[circle,minimum size=3mm,inner sep=0,fill=black] at (-0.5,0)(s1){};
\node[circle,minimum size=3mm,inner sep=0,fill=black] at (-2,0)(x1){};
\node[circle,minimum size=3mm,inner sep=0,fill=black] at (3,0)(x2){};
\draw(x2)--(s2);
\draw(x1)--(s1);
\draw (0.5,0.5) circle (3) node[label={[label distance=30mm]90:$X$}]{};
\node[draw, rectangle, fit=(d12), inner sep=10pt] (d12BOX) {};
\node[anchor=north west,xshift=5,yshift=-5] at (d12BOX.north west) {\Huge$\boxed\vee$};
\end{scope}
\end{scope}
\begin{scope}[local bounding box=d13,shift={($(d12BOX.south) + (-0.5,-7)$)}]
	\scope[even odd rule]
	\clip \circleB;
	\clip \circleC;
	\fill[red] \circleB ;
	\endscope
\draw \circleB node[label={[label distance=8mm,xshift=-5]-90:$S_{1}$}] (S1){};
\draw \circleC node [label={[label distance=8mm,xshift=5]-90:$S_{2}$}] (S2){};
\draw (0.5,0) circle (2) node[label={[label distance=20mm]90:$M$}] (M) {};
\node[circle,minimum size=3mm,inner sep=0,fill=black] at (0.5,1.5)(m1){};
\node[circle,minimum size=3mm,inner sep=0,fill=black] at (1.5,0)(s2){};
\node[circle,minimum size=3mm,inner sep=0,fill=black] at (-0.5,0)(s1){};
\node[circle,minimum size=3mm,inner sep=0,fill=black] at (-2,0)(x1){};
\node[circle,minimum size=3mm,inner sep=0,fill=black] at (3,0)(x2){};
\draw(x2)--(s2);
\draw(x1)--(s1);
\draw (0.5,0.5) circle (3) node[label={[label distance=30mm]90:$X$}]{};
\node[draw, rectangle, fit=(d13), inner sep=10pt,label={[font=\Huge]-90:$d^{+}_{\neg s_{1}\neg s_{2}}$}] (d13BOX) {};
\end{scope}
\draw[thick,->,double equal sign distance](d1BOX)--node[right,label={[align=center,font=\Huge]180:\scshape{CopySpider}}]{}(d2BOX);
\draw[thick,->,double equal sign distance](d2BOX)--node[right,label={[align=center,font=\Huge]180:\scshape{ConjunctionElimination}}]{}(d3BOX);
\draw[thick,->,double equal sign distance](d3BOX)--node[right,label={[align=center,font=\Huge]180:\scshape{AddFeet}}]{}(d4BOX);
\draw[thick,->,double equal sign distance](d4BOX)--node[right,label={[align=center,font=\Huge]180:\scshape{SplitSpider}}]{}(d5BOX);
\draw[thick,->,double equal sign distance,shorten >=20pt](d5BOX.south)--(d6BOX.north);
\draw[thick,->,double equal sign distance,shorten >=20pt](d5BOX.south)--(d7BOX.north);
\draw[thick,->,double equal sign distance](d6BOX)--node[right,label={[align=center,font=\Huge]180:\scshape{CopySpider}}]{}(d8BOX);
\draw[thick,->,double equal sign distance](d7BOX)--node[right,label={[align=center,font=\Huge]0:\scshape{CopySpider}}]{}(d9BOX);
\draw[thick,->,double equal sign distance](d8BOX)--node[right,label={[align=center,font=\Huge]180:\scshape{ConjunctionElimination}}]{}(d10BOX);
\draw[thick,->,double equal sign distance](d9BOX)--node[right,label={[align=center,font=\Huge]0:\scshape{ConjunctionElimination}}]{}(d11BOX);
\draw[thick,->,double equal sign distance,shorten >=20pt](d10BOX.south)--(d12BOX.north);
\draw[thick,->,double equal sign distance,shorten >=20pt](d11BOX.south)--(d12BOX.north);
\draw[thick,->,double equal sign distance](d12BOX)--node[right,label={[align=center,font=\Huge]0:\scshape{Idempotency}}]{}(d13BOX);
\end{tikzpicture}
\end{adjustbox}
}
\caption{Proof trees witnessing meta-term construction of $\mathbf{S}$ and $\mathbf{\bar{S}}$. Each tree corresponds to a derivation in the inference system $\mathcal{R}$, transforming an initial diagram representing a seme pair into a meta-term.} 
\label{ProofTrees5-6}
\end{figure}
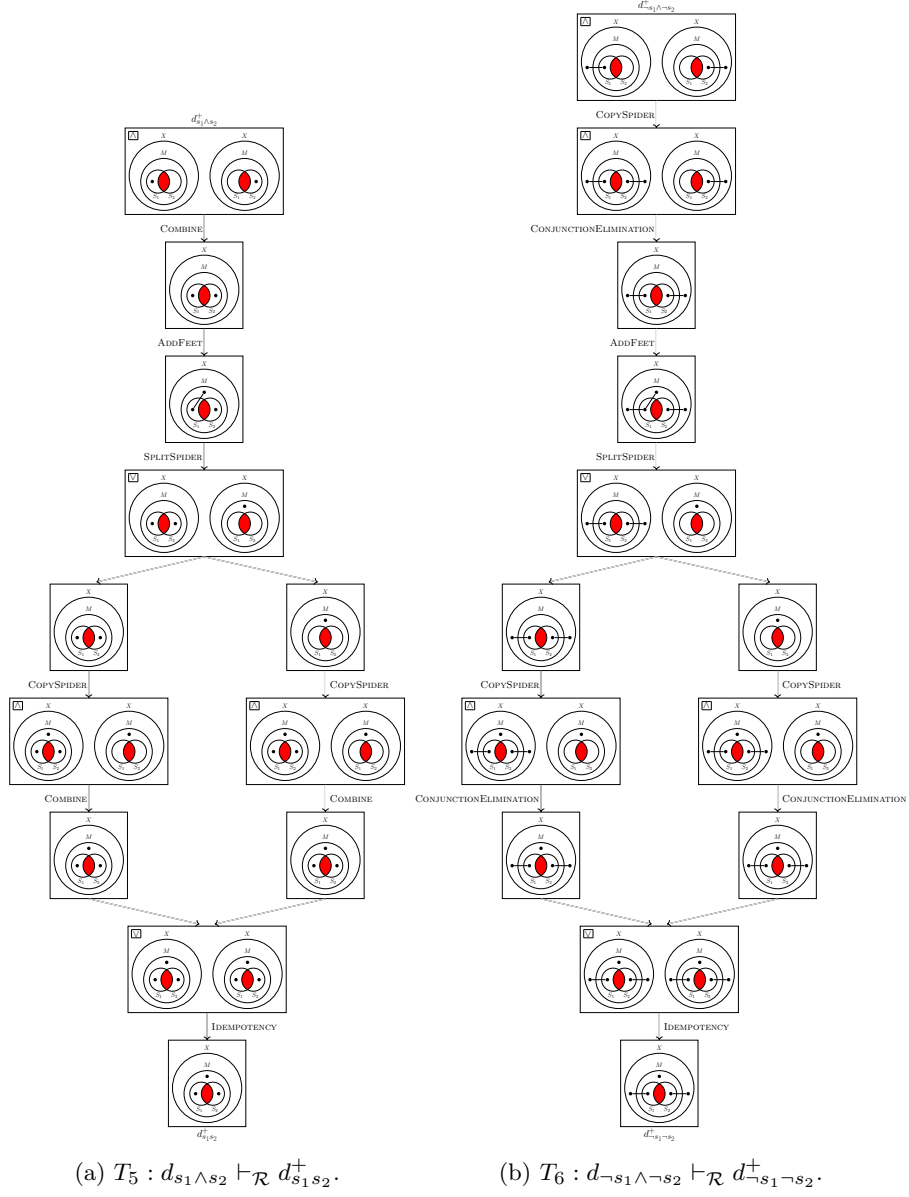

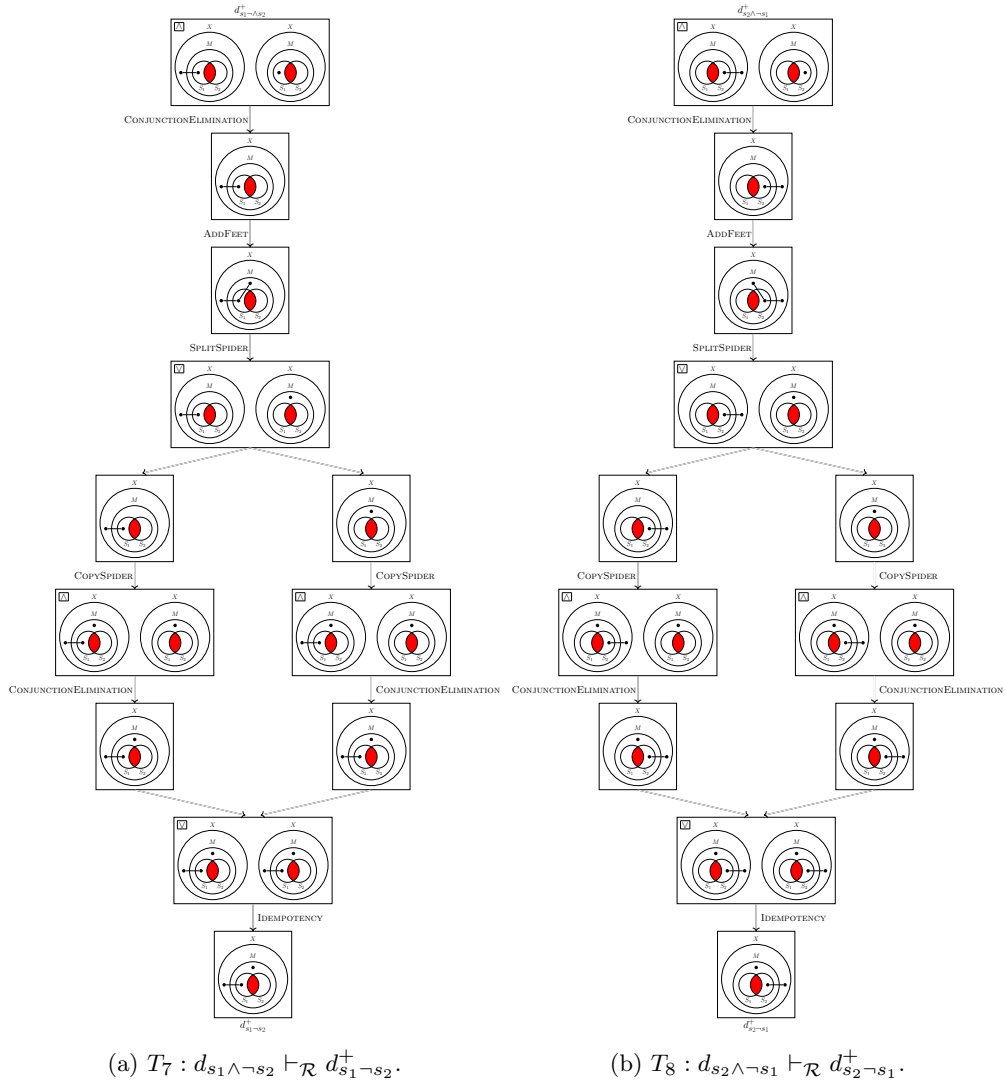
\begin{figure}
\centering
\subfloat[$T_7: d_{s_{1} \wedge \neg s_{2}} \vdash_{\mathcal{R}} d^{+}_{s_{1}\neg s_{2}}$.\label{ProofTree7}]{%
\begin{adjustbox}{width=0.5\linewidth}
\begin{tikzpicture}[line width=0.5mm,font=\Large]
\def\circleB{(0,0) circle (1)}
\def\circleC{(1,0) circle (1)}
\def\circleA{(0.5,1) circle (1)}
\tikzset{fontscale/.style = {font=\relsize{#1}}}
%
\begin{scope}[local bounding box=d2] 
	\scope[even odd rule]
	\clip \circleB;
	\clip \circleC;
	\fill[red] \circleB ;
	\endscope
\draw \circleB node[label={[label distance=8mm,xshift=-5]-90:$S_{1}$}] (S1){};
\draw \circleC node [label={[label distance=8mm,xshift=5]-90:$S_{2}$}] (S2){};
\draw (0.5,0) circle (2) node[label={[label distance=20mm]90:$M$}] (M) {};
\node[circle,minimum size=3mm,inner sep=0,fill=black] at (-0.5,0)(s1){};
\node[circle,minimum size=3mm,inner sep=0,fill=black] at (-2,0)(x1){};
\draw(x1)--(s1);
\draw (0.5,0.5) circle (3) node[label={[label distance=30mm]90:$X$}]{};
\begin{scope}[xshift=70mm]
	\scope[even odd rule]
	\clip \circleB;
	\clip \circleC;
	\fill[red] \circleB ;
	\endscope
\draw \circleB node[label={[label distance=8mm,xshift=-5]-90:$S_{1}$}] (S1){};
\draw \circleC node [label={[label distance=8mm,xshift=5]-90:$S_{2}$}] (S2){};
\draw (0.5,0) circle (2) node[label={[label distance=20mm]90:$M$}] (M) {};
\node[circle,minimum size=3mm,inner sep=0,fill=black] at (-0.5,0)(s1){};
\draw (0.5,0.5) circle (3) node[label={[label distance=30mm]90:$X$}]{};
\node[draw, rectangle, fit=(d2), inner sep=10pt,label={[font=\Huge]90:$d^{+}_{s_{1}\neg \wedge s_{2}}$}] (d2BOX) {};
\node[anchor=north west,xshift=5,yshift=-5] at (d2BOX.north west) {\Huge$\boxed\wedge$};
\end{scope}
\end{scope}
\begin{scope}[local bounding box=d3,shift={($(d2BOX.south) + (-0.5,-7)$)}]
	\scope[even odd rule]
	\clip \circleB;
	\clip \circleC;
	\fill[red] \circleB ;
	\endscope
\draw \circleB node[label={[label distance=8mm,xshift=-5]-90:$S_{1}$}] (S1){};
\draw \circleC node [label={[label distance=8mm,xshift=5]-90:$S_{2}$}] (S2){};
\draw (0.5,0) circle (2) node[label={[label distance=20mm]90:$M$}] (M) {};
\node[circle,minimum size=3mm,inner sep=0,fill=black] at (-0.5,0)(s1){};
\node[circle,minimum size=3mm,inner sep=0,fill=black] at (-2,0)(x1){};
\draw(x1)--(s1);
\draw (0.5,0.5) circle (3) node[label={[label distance=30mm]90:$X$}]{};
\node[draw, rectangle, fit=(d3), inner sep=10pt] (d3BOX) {};
\end{scope}
\begin{scope}[local bounding box=d4,shift={($(d3BOX.south) + (-0.5,-7)$)}]
	\scope[even odd rule]
	\clip \circleB;
	\clip \circleC;
	\fill[red] \circleB ;
	\endscope
\draw \circleB node[label={[label distance=8mm,xshift=-5]-90:$S_{1}$}] (S1){};
\draw \circleC node [label={[label distance=8mm,xshift=5]-90:$S_{2}$}] (S2){};
\draw (0.5,0) circle (2) node[label={[label distance=20mm]90:$M$}] (M) {};
\node[circle,minimum size=3mm,inner sep=0,fill=black] at (-0.5,0)(s1){};
\node[circle,minimum size=3mm,inner sep=0,fill=black] at (0.5,1.5)(m1){};
\node[circle,minimum size=3mm,inner sep=0,fill=black] at (-2,0)(x1){};
\draw(x1)--(s1);
\draw(m1)--(s1);
\draw (0.5,0.5) circle (3) node[label={[label distance=30mm]90:$X$}]{};
\node[draw, rectangle, fit=(d4), inner sep=10pt] (d4BOX) {};
\end{scope}
\begin{scope}[local bounding box=d5,shift={($(d4BOX.south) + (3,-7)$)}] 
	\scope[even odd rule]
	\clip \circleB;
	\clip \circleC;
	\fill[red] \circleB ;
	\endscope
\draw \circleB node[label={[label distance=8mm,xshift=-5]-90:$S_{1}$}] (S1){};
\draw \circleC node [label={[label distance=8mm,xshift=5]-90:$S_{2}$}] (S2){};
\draw (0.5,0) circle (2) node[label={[label distance=20mm]90:$M$}] (M) {};
\node[circle,minimum size=3mm,inner sep=0,fill=black] at (0.5,1.5)(m1){};
\draw (0.5,0.5) circle (3) node[label={[label distance=30mm]90:$X$}]{};
\begin{scope}[xshift=-70mm]
	\scope[even odd rule]
	\clip \circleB;
	\clip \circleC;
	\fill[red] \circleB ;
	\endscope
\draw \circleB node[label={[label distance=8mm,xshift=-5]-90:$S_{1}$}] (S1){};
\draw \circleC node [label={[label distance=8mm,xshift=5]-90:$S_{2}$}] (S2){};
\draw (0.5,0) circle (2) node[label={[label distance=20mm]90:$M$}] (M) {};
\node[circle,minimum size=3mm,inner sep=0,fill=black] at (-0.5,0)(s1){};
\node[circle,minimum size=3mm,inner sep=0,fill=black] at (-2,0)(x1){};
\draw(x1)--(s1);
\draw (0.5,0.5) circle (3) node[label={[label distance=30mm]90:$X$}]{};
\node[draw, rectangle, fit=(d5), inner sep=10pt] (d5BOX) {};
\node[anchor=north west,xshift=5,yshift=-5] at (d5BOX.north west) {\Huge$\boxed\vee$};
\end{scope}
\end{scope}
\begin{scope}[local bounding box=d6,shift={($(d5BOX.south) + (-10.5,-7)$)}]
	\scope[even odd rule]
	\clip \circleB;
	\clip \circleC;
	\fill[red] \circleB ;
	\endscope
\draw \circleB node[label={[label distance=8mm,xshift=-5]-90:$S_{1}$}] (S1){};
\draw \circleC node [label={[label distance=8mm,xshift=5]-90:$S_{2}$}] (S2){};
\draw (0.5,0) circle (2) node[label={[label distance=20mm]90:$M$}] (M) {};
\node[circle,minimum size=3mm,inner sep=0,fill=black] at (-0.5,0)(s1){};
\node[circle,minimum size=3mm,inner sep=0,fill=black] at (-2,0)(x1){};
\draw(x1)--(s1);
\draw (0.5,0.5) circle (3) node[label={[label distance=30mm]90:$X$}]{};
\node[draw, rectangle, fit=(d6), inner sep=10pt] (d6BOX) {};
\end{scope}
\begin{scope}[local bounding box=d7,shift={($(d5BOX.south) + (10,-7)$)}]
	\scope[even odd rule]
	\clip \circleB;
	\clip \circleC;
	\fill[red] \circleB ;
	\endscope
\draw \circleB node[label={[label distance=8mm,xshift=-5]-90:$S_{1}$}] (S1){};
\draw \circleC node [label={[label distance=8mm,xshift=5]-90:$S_{2}$}] (S2){};
\draw (0.5,0) circle (2) node[label={[label distance=20mm]90:$M$}] (M) {};
\node[circle,minimum size=3mm,inner sep=0,fill=black] at (0.5,1.5)(m1){};
\draw (0.5,0.5) circle (3) node[label={[label distance=30mm]90:$X$}]{};
\node[draw, rectangle, fit=(d7), inner sep=10pt] (d7BOX) {};
\end{scope}
\begin{scope}[local bounding box=d8,shift={($(d6BOX.south) + (-4,-7)$)}] 
	\scope[even odd rule]
	\clip \circleB;
	\clip \circleC;
	\fill[red] \circleB ;
	\endscope
\draw \circleB node[label={[label distance=8mm,xshift=-5]-90:$S_{1}$}] (S1){};
\draw \circleC node [label={[label distance=8mm,xshift=5]-90:$S_{2}$}] (S2){};
\draw (0.5,0) circle (2) node[label={[label distance=20mm]90:$M$}] (M) {};
\node[circle,minimum size=3mm,inner sep=0,fill=black] at (0.5,1.5)(m1){};
\node[circle,minimum size=3mm,inner sep=0,fill=black] at (-0.5,0)(s1){};
\node[circle,minimum size=3mm,inner sep=0,fill=black] at (-2,0)(x1){};
\draw(x1)--(s1);
\draw (0.5,0.5) circle (3) node[label={[label distance=30mm]90:$X$}]{};
\begin{scope}[xshift=70mm]
	\scope[even odd rule]
	\clip \circleB;
	\clip \circleC;
	\fill[red] \circleB ;
	\endscope
\draw \circleB node[label={[label distance=8mm,xshift=-5]-90:$S_{1}$}] (S1){};
\draw \circleC node [label={[label distance=8mm,xshift=5]-90:$S_{2}$}] (S2){};
\draw (0.5,0) circle (2) node[label={[label distance=20mm]90:$M$}] (M) {};
\node[circle,minimum size=3mm,inner sep=0,fill=black] at (0.5,1.5)(m1){};
\draw (0.5,0.5) circle (3) node[label={[label distance=30mm]90:$X$}]{};
\node[draw, rectangle, fit=(d8), inner sep=10pt] (d8BOX) {};
\node[anchor=north west,xshift=5,yshift=-5] at (d8BOX.north west) {\Huge$\boxed\wedge$};
\end{scope}
\end{scope}
\begin{scope}[local bounding box=d9,shift={($(d7BOX.south) + (-4,-7)$)}] 
	\scope[even odd rule]
	\clip \circleB;
	\clip \circleC;
	\fill[red] \circleB ;
	\endscope
\draw \circleB node[label={[label distance=8mm,xshift=-5]-90:$S_{1}$}] (S1){};
\draw \circleC node [label={[label distance=8mm,xshift=5]-90:$S_{2}$}] (S2){};
\draw (0.5,0) circle (2) node[label={[label distance=20mm]90:$M$}] (M) {};
\node[circle,minimum size=3mm,inner sep=0,fill=black] at (0.5,1.5)(m1){};
\node[circle,minimum size=3mm,inner sep=0,fill=black] at (-0.5,0)(s1){};
\node[circle,minimum size=3mm,inner sep=0,fill=black] at (-2,0)(x1){};
\draw(x1)--(s1);
\draw (0.5,0.5) circle (3) node[label={[label distance=30mm]90:$X$}]{};
\begin{scope}[xshift=70mm]
	\scope[even odd rule]
	\clip \circleB;
	\clip \circleC;
	\fill[red] \circleB ;
	\endscope
\draw \circleB node[label={[label distance=8mm,xshift=-5]-90:$S_{1}$}] (S1){};
\draw \circleC node [label={[label distance=8mm,xshift=5]-90:$S_{2}$}] (S2){};
\draw (0.5,0) circle (2) node[label={[label distance=20mm]90:$M$}] (M) {};
\node[circle,minimum size=3mm,inner sep=0,fill=black] at (0.5,1.5)(m1){};
\draw (0.5,0.5) circle (3) node[label={[label distance=30mm]90:$X$}]{};
\node[draw, rectangle, fit=(d9), inner sep=10pt] (d9BOX) {};
\node[anchor=north west,xshift=5,yshift=-5] at (d9BOX.north west) {\Huge$\boxed\wedge$};
\end{scope}
\end{scope}
\begin{scope}[local bounding box=d10,shift={($(d8BOX.south) + (-0.5,-7)$)}]
	\scope[even odd rule]
	\clip \circleB;
	\clip \circleC;
	\fill[red] \circleB ;
	\endscope
\draw \circleB node[label={[label distance=8mm,xshift=-5]-90:$S_{1}$}] (S1){};
\draw \circleC node [label={[label distance=8mm,xshift=5]-90:$S_{2}$}] (S2){};
\draw (0.5,0) circle (2) node[label={[label distance=20mm]90:$M$}] (M) {};
\node[circle,minimum size=3mm,inner sep=0,fill=black] at (0.5,1.5)(m1){};
\node[circle,minimum size=3mm,inner sep=0,fill=black] at (-0.5,0)(s1){};
\node[circle,minimum size=3mm,inner sep=0,fill=black] at (-2,0)(x1){};
\draw(s1)--(x1);
\draw (0.5,0.5) circle (3) node[label={[label distance=30mm]90:$X$}]{};
\node[draw, rectangle, fit=(d10), inner sep=10pt] (d10BOX) {};
\end{scope}
\begin{scope}[local bounding box=d11,shift={($(d9BOX.south) + (-0.5,-7)$)}]
	\scope[even odd rule]
	\clip \circleB;
	\clip \circleC;
	\fill[red] \circleB ;
	\endscope
\draw \circleB node[label={[label distance=8mm,xshift=-5]-90:$S_{1}$}] (S1){};
\draw \circleC node [label={[label distance=8mm,xshift=5]-90:$S_{2}$}] (S2){};
\draw (0.5,0) circle (2) node[label={[label distance=20mm]90:$M$}] (M) {};
\node[circle,minimum size=3mm,inner sep=0,fill=black] at (0.5,1.5)(m1){};
\node[circle,minimum size=3mm,inner sep=0,fill=black] at (-0.5,0)(s1){};
\node[circle,minimum size=3mm,inner sep=0,fill=black] at (-2,0)(x1){};
\draw(x1)--(s1);
\draw (0.5,0.5) circle (3) node[label={[label distance=30mm]90:$X$}]{};
\node[draw, rectangle, fit=(d11), inner sep=10pt] (d11BOX) {};
\end{scope}
\begin{scope}[local bounding box=d12,shift={($(d11BOX.south) + (-7.25,-7)$)}] 
	\scope[even odd rule]
	\clip \circleB;
	\clip \circleC;
	\fill[red] \circleB ;
	\endscope
\draw \circleB node[label={[label distance=8mm,xshift=-5]-90:$S_{1}$}] (S1){};
\draw \circleC node [label={[label distance=8mm,xshift=5]-90:$S_{2}$}] (S2){};
\draw (0.5,0) circle (2) node[label={[label distance=20mm]90:$M$}] (M) {};
\node[circle,minimum size=3mm,inner sep=0,fill=black] at (0.5,1.5)(m1){};
\node[circle,minimum size=3mm,inner sep=0,fill=black] at (-0.5,0)(s1){};
\node[circle,minimum size=3mm,inner sep=0,fill=black] at (-2,0)(x1){};
\draw(x1)--(s1);
\draw (0.5,0.5) circle (3) node[label={[label distance=30mm]90:$X$}]{};
\begin{scope}[xshift=-70mm]
	\scope[even odd rule]
	\clip \circleB;
	\clip \circleC;
	\fill[red] \circleB ;
	\endscope
\draw \circleB node[label={[label distance=8mm,xshift=-5]-90:$S_{1}$}] (S1){};
\draw \circleC node [label={[label distance=8mm,xshift=5]-90:$S_{2}$}] (S2){};
\draw (0.5,0) circle (2) node[label={[label distance=20mm]90:$M$}] (M) {};
\node[circle,minimum size=3mm,inner sep=0,fill=black] at (0.5,1.5)(m1){};
\node[circle,minimum size=3mm,inner sep=0,fill=black] at (-0.5,0)(s1){};
\node[circle,minimum size=3mm,inner sep=0,fill=black] at (-2,0)(x1){};;
\draw(x1)--(s1);
\draw (0.5,0.5) circle (3) node[label={[label distance=30mm]90:$X$}]{};
\node[draw, rectangle, fit=(d12), inner sep=10pt] (d12BOX) {};
\node[anchor=north west,xshift=5,yshift=-5] at (d12BOX.north west) {\Huge$\boxed\vee$};
\end{scope}
\end{scope}
\begin{scope}[local bounding box=d13,shift={($(d12BOX.south) + (-0.5,-7)$)}]
	\scope[even odd rule]
	\clip \circleB;
	\clip \circleC;
	\fill[red] \circleB ;
	\endscope
\draw \circleB node[label={[label distance=8mm,xshift=-5]-90:$S_{1}$}] (S1){};
\draw \circleC node [label={[label distance=8mm,xshift=5]-90:$S_{2}$}] (S2){};
\draw (0.5,0) circle (2) node[label={[label distance=20mm]90:$M$}] (M) {};
\node[circle,minimum size=3mm,inner sep=0,fill=black] at (0.5,1.5)(m1){};
\node[circle,minimum size=3mm,inner sep=0,fill=black] at (-0.5,0)(s1){};
\node[circle,minimum size=3mm,inner sep=0,fill=black] at (-2,0)(x1){};
\draw(x1)--(s1);
\draw (0.5,0.5) circle (3) node[label={[label distance=30mm]90:$X$}]{};
\node[draw, rectangle, fit=(d13), inner sep=10pt,label={[font=\Huge]-90:$d^{+}_{ s_{1}\neg s_{2}}$}] (d13BOX) {};
\end{scope}
\draw[thick,->,double equal sign distance](d2BOX)--node[right,label={[align=center,font=\Huge]180:\scshape{ConjunctionElimination}}]{}(d3BOX);
\draw[thick,->,double equal sign distance](d3BOX)--node[right,label={[align=center,font=\Huge]180:\scshape{AddFeet}}]{}(d4BOX);
\draw[thick,->,double equal sign distance](d4BOX)--node[right,label={[align=center,font=\Huge]180:\scshape{SplitSpider}}]{}(d5BOX);
\draw[thick,->,double equal sign distance,shorten >=20pt](d5BOX.south)--(d6BOX.north);
\draw[thick,->,double equal sign distance,shorten >=20pt](d5BOX.south)--(d7BOX.north);
\draw[thick,->,double equal sign distance](d6BOX)--node[right,label={[align=center,font=\Huge]180:\scshape{CopySpider}}]{}(d8BOX);
\draw[thick,->,double equal sign distance](d7BOX)--node[right,label={[align=center,font=\Huge]0:\scshape{CopySpider}}]{}(d9BOX);
\draw[thick,->,double equal sign distance](d8BOX)--node[right,label={[align=center,font=\Huge]180:\scshape{ConjunctionElimination}}]{}(d10BOX);
\draw[thick,->,double equal sign distance](d9BOX)--node[right,label={[align=center,font=\Huge]0:\scshape{ConjunctionElimination}}]{}(d11BOX);
\draw[thick,->,double equal sign distance,shorten >=20pt](d10BOX.south)--(d12BOX.north);
\draw[thick,->,double equal sign distance,shorten >=20pt](d11BOX.south)--(d12BOX.north);
\draw[thick,->,double equal sign distance](d12BOX)--node[right,label={[align=center,font=\Huge]0:\scshape{Idempotency}}]{}(d13BOX);
\end{tikzpicture}
\end{adjustbox}
}
\subfloat[$T_8: d_{s_{2}\wedge \neg s_{1}} \vdash_{\mathcal{R}} d^{+}_{s_{2} \neg s_{1}}$.\label{ProofTree8}]{%
\begin{adjustbox}{width=0.5\linewidth}
\begin{tikzpicture}[line width=0.5mm,font=\Large]
\def\circleB{(0,0) circle (1)}
\def\circleC{(1,0) circle (1)}
\def\circleA{(0.5,1) circle (1)}
\tikzset{fontscale/.style = {font=\relsize{#1}}}
%
\begin{scope}[local bounding box=d2,shift={($(d1BOX.south) + (-4,-7)$)}] 
	\scope[even odd rule]
	\clip \circleB;
	\clip \circleC;
	\fill[red] \circleB ;
	\endscope
\draw \circleB node[label={[label distance=8mm,xshift=-5]-90:$S_{1}$}] (S1){};
\draw \circleC node [label={[label distance=8mm,xshift=5]-90:$S_{2}$}] (S2){};
\draw (0.5,0) circle (2) node[label={[label distance=20mm]90:$M$}] (M) {};
\node[circle,minimum size=3mm,inner sep=0,fill=black] at (1.5,0)(s1){};
\node[circle,minimum size=3mm,inner sep=0,fill=black] at (3,0)(x1){};
\draw(x1)--(s1);
\draw (0.5,0.5) circle (3) node[label={[label distance=30mm]90:$X$}]{};
\begin{scope}[xshift=70mm]
	\scope[even odd rule]
	\clip \circleB;
	\clip \circleC;
	\fill[red] \circleB ;
	\endscope
\draw \circleB node[label={[label distance=8mm,xshift=-5]-90:$S_{1}$}] (S1){};
\draw \circleC node [label={[label distance=8mm,xshift=5]-90:$S_{2}$}] (S2){};
\draw (0.5,0) circle (2) node[label={[label distance=20mm]90:$M$}] (M) {};
\node[circle,minimum size=3mm,inner sep=0,fill=black] at (1.5,0)(s1){};
\draw (0.5,0.5) circle (3) node[label={[label distance=30mm]90:$X$}]{};
\node[draw, rectangle, fit=(d2), inner sep=10pt,label={[font=\Huge]90:$d^{+}_{s_{2}\wedge\neg s_{1}}$}] (d2BOX) {};
\node[anchor=north west,xshift=5,yshift=-5] at (d2BOX.north west) {\Huge$\boxed\wedge$};
\end{scope}
\end{scope}
\begin{scope}[local bounding box=d3,shift={($(d2BOX.south) + (-0.5,-7)$)}]
	\scope[even odd rule]
	\clip \circleB;
	\clip \circleC;
	\fill[red] \circleB ;
	\endscope
\draw \circleB node[label={[label distance=8mm,xshift=-5]-90:$S_{1}$}] (S1){};
\draw \circleC node [label={[label distance=8mm,xshift=5]-90:$S_{2}$}] (S2){};
\draw (0.5,0) circle (2) node[label={[label distance=20mm]90:$M$}] (M) {};
\node[circle,minimum size=3mm,inner sep=0,fill=black] at (1.5,0)(s1){};
\node[circle,minimum size=3mm,inner sep=0,fill=black] at (3,0)(x1){};
\draw(x1)--(s1);
\draw (0.5,0.5) circle (3) node[label={[label distance=30mm]90:$X$}]{};
\node[draw, rectangle, fit=(d3), inner sep=10pt] (d3BOX) {};
\end{scope}
\begin{scope}[local bounding box=d4,shift={($(d3BOX.south) + (-0.5,-7)$)}]
	\scope[even odd rule]
	\clip \circleB;
	\clip \circleC;
	\fill[red] \circleB ;
	\endscope
\draw \circleB node[label={[label distance=8mm,xshift=-5]-90:$S_{1}$}] (S1){};
\draw \circleC node [label={[label distance=8mm,xshift=5]-90:$S_{2}$}] (S2){};
\draw (0.5,0) circle (2) node[label={[label distance=20mm]90:$M$}] (M) {};
\node[circle,minimum size=3mm,inner sep=0,fill=black] at (1.5,0)(s1){};
\node[circle,minimum size=3mm,inner sep=0,fill=black] at (0.5,1.5)(m1){};
\node[circle,minimum size=3mm,inner sep=0,fill=black] at (3,0)(x1){};
\draw(x1)--(s1);
\draw(m1)--(s1);
\draw (0.5,0.5) circle (3) node[label={[label distance=30mm]90:$X$}]{};
\node[draw, rectangle, fit=(d4), inner sep=10pt] (d4BOX) {};
\end{scope}
\begin{scope}[local bounding box=d5,shift={($(d4BOX.south) + (3,-7)$)}] 
	\scope[even odd rule]
	\clip \circleB;
	\clip \circleC;
	\fill[red] \circleB ;
	\endscope
\draw \circleB node[label={[label distance=8mm,xshift=-5]-90:$S_{1}$}] (S1){};
\draw \circleC node [label={[label distance=8mm,xshift=5]-90:$S_{2}$}] (S2){};
\draw (0.5,0) circle (2) node[label={[label distance=20mm]90:$M$}] (M) {};
\node[circle,minimum size=3mm,inner sep=0,fill=black] at (0.5,1.5)(m1){};
\draw (0.5,0.5) circle (3) node[label={[label distance=30mm]90:$X$}]{};
\begin{scope}[xshift=-70mm]
	\scope[even odd rule]
	\clip \circleB;
	\clip \circleC;
	\fill[red] \circleB ;
	\endscope
\draw \circleB node[label={[label distance=8mm,xshift=-5]-90:$S_{1}$}] (S1){};
\draw \circleC node [label={[label distance=8mm,xshift=5]-90:$S_{2}$}] (S2){};
\draw (0.5,0) circle (2) node[label={[label distance=20mm]90:$M$}] (M) {};
\node[circle,minimum size=3mm,inner sep=0,fill=black] at (1.5,0)(s1){};
\node[circle,minimum size=3mm,inner sep=0,fill=black] at (3,0)(x1){};
\draw(x1)--(s1);
\draw (0.5,0.5) circle (3) node[label={[label distance=30mm]90:$X$}]{};
\node[draw, rectangle, fit=(d5), inner sep=10pt] (d5BOX) {};
\node[anchor=north west,xshift=5,yshift=-5] at (d5BOX.north west) {\Huge$\boxed\vee$};
\end{scope}
\end{scope}
\begin{scope}[local bounding box=d6,shift={($(d5BOX.south) + (-10.5,-7)$)}]
	\scope[even odd rule]
	\clip \circleB;
	\clip \circleC;
	\fill[red] \circleB ;
	\endscope
\draw \circleB node[label={[label distance=8mm,xshift=-5]-90:$S_{1}$}] (S1){};
\draw \circleC node [label={[label distance=8mm,xshift=5]-90:$S_{2}$}] (S2){};
\draw (0.5,0) circle (2) node[label={[label distance=20mm]90:$M$}] (M) {};
\node[circle,minimum size=3mm,inner sep=0,fill=black] at (1.5,0)(s1){};
\node[circle,minimum size=3mm,inner sep=0,fill=black] at (3,0)(x1){};
\draw(x1)--(s1);
\draw (0.5,0.5) circle (3) node[label={[label distance=30mm]90:$X$}]{};
\node[draw, rectangle, fit=(d6), inner sep=10pt] (d6BOX) {};
\end{scope}
\begin{scope}[local bounding box=d7,shift={($(d5BOX.south) + (10,-7)$)}]
	\scope[even odd rule]
	\clip \circleB;
	\clip \circleC;
	\fill[red] \circleB ;
	\endscope
\draw \circleB node[label={[label distance=8mm,xshift=-5]-90:$S_{1}$}] (S1){};
\draw \circleC node [label={[label distance=8mm,xshift=5]-90:$S_{2}$}] (S2){};
\draw (0.5,0) circle (2) node[label={[label distance=20mm]90:$M$}] (M) {};
\node[circle,minimum size=3mm,inner sep=0,fill=black] at (0.5,1.5)(m1){};
\draw (0.5,0.5) circle (3) node[label={[label distance=30mm]90:$X$}]{};
\node[draw, rectangle, fit=(d7), inner sep=10pt] (d7BOX) {};
\end{scope}
\begin{scope}[local bounding box=d8,shift={($(d6BOX.south) + (-4,-7)$)}] 
	\scope[even odd rule]
	\clip \circleB;
	\clip \circleC;
	\fill[red] \circleB ;
	\endscope
\draw \circleB node[label={[label distance=8mm,xshift=-5]-90:$S_{1}$}] (S1){};
\draw \circleC node [label={[label distance=8mm,xshift=5]-90:$S_{2}$}] (S2){};
\draw (0.5,0) circle (2) node[label={[label distance=20mm]90:$M$}] (M) {};
\node[circle,minimum size=3mm,inner sep=0,fill=black] at (0.5,1.5)(m1){};
\node[circle,minimum size=3mm,inner sep=0,fill=black] at (1.5,0)(s1){};
\node[circle,minimum size=3mm,inner sep=0,fill=black] at (3,0)(x1){};
\draw(x1)--(s1);
\draw (0.5,0.5) circle (3) node[label={[label distance=30mm]90:$X$}]{};
\begin{scope}[xshift=70mm]
	\scope[even odd rule]
	\clip \circleB;
	\clip \circleC;
	\fill[red] \circleB ;
	\endscope
\draw \circleB node[label={[label distance=8mm,xshift=-5]-90:$S_{1}$}] (S1){};
\draw \circleC node [label={[label distance=8mm,xshift=5]-90:$S_{2}$}] (S2){};
\draw (0.5,0) circle (2) node[label={[label distance=20mm]90:$M$}] (M) {};
\node[circle,minimum size=3mm,inner sep=0,fill=black] at (0.5,1.5)(m1){};
\draw (0.5,0.5) circle (3) node[label={[label distance=30mm]90:$X$}]{};
\node[draw, rectangle, fit=(d8), inner sep=10pt] (d8BOX) {};
\node[anchor=north west,xshift=5,yshift=-5] at (d8BOX.north west) {\Huge$\boxed\wedge$};
\end{scope}
\end{scope}
\begin{scope}[local bounding box=d9,shift={($(d7BOX.south) + (-4,-7)$)}] 
	\scope[even odd rule]
	\clip \circleB;
	\clip \circleC;
	\fill[red] \circleB ;
	\endscope
\draw \circleB node[label={[label distance=8mm,xshift=-5]-90:$S_{1}$}] (S1){};
\draw \circleC node [label={[label distance=8mm,xshift=5]-90:$S_{2}$}] (S2){};
\draw (0.5,0) circle (2) node[label={[label distance=20mm]90:$M$}] (M) {};
\node[circle,minimum size=3mm,inner sep=0,fill=black] at (0.5,1.5)(m1){};
\node[circle,minimum size=3mm,inner sep=0,fill=black] at (1.5,0)(s1){};
\node[circle,minimum size=3mm,inner sep=0,fill=black] at (3,0)(x1){};
\draw(x1)--(s1);
\draw (0.5,0.5) circle (3) node[label={[label distance=30mm]90:$X$}]{};
\begin{scope}[xshift=70mm]
	\scope[even odd rule]
	\clip \circleB;
	\clip \circleC;
	\fill[red] \circleB ;
	\endscope
\draw \circleB node[label={[label distance=8mm,xshift=-5]-90:$S_{1}$}] (S1){};
\draw \circleC node [label={[label distance=8mm,xshift=5]-90:$S_{2}$}] (S2){};
\draw (0.5,0) circle (2) node[label={[label distance=20mm]90:$M$}] (M) {};
\node[circle,minimum size=3mm,inner sep=0,fill=black] at (0.5,1.5)(m1){};
\draw (0.5,0.5) circle (3) node[label={[label distance=30mm]90:$X$}]{};
\node[draw, rectangle, fit=(d9), inner sep=10pt] (d9BOX) {};
\node[anchor=north west,xshift=5,yshift=-5] at (d9BOX.north west) {\Huge$\boxed\wedge$};
\end{scope}
\end{scope}
\begin{scope}[local bounding box=d10,shift={($(d8BOX.south) + (-0.5,-7)$)}]
	\scope[even odd rule]
	\clip \circleB;
	\clip \circleC;
	\fill[red] \circleB ;
	\endscope
\draw \circleB node[label={[label distance=8mm,xshift=-5]-90:$S_{1}$}] (S1){};
\draw \circleC node [label={[label distance=8mm,xshift=5]-90:$S_{2}$}] (S2){};
\draw (0.5,0) circle (2) node[label={[label distance=20mm]90:$M$}] (M) {};
\node[circle,minimum size=3mm,inner sep=0,fill=black] at (0.5,1.5)(m1){};
\node[circle,minimum size=3mm,inner sep=0,fill=black] at (1.5,0)(s1){};
\node[circle,minimum size=3mm,inner sep=0,fill=black] at (3,0)(x1){};
\draw(s1)--(x1);
\draw (0.5,0.5) circle (3) node[label={[label distance=30mm]90:$X$}]{};
\node[draw, rectangle, fit=(d10), inner sep=10pt] (d10BOX) {};
\end{scope}
\begin{scope}[local bounding box=d11,shift={($(d9BOX.south) + (-0.5,-7)$)}]
	\scope[even odd rule]
	\clip \circleB;
	\clip \circleC;
	\fill[red] \circleB ;
	\endscope
\draw \circleB node[label={[label distance=8mm,xshift=-5]-90:$S_{1}$}] (S1){};
\draw \circleC node [label={[label distance=8mm,xshift=5]-90:$S_{2}$}] (S2){};
\draw (0.5,0) circle (2) node[label={[label distance=20mm]90:$M$}] (M) {};
\node[circle,minimum size=3mm,inner sep=0,fill=black] at (0.5,1.5)(m1){};
\node[circle,minimum size=3mm,inner sep=0,fill=black] at (1.5,0)(s1){};
\node[circle,minimum size=3mm,inner sep=0,fill=black] at (3,0)(x1){};
\draw(x1)--(s1);
\draw (0.5,0.5) circle (3) node[label={[label distance=30mm]90:$X$}]{};
\node[draw, rectangle, fit=(d11), inner sep=10pt] (d11BOX) {};
\end{scope}
\begin{scope}[local bounding box=d12,shift={($(d11BOX.south) + (-7.25,-7)$)}] 
	\scope[even odd rule]
	\clip \circleB;
	\clip \circleC;
	\fill[red] \circleB ;
	\endscope
\draw \circleB node[label={[label distance=8mm,xshift=-5]-90:$S_{1}$}] (S1){};
\draw \circleC node [label={[label distance=8mm,xshift=5]-90:$S_{2}$}] (S2){};
\draw (0.5,0) circle (2) node[label={[label distance=20mm]90:$M$}] (M) {};
\node[circle,minimum size=3mm,inner sep=0,fill=black] at (0.5,1.5)(m1){};
\node[circle,minimum size=3mm,inner sep=0,fill=black] at (1.5,0)(s1){};
\node[circle,minimum size=3mm,inner sep=0,fill=black] at (3,0)(x1){};
\draw(x1)--(s1);
\draw (0.5,0.5) circle (3) node[label={[label distance=30mm]90:$X$}]{};
\begin{scope}[xshift=-70mm]
	\scope[even odd rule]
	\clip \circleB;
	\clip \circleC;
	\fill[red] \circleB ;
	\endscope
\draw \circleB node[label={[label distance=8mm,xshift=-5]-90:$S_{1}$}] (S1){};
\draw \circleC node [label={[label distance=8mm,xshift=5]-90:$S_{2}$}] (S2){};
\draw (0.5,0) circle (2) node[label={[label distance=20mm]90:$M$}] (M) {};
\node[circle,minimum size=3mm,inner sep=0,fill=black] at (0.5,1.5)(m1){};
\node[circle,minimum size=3mm,inner sep=0,fill=black] at (1.5,0)(s1){};
\node[circle,minimum size=3mm,inner sep=0,fill=black] at (3,0)(x1){};;
\draw(x1)--(s1);
\draw (0.5,0.5) circle (3) node[label={[label distance=30mm]90:$X$}]{};
\node[draw, rectangle, fit=(d12), inner sep=10pt] (d12BOX) {};
\node[anchor=north west,xshift=5,yshift=-5] at (d12BOX.north west) {\Huge$\boxed\vee$};
\end{scope}
\end{scope}
\begin{scope}[local bounding box=d13,shift={($(d12BOX.south) + (-0.5,-7)$)}]
	\scope[even odd rule]
	\clip \circleB;
	\clip \circleC;
	\fill[red] \circleB ;
	\endscope
\draw \circleB node[label={[label distance=8mm,xshift=-5]-90:$S_{1}$}] (S1){};
\draw \circleC node [label={[label distance=8mm,xshift=5]-90:$S_{2}$}] (S2){};
\draw (0.5,0) circle (2) node[label={[label distance=20mm]90:$M$}] (M) {};
\node[circle,minimum size=3mm,inner sep=0,fill=black] at (0.5,1.5)(m1){};
\node[circle,minimum size=3mm,inner sep=0,fill=black] at (1.5,0)(s1){};
\node[circle,minimum size=3mm,inner sep=0,fill=black] at (3,0)(x1){};
\draw(x1)--(s1);
\draw (0.5,0.5) circle (3) node[label={[label distance=30mm]90:$X$}]{};
\node[draw, rectangle, fit=(d13), inner sep=10pt,label={[font=\Huge]-90:$d^{+}_{ s_{2}\neg s_{1}}$}] (d13BOX) {};
\end{scope}
\draw[thick,->,double equal sign distance](d2BOX)--node[right,label={[align=center,font=\Huge]180:\scshape{ConjunctionElimination}}]{}(d3BOX);
\draw[thick,->,double equal sign distance](d3BOX)--node[right,label={[align=center,font=\Huge]180:\scshape{AddFeet}}]{}(d4BOX);
\draw[thick,->,double equal sign distance](d4BOX)--node[right,label={[align=center,font=\Huge]180:\scshape{SplitSpider}}]{}(d5BOX);
\draw[thick,->,double equal sign distance,shorten >=20pt](d5BOX.south)--(d6BOX.north);
\draw[thick,->,double equal sign distance,shorten >=20pt](d5BOX.south)--(d7BOX.north);
\draw[thick,->,double equal sign distance](d6BOX)--node[right,label={[align=center,font=\Huge]180:\scshape{CopySpider}}]{}(d8BOX);
\draw[thick,->,double equal sign distance](d7BOX)--node[right,label={[align=center,font=\Huge]0:\scshape{CopySpider}}]{}(d9BOX);
\draw[thick,->,double equal sign distance](d8BOX)--node[right,label={[align=center,font=\Huge]180:\scshape{ConjunctionElimination}}]{}(d10BOX);
\draw[thick,->,double equal sign distance](d9BOX)--node[right,label={[align=center,font=\Huge]0:\scshape{ConjunctionElimination}}]{}(d11BOX);
\draw[thick,->,double equal sign distance,shorten >=20pt](d10BOX.south)--(d12BOX.north);
\draw[thick,->,double equal sign distance,shorten >=20pt](d11BOX.south)--(d12BOX.north);
\draw[thick,->,double equal sign distance](d12BOX)--node[right,label={[align=center,font=\Huge]0:\scshape{Idempotency}}]{}(d13BOX);
\end{tikzpicture}
\end{adjustbox}
}
\caption{Proof trees witnessing meta-term construction of Positive Deixis and Negative Deixis. Each tree corresponds to a derivation in the inference system $\mathcal{R}$, transforming an initial diagram representing a seme pair into a meta-term.} 
\label{ProofTrees7-8}
\end{figure}

\begin{proposition}[Derivation of Greimasian meta-terms]\label{The:MetaTerms}
Let 
\[
(a,b)\in\{(s_1,s_2),(\neg s_1,\neg s_2),(s_1,\neg s_2),(s_2,\neg s_1)\},
\]
and let $d_a$ and $d_b$ be unitary spider diagrams representing the corresponding semic configurations in the language with contours $L=\{S_1,S_2,M,X\}$, where $S_1,S_2 \subseteq M \subseteq X$ and $S_1 \cap S_2 = \varnothing$. Let $d_{a \wedge b}$ denote the diagrammatic combination of $d_a$ and $d_b$ via the rule \textsc{Combine}, when defined. Then there exists a derivation in the inference system $\mathcal R$,
\[
d_{a \wedge b} \;\vdash_{\mathcal R}\; d^{+}_{ab},
\]
such that the resulting diagram $d^{+}_{ab}$ satisfies the following:
\begin{enumerate}
    \item it contains witnesses for the two input semic configurations $a$ and $b$,
    \item it contains a witness whose habitat lies within $M$,
    \item the witness in $M$ is distinct from the input witnesses and represents the
    sememic actualisation of the compound configuration,
    \item the shading condition $S_1 \cap S_2 = \varnothing$ is preserved.
\end{enumerate}

The derivation proceeds by a uniform schema consisting of applications of \textsc{Combine}, \textsc{AddFeet}, \textsc{SplitSpider}, and \textsc{Idempotency}, together with structural normalisation steps involving \textsc{ConjunctionElimination} and \textsc{CopySpider} where required. These additional steps ensure that all intermediate diagrams satisfy the $\alpha$-diagram condition (Definition~\ref{def:AlphaDiagram}), in particular in cases where the input configurations contain spiders whose habitats span multiple zones. The four admissible pairs yield the canonical Greimasian meta-terms:
\[
\begin{aligned}
&d^{+}_{s_1 s_2} \quad\emph{(complex term $S$)},\\
&d^{+}_{\neg s_1 \neg s_2} \quad\emph{(neutral term $\bar S$)},\\
&d^{+}_{s_1 \neg s_2} \quad\emph{(Positive Deixis)},\\
&d^{+}_{s_2 \neg s_1} \quad\emph{(Negative Deixis)}.
\end{aligned}
\]
\end{proposition}

\begin{proof}[Sketch of proof of Proposition~\ref{The:MetaTerms}]
Fix a pair $(a,b)$. We begin with the diagram $d_{a \wedge b}$, representing the conjunctive configuration of the two input semes. When both $d_a$ and $d_b$ are $\alpha$-diagrams, $d_{a \wedge b}$ is obtained directly using \textsc{Combine}. In cases where one or both input diagrams contain spiders whose habitats span multiple zones (as in diagrammatic negations), we first apply \textsc{ConjunctionElimination} and \textsc{CopySpider} to obtain equivalent $\alpha$-diagrams to which \textsc{Combine} is applicable. From this normalised configuration, we apply \textsc{AddFeet} and \textsc{SplitSpider} to expand and refine spider habitats, thereby constructing intermediate diagrams in which the contributions of the two input configurations are jointly represented while preserving their distinctness. A further application of \textsc{Combine} introduces a new spider whose habitat lies within $M$, representing the sememic actualisation of the compound configuration. Finally, \textsc{Idempotency} removes redundant disjunctive structure, yielding a diagram $d^{+}_{ab}$ that contains the original witnesses together with a new witness in $M$, while preserving all shading constraints. The derivation schema is uniform across the four admissible pairs. The proof trees $T_5$–$T_8$ (Figures~\ref{ProofTrees1-2}--\ref{ProofTrees7-8}) provide concrete instances, with additional applications of \textsc{ConjunctionElimination} required in the derivation of the neutral term to account for the absence of a shared determinate habitat between the input configurations.
\end{proof}


\begin{remark}
Proposition~\ref{The:MetaTerms} makes precise the interpretation of the Greimasian operation `+' as a derivational construction rather than a logical or set-theoretic combination. In particular, the derivation schema is uniform in the choice of input pair, and thus applies not only to the four canonical pairs of the semiotic square, but to any conjunctive pair of semic configurations. This yields additional derived configurations such as:
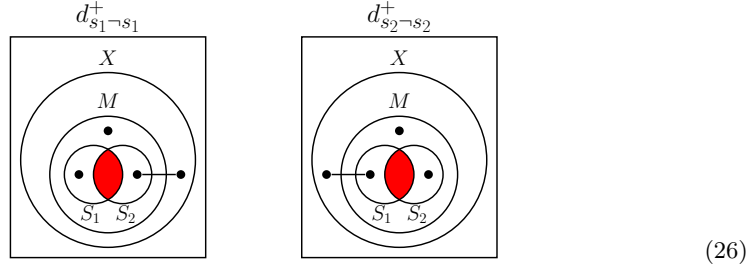
\begin{align}
\resizebox{0.5\textwidth}{!}{
\begin{tikzpicture}[line width=0.5mm,font=\Large]
\begin{scope}[local bounding box=d1]
	\scope[even odd rule]
	\clip \circleB;
	\clip \circleC;
	\fill[red] \circleB ;
	\endscope
\draw \circleB node[label={[label distance=8mm,xshift=-3,font=\huge]-90:$S_{1}$}] (S1){};
\draw \circleC node [label={[label distance=8mm,xshift=3,font=\huge]-90:$S_{2}$}] (S2){};
\draw (0.5,0) circle (2) node[label={[label distance=20mm,font=\huge]90:$M$}] (M) {};
\node[circle,minimum size=3mm,inner sep=0,fill=black] at (0.5,1.5)(m1){};
\node[circle,minimum size=3mm,inner sep=0,fill=black] at (1.5,0)(s2){};
\node[circle,minimum size=3mm,inner sep=0,fill=black] at (-0.5,0)(s1){};
\node[circle,minimum size=3mm,inner sep=0,fill=black] at (3,0)(x2){};
\draw(x2)--(s2);){};
\draw (0.5,0.5) circle (3) node[label={[label distance=30mm,font=\huge]90:$X$}]{};
\node[draw, rectangle, fit=(d1), inner sep=10pt, label={[font=\Huge]$d^{+}_{s_{1} \neg s_{1}}$}] (d1rec) {};
\end{scope}
\begin{scope}[local bounding box=d2,shift={(10,0)}] 
	\scope[even odd rule]
	\clip \circleB;
	\clip \circleC;
	\fill[red] \circleB ;
	\endscope
\draw \circleB node[label={[label distance=8mm,xshift=-3,font=\huge]-90:$S_{1}$}] (S1){};
\draw \circleC node [label={[label distance=8mm,xshift=3,font=\huge]-90:$S_{2}$}] (S2){};
\draw (0.5,0) circle (2) node[label={[label distance=20mm,font=\huge]90:$M$}] (M) {};
\node[circle,minimum size=3mm,inner sep=0,fill=black] at (0.5,1.5)(m1){};
\node[circle,minimum size=3mm,inner sep=0,fill=black] at (-0.5,0)(s1){};
\node[circle,minimum size=3mm,inner sep=0,fill=black] at (1.5,0)(s2){};
\node[circle,minimum size=3mm,inner sep=0,fill=black] at (-2,0)(x2){};
\draw(x2)--(s1);
\draw (0.5,0.5) circle (3) node[label={[label distance=30mm,font=\huge]90:$X$}]{};
\node[draw, rectangle, fit=(d2), inner sep=10pt, label={[font=\Huge]$d^{+}_{s_{2} \neg s_{2}}$}] (d2rec){};
\end{scope}
\end{tikzpicture}
}
\end{align}
which are obtained from internally contradictory inputs via: 
\begin{align*}
T_{9}: d_{s_{1}\wedge \neg s_{1}} \vdash_{\mathcal{R}} d^{+}_{s_{1} \neg s_{1}},\\
T_{10}: d_{s_{2}\wedge \neg s_{2}} \vdash_{\mathcal{R}} d^{+}_{s_{2} \neg s_{2}},
\end{align*}
and have been designated the positive schema $(s_{1},\neg s_{1})$, and negative schema $(s_{2}, \neg s_{2})$ by \citet[page 309]{GreimasCourtes1982}. Our derivation shows that the inference system $\mathcal R$ is thus closed under binary composition: the Greimasian meta-terms arise as a distinguished subset of a larger class of derivable sememic configurations. In contrast to the classical semiotic square, where the positive and negative schemas are typically excluded in accordance with the principle of non-contradiction, our diagrammatic setting permits their construction as well-defined semantic objects. From a proof-theoretic perspective, these additional terms correspond to derivations in which incompatible habitats (i.e. $d_{1},d_{2}$ and $d_{3},d_{4}$) are preserved at the level of input witnesses, while a new witness in $M$ mediates their joint representation. Such configurations have been noted in semiotic analysis such as those found in Majorcan folk tales ``\textit{Axio era y no era}'' (``it was and it was not'') by \citet[page 239]{Jakobson1963}.
\end{remark}

\subsection{A proof-theoretic semantics of the meta-terms}

In order to motivate our semio-logical interpretation of the derived meta-terms from the witnessed proof trees, we first provide a semiotic square by H\'{e}bert that is based on one of \citet[page 200]{Floch1985} that uses a contrary seme pair of $s_{1}$ = /masculine/, $s_{2}$ = /feminine/.

\begin{example}\label{HebertExample}

\citet[page 83]{Hebert2020} uses the following abstract Greimas square as a tool for investigating lexicalisation and homogeneity at each meta-term position:

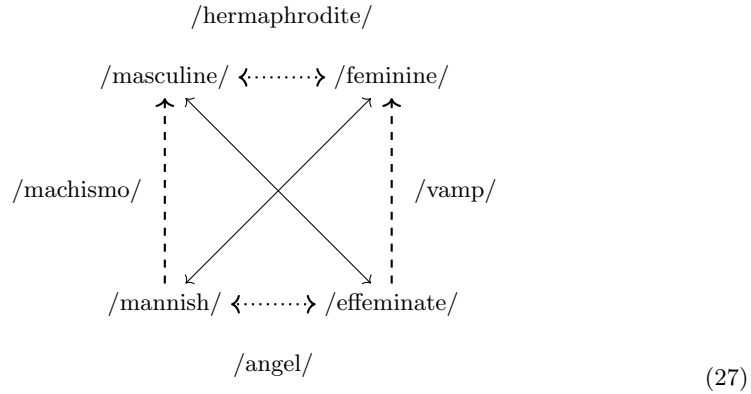
\begin{equation}
\begin{tikzpicture}[scale=0.75]
\node at (0,0)(s1){/masculine/};
\node at (4,0)(s2){/feminine/};
\node (not-s2) at (0,-4){/mannish/};
\node at (4,-4)(not-s1){/effeminate/};
\draw[thick,dotted,<->] (s1) -- node[label={[label distance=4mm]90:/hermaphrodite/}] {} (s2);
\draw[thick,dotted,<->]  (not-s1) -- node[label={[label distance=4mm]-90:/angel/}] {}(not-s2);
\draw[thick,dashed,->] (not-s2)--node[left,xshift=-5]{/machismo/}(s1);
\draw[thick,dashed,->] (not-s1)--node[right,xshift=5]{/vamp/}(s2);
\draw[<->](s1)--(not-s1);
\draw[<->](s2)--(not-s2);
\end{tikzpicture}
\end{equation}

Like the Greimas example given previously, H\'{e}bert constructs the meta-terms as described in clause (2) of Proposition~\ref{The:MetaTerms} using the semic operator $`+'$ so that at the $\mathbf{S}$ and $\mathbf{\bar{S}}$ positions, there are opposing terms that `combine' the semes. The resulting sememes function at the figurative level to describe anthropomorphic actants, while the Positive Dexies and Negative Dexies function at the axiological level in which the compound terms /machismo/ and /vamp/ (in the sense of ultra-feminine or seductive) are embodiments of the realising modality of `being'. Here, H\'{e}bert notes that because the square is abstract (i.e., not linked to a particular source text) he ``chose to represent only the `natural', `spontaneous' states of masculinity/femininity (in a general sense, that is not just biological), despite having to refer to unreal beings (angels).'' \citep[page 85]{Hebert2020}
\end{example}

We now relate the proof-theoretic construction of Proposition~\ref{The:MetaTerms} to its semantic interpretation in the semiotic square. Consider the meta-term $\mathbf{S}$ we derived through
\[
T_{5}: d_{s_{1}\wedge s_{2}} \vdash_{\mathcal{R}} d^{+}_{s_{1}s_{2}},
\]
for which we give its proof tree in Figure~\ref{ProofTree5}. Here, we have accommodated for the inherent structure of a sememe as a collection of semes by introducing a spider into $M$ via the inference rule \textsc{AddFeet}. The contour accounts for instances that associate to a sememe, but does not merely aggregate the semes $S_1$ and $S_2$. Rather, it induces a new semantic configuration in which their previously disjoint realisations are mediated at a higher level. Within a Greimassian semantics, this equates to a formalisation of how the meta-term $\mathbf{S}$ is a hyponym to the lower-level terms of the complex axis. In effect, the identification of the meta-term in $M$ is the taxonomic organisation of $s_{2}$ and $s_{2}$ into the complex axis $\mathfrak{c}$. 

In H\'{e}bert's square given in Example~\ref{HebertExample}, the syntactic interpretation indicates that although the two semes /masculine/ and /feminine/ are disjoint, and therefore defined via distinct semic categories, at the $\mathbf{S}$ position in the square, there exists the term /hermaphrodite/ in $M$, such that together with /masculine/ ($s_{1}$) and /feminine/ ($s_{2}$), naturally form a subset collection in $X$, the semantic universe. Here, H\'{e}bert suggest that the distinction between terms and meta-terms ``depends on whether [the] two terms judged as distinct and different are apprehended successively or simultaneously.'' \citep[page 88]{Hebert2020} Thus, although we have $Z^{\ast}(d_{1},d_{3})=\{(\{S_1,S_2\},\varnothing)\}$, the existence of a spider in $M$ expresses the formation of a sememe as a unifying semantic construct---i.e., it simultaneously expresses /masculine and /feminine/. Accordingly, we interpret a sememe in the Greimas square not as a `sum' of semes, but as a new semantic object constructed via derivation. The introduction of a spider in $M$ does not collapse the incompatibility of $S_1$ and $S_2$, but mediates their joint representation at a higher level of the diagrammatic hierarchy. In this sense, the operation ‘+’ is realised not as aggregation, but as a diagrammatic lift from a conjunctive seme-pair to a new configuration whose witness lies in $M$.

The interpretation given above for the complex term $\mathbf{S}$ extends uniformly to the remaining meta-terms derived in Proposition~\ref{The:MetaTerms}. In each case, the proof-theoretic construction realises the operation `$+$' as a derivational lift in which a new witness in $M$ mediates between the input configurations while preserving their structural constraints. Consider first the Positive Deixis, obtained via the derivation
\[
T_{7} : d_{s_1 \wedge \neg s_2} \;\vdash_{\mathcal R}\; d^{+}_{s_1 \neg s_2}.
\]
At the proof-theoretic level, this construction combines a determinate seme configuration $s_1$ with a virtualised configuration $\neg s_2$, whose witness ranges over multiple zones (cf. Definition~\ref{Def:Negation}). The derivation proceeds by refining and recombining these habitats so as to introduce a new spider in $M$, whose existence depends on both the actualisation of $s_1$ and the indeterminacy of $\neg s_2$. Unlike the case of the complex term $\mathbf{S}$, the inputs here are not symmetric: one is fully realised, while the other remains partially underdetermined. Semantically, this asymmetry corresponds to the axiological orientation of the Positive Deixis. The resulting sememe does not merely combine $s_1$ and $\neg s_2$, but resolves the indeterminacy of $\neg s_2$ in the direction of $s_1$. In this sense, the derived witness in $M$ may be interpreted as a configuration in which the virtualised opposition is stabilised through a dominant semantic pole. This aligns with the interpretation of the Positive Deixis as expressing a realised or actualised modality: e.g. being /machismo/ in the H\'{e}bert square (Example~\ref{HebertExample}) or identifying a /chauvinist/ presence. The diagrammatic implication of $\neg s_{2} \rightarrow s_{1} \cong d_{7} \Rightarrow_{\text{diag}} d_{8}$, and its covering hyponymic meta-term of the Positive Deixis operates in the thymic category of euphoria in which the underlying contrary relation is preserved but selectively resolved.

A dual analysis applies to the Negative Deixis, obtained via
\[
T_{8} : d_{s_2 \wedge \neg s_1} \;\vdash_{\mathcal R}\; d^{+}_{s_2 \neg s_1}.
\]
Here, the roles of $s_1$ and $s_2$ are reversed: the derivation combines a determinate instance of $s_2$ with the indeterminate configuration $\neg s_1$, yielding a new witness in $M$ that mediates their interaction. As in the previous case, the derivation preserves the structural incompatibility of $S_1$ and $S_2$, while introducing a higher-level configuration that resolves the virtualised component in favour of $s_2$. The resulting sememe thus encodes an opposing axiological orientation (the thymic category of disphoria), corresponding to what is traditionally identified as the Negative Deixis. In the H\'{e}bert square example, this encapsulates the meta-term as ``intensifying a term [/vamp/] by aﬃrming a semantic value and simultaneously negating the opposite of that value.'' \citep[page 87]{Hebert2020}

\begin{figure}[h!]
\centering
\resizebox{0.85\textwidth}{!}{
\begin{tikzpicture}[line width=0.5mm,font=\Large]
\def\circleB{(0,0) circle (1)}
\def\circleC{(1,0) circle (1)}
\def\circleA{(0.5,1) circle (1)}
\begin{scope}[local bounding box=d1]
	\scope[even odd rule]
	\clip \circleB;
	\clip \circleC;
	\fill[red] \circleB ;
	\endscope
\draw \circleB node[label={[label distance=8mm,xshift=-5]-90:$S_{1}$}] (S1){};
\draw \circleC node [label={[label distance=8mm,xshift=5]-90:$S_{2}$}] (S2){};
\draw (0.5,0) circle (2) node[label={[label distance=20mm]90:$M$}] (M) {};
\node[circle,minimum size=3mm,inner sep=0,fill=black] at (-0.5,0)(s1){};
\draw (0.5,0.5) circle (3) node[label={[label distance=30mm]90:$X$}]{};
\node[draw, rectangle, fit=(d1), inner sep=10pt, label={[font=\huge]$d_{1}$}] (d1rec) {};
\end{scope}
\begin{scope}[local bounding box=d2,shift={(12,0)}]
	\scope[even odd rule]
	\clip \circleB;
	\clip \circleC;
	\fill[red] \circleB ;
	\endscope
\draw \circleB node[label={[label distance=8mm,xshift=-5]-90:$S_{1}$}] (S1){};
\draw \circleC node [label={[label distance=8mm,xshift=5]-90:$S_{2}$}] (S2){};
\draw (0.5,0) circle (2) node[label={[label distance=20mm]90:$M$}] (M) {};
\node[circle,minimum size=3mm,inner sep=0,fill=black] at (1.5,0)(s2){};
\draw (0.5,0.5) circle (3) node[label={[label distance=30mm]90:$X$}]{};
\node[draw, rectangle, fit=(d2), inner sep=10pt, label={[font=\huge]$d_{3}$}] (d2rec){};
\end{scope}
\begin{scope}[local bounding box=d3,shift={(0,-12)}]
	\scope[even odd rule]
	\clip \circleB;
	\clip \circleC;
	\fill[red] \circleB ;
	\endscope
\draw \circleB node[label={[label distance=8mm,xshift=-5]-90:$S_{1}$}] (S1){};
\draw \circleC node [label={[label distance=8mm,xshift=5]-90:$S_{2}$}] (S2){};
\draw (0.5,0) circle (2) node[label={[label distance=20mm]90:$M$}] (M) {};
\node[circle,minimum size=3mm,inner sep=0,fill=black] at (-0.5,0)(s1){};
\node[circle,minimum size=3mm,inner sep=0,fill=black] at (-2,0)(x1){};
\draw(s1)--(x1);
\draw (0.5,0.5) circle (3) node[label={[label distance=30mm]90:$X$}]{};
\node[draw, rectangle, fit=(d3), inner sep=10pt, label={[font=\huge]-90:$d_{4}$}] (d3rec) {};
\end{scope}
\begin{scope}[local bounding box=d4,shift={(12,-12)}]
	\scope[even odd rule]
	\clip \circleB;
	\clip \circleC;
	\fill[red] \circleB ;
	\endscope
\draw \circleB node[label={[label distance=8mm,xshift=-5]-90:$S_{1}$}] (S1){};
\draw \circleC node [label={[label distance=8mm,xshift=5]-90:$S_{2}$}] (S2){};
\draw (0.5,0) circle (2) node[label={[label distance=20mm]90:$M$}] (M) {};
\node[circle,minimum size=3mm,inner sep=0,fill=black] at (1.5,0)(s2){};
\node[circle,minimum size=3mm,inner sep=0,fill=black] at (3,0)(x1){};
\draw(s2)--(x1);
\draw (0.5,0.5) circle (3) node[label={[label distance=30mm]90:$X$}]{};
\node[draw, rectangle, fit=(d4), inner sep=10pt, label={[font=\huge]-90:$d_{2}$}] (d4rec){};
\end{scope}
\draw[<->,dashed,shorten <=5pt, shorten >=5pt](d1rec)--node[label={[font=\huge]90:$\#_{\mathrm{diag}}$}]{}(d2rec);
\draw[<->,dashed,shorten <=5pt, shorten >=5pt](d3rec)--node[label={[font=\huge]-90:$\#_{\mathrm{diag}}$}]{}(d4rec);
\draw[->,shorten <=15pt, shorten >=15pt] (d1rec) to [bend left=10] node[pos=0.75,label={[font=\huge]0:$T_{1}$}]{}(d4rec);
\draw[->,shorten <=15pt, shorten >=15pt] (d2rec) to [bend left=10] node[xshift=5,pos=0.75,label={[font=\huge]-90:$T_{3}$}]{}(d3rec);
\draw[->,shorten <=5pt, shorten >=5pt] (d3rec.north)--node[label={[font=\huge]180:$T_{4}$}]{}(d1rec.south);
\draw[->,shorten <=5pt, shorten >=5pt] (d4rec.north)--node[label={[font=\huge]0:$T_{2}$}]{}(d2rec.south);
\begin{scope}[local bounding box=T9,shift={(5,-5.5)},scale=0.25,line width=0.25mm]
	\scope[even odd rule]
	\clip \circleB;
	\clip \circleC;
	\fill[red] \circleB ;
	\endscope
\draw \circleB node(S1){};
\draw \circleC node(S2){};
\draw (0.5,0) circle (2) node (M) {};
\node[circle,minimum size=1mm,inner sep=0,fill=black] at (0.5,1.5)(m1){};
\node[circle,minimum size=1mm,inner sep=0,fill=black] at (1.5,0)(s2){};
\node[circle,minimum size=1mm,inner sep=0,fill=black] at (-0.5,0)(s1){};
\node[circle,minimum size=1mm,inner sep=0,fill=black] at (3,0)(x2){};
\draw(x2)--(s2);){};
\draw (0.5,0.5) circle (3) node{};
\node[draw, rectangle, fit=(T9), inner sep=10pt, label={[font=\huge]$d^{+}_{s_{1} \neg s_{1}}$}] (T9rec) {};
\end{scope}
\begin{scope}[local bounding box=T10,shift={(2.5,-5.5)},scale=0.25,line width=0.25mm] 
	\scope[even odd rule]
	\clip \circleB;
	\clip \circleC;
	\fill[red] \circleB ;
	\endscope
\draw \circleB node (S1){};
\draw \circleC node (S2){};
\draw (0.5,0) circle (2) node (M) {};
\node[circle,minimum size=1mm,inner sep=0,fill=black] at (0.5,1.5)(m1){};
\node[circle,minimum size=1mm,inner sep=0,fill=black] at (-0.5,0)(s1){};
\node[circle,minimum size=1mm,inner sep=0,fill=black] at (1.5,0)(s2){};
\node[circle,minimum size=1mm,inner sep=0,fill=black] at (-2,0)(x2){};
\draw(x2)--(s1);
\draw (0.5,0.5) circle (3) node{};
\node[draw, rectangle, fit=(T10), inner sep=10pt, label={[font=\huge]$d^{+}_{s_{2} \neg s_{2}}$}] (T10rec){};
\end{scope}
\draw[decorate, decoration={brace, amplitude=12pt, raise=15pt}]  (d1.north) -- coordinate (T1) (d2.north) node [label={[font=\huge,pos=0.5,label distance=30]90:$T_{5}$}]{};
\begin{scope}[local bounding box=MT1,shift={($(T1) + (-0.5,5)$)}]
	\scope[even odd rule]
	\clip \circleB;
	\clip \circleC;
	\fill[red] \circleB ;
	\endscope
\draw \circleB node[label={[label distance=8mm,xshift=-5]-90:$S_{1}$}] (S1){};
\draw \circleC node [label={[label distance=8mm,xshift=5]-90:$S_{2}$}] (S2){};
\draw (0.5,0) circle (2) node[label={[label distance=20mm]90:$M$}] (M) {};
\node[circle,minimum size=3mm,inner sep=0,fill=black] at (0.5,1.5)(m1){};
\node[circle,minimum size=3mm,inner sep=0,fill=black] at (1.5,0)(s2){};
\node[circle,minimum size=3mm,inner sep=0,fill=black] at (-0.5,0)(s1){};
\draw (0.5,0.5) circle (3) node[label={[label distance=30mm]90:$X$}]{};
\node[draw, rectangle, fit=(MT1), inner sep=10pt,label={[font=\huge]90:$d^{+}_{s_{1}s_{2}}$}] (MT1BOX) {};
\end{scope}
\draw[decorate, decoration={brace, mirror, amplitude=12pt, raise=15pt}]  (d3.south) -- coordinate (T2) (d4.south) node [label={[font=\huge,pos=0.5,label distance=30]-90:$T_{6}$}]{};
\begin{scope}[local bounding box=MT2,shift={($(T2) + (-0.5,-6.75)$)}]
	\scope[even odd rule]
	\clip \circleB;
	\clip \circleC;
	\fill[red] \circleB ;
	\endscope
\draw \circleB node[label={[label distance=8mm,xshift=-5]-90:$S_{1}$}] (S1){};
\draw \circleC node [label={[label distance=8mm,xshift=5]-90:$S_{2}$}] (S2){};
\draw (0.5,0) circle (2) node[label={[label distance=20mm]90:$M$}] (M) {};
\node[circle,minimum size=3mm,inner sep=0,fill=black] at (0.5,1.5)(m1){};
\node[circle,minimum size=3mm,inner sep=0,fill=black] at (1.5,0)(s2){};
\node[circle,minimum size=3mm,inner sep=0,fill=black] at (-0.5,0)(s1){};
\node[circle,minimum size=3mm,inner sep=0,fill=black] at (-2,0)(x1){};
\node[circle,minimum size=3mm,inner sep=0,fill=black] at (3,0)(x2){};
\draw(x2)--(s2);
\draw(x1)--(s1);
\draw (0.5,0.5) circle (3) node[label={[label distance=30mm]90:$X$}]{};
\node[draw, rectangle, fit=(MT2), inner sep=10pt,label={[font=\huge]-90:$d^{+}_{\neg s_{1}\neg s_{2}}$}] (MT2BOX) {};
\end{scope}
\draw[decorate, decoration={brace, mirror, amplitude=12pt, raise=15pt}]  (d1.west) -- coordinate (T3) (d3.west) node [label={[font=\huge,pos=0.5,label distance=30]180:$T_{7}$}]{};
\begin{scope}[local bounding box=MT3,shift={($(T3) + (-6,-0.75)$)}]
	\scope[even odd rule]
	\clip \circleB;
	\clip \circleC;
	\fill[red] \circleB ;
	\endscope
\draw \circleB node[label={[label distance=8mm,xshift=-5]-90:$S_{1}$}] (S1){};
\draw \circleC node [label={[label distance=8mm,xshift=5]-90:$S_{2}$}] (S2){};
\draw (0.5,0) circle (2) node[label={[label distance=20mm]90:$M$}] (M) {};
\node[circle,minimum size=3mm,inner sep=0,fill=black] at (0.5,1.5)(m1){};
\node[circle,minimum size=3mm,inner sep=0,fill=black] at (-0.5,0)(s1){};
\node[circle,minimum size=3mm,inner sep=0,fill=black] at (-2,0)(x1){};
\draw(x1)--(s1);
\draw (0.5,0.5) circle (3) node[label={[label distance=30mm]90:$X$}]{};
\node[draw, rectangle, fit=(MT3), inner sep=10pt,label={[font=\huge]90:$d^{+}_{ s_{1}\neg s_{2}}$}] (MT3BOX) {};
\end{scope}
\draw[decorate, decoration={brace, amplitude=12pt, raise=15pt}]  (d2.east) -- coordinate (T4) (d4.east) node [label={[font=\huge,pos=0.5,label distance=30]0:$T_{8}$}]{};
\begin{scope}[local bounding box=MT4,shift={($(T4) + (5,-0.75)$)}]
	\scope[even odd rule]
	\clip \circleB;
	\clip \circleC;
	\fill[red] \circleB ;
	\endscope
\draw \circleB node[label={[label distance=8mm,xshift=-5]-90:$S_{1}$}] (S1){};
\draw \circleC node [label={[label distance=8mm,xshift=5]-90:$S_{2}$}] (S2){};
\draw (0.5,0) circle (2) node[label={[label distance=20mm]90:$M$}] (M) {};
\node[circle,minimum size=3mm,inner sep=0,fill=black] at (0.5,1.5)(m1){};
\node[circle,minimum size=3mm,inner sep=0,fill=black] at (1.5,0)(s2){};
\node[circle,minimum size=3mm,inner sep=0,fill=black] at (3,0)(x2){};
\draw(x2)--(s2);
\draw (0.5,0.5) circle (3) node[label={[label distance=30mm]90:$X$}]{};
\node[draw, rectangle, fit=(MT4), inner sep=10pt,label={[font=\huge]90:$d^{+}_{s_{2} \neg s_{1}}$}] (MT4BOX) {};
\end{scope}
\end{tikzpicture}
}
\caption{Diagrammatic reconstruction of the Greimas semiotic square in the language of unitary spider diagrams. Each vertex is represented by a unitary diagram encoding the corresponding semic configuration, while the arrows as $T_{i}: d \vdash_{\mathcal{R}} d^{\prime}$ denote derivations in the inference system $\mathcal{R}$, indexed by the proof trees $T_{i}$. We use `\{' to indicate binary-input derivations as $d_{a \wedge b} \vdash_{\mathcal{R}} d_{ab}^{+}$. The relation $\#_{\mathrm{diag}}$ marks diagrammatic contrariety between $s_1, s_2$ and their negations $\neg s_{1}, \neg s_{2}$. The central region includes the two non-canonical derived meta-terms admitted by Proposition~\ref{The:MetaTerms}. The figure provides a global summary of the proof-theoretic and semantic structure developed in the paper.} 
\label{SDSquare}
\end{figure}
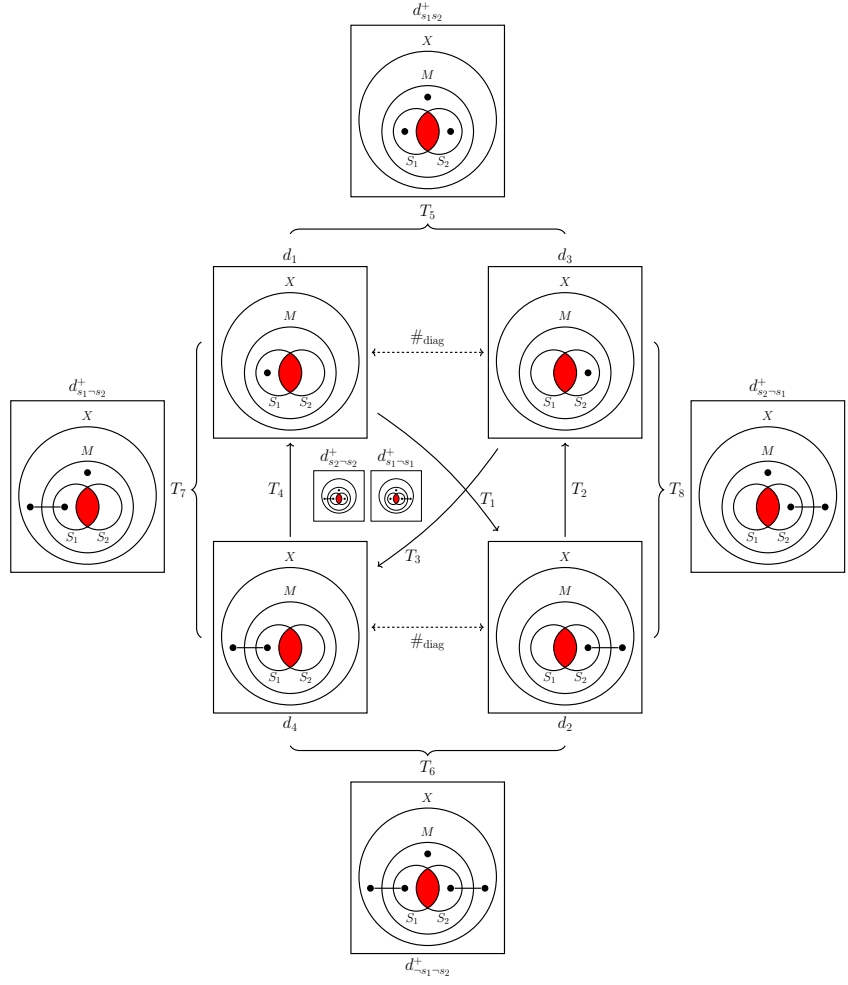

Finally, consider the neutral term $\mathbf{\bar{S}}$ (Figure~\ref{ProofTree6}), derived via
\[
T_{6} : d_{\neg s_1 \wedge \neg s_2} \;\vdash_{\mathcal R}\; d^{+}_{\neg s_1 \neg s_2}.
\]
We note here that the derivation of $\mathbf{\bar{S}}$ involves additional applications of \textsc{ConjunctionElimination} compared to the other meta-terms. This reflects the fact that the input configurations $\neg s_{2}$ and $\neg s_{2}$ are both virtualised and do not share a common determinate habitat. As a result, their combination cannot be effected directly within a single zone, and must instead proceed via intermediate projections and duplications of spiders. This proof-theoretic distinction corresponds to the semantic characterisation of the neutral term as in the thymic category of aphoria such that both contraries are suspended rather than jointly actualised. Our construction combines these indeterminate habitats and introduces a new witness in $M$ that is supported by the joint exclusion of $s_1$ and $s_2$. 
Crucially, the derivation does not resolve this indeterminacy into either pole, but instead preserves it at the level of the resulting configuration. 

From a semantic perspective, the neutral term thus corresponds to a configuration in which both contraries are suspended. In the H\'{e}bert example this is exemplified by the non-gender actant of an /angel/ given that the neutral term $\mathbf{\bar{S}}$ ``only contains those elements \textit{marked} as `neither one', not the elements that simply belong to the resideual class of the square.'' \citep[page 87]{Hebert2020} Thus, the derived witness in $M$ does not mediate a synthesis of $s_1$ and $s_2$, but rather represents a hypotactic region within the semantic universe of $X$ such that their absence is jointly realised. This captures the interpretation of $\mathbf{\bar{S}}$ as a neutral or transcendent term, whose semantic content is defined not by the presence of either seme, but by their simultaneous exclusion within the diagrammatic hierarchy.

Taken together, these cases show that the four canonical meta-terms of the semiotic square arise from a single proof-theoretic mechanism, with their semantic differences determined by the structure of their input configurations. It also follows closely the description Greimas gives on the identification of syntagmatic units given that 

\begin{quote}
the construction of the sememes [meta-terms] is based on a semic analysis allowing us to organize the occurrences in parallel classes which are disjoined because of semic oppositions. In other words the model is the result of a semic description which rejoins and verifies the actantial structure. \citep[page 265]{Greimas1983}
\end{quote}

The operation `$+$' is thus uniformly realised as a derivational lift on all semes of the square, while the distinction between complex, neutral, and deictic terms emerges from the interaction between actualised and virtualised components in the underlying derivations. The constructions developed in this section are summarised diagrammatically in Figure~\ref{SDSquare}.

\section{Conclusion}

In this article, we have developed a diagrammatic proof system for a fragment of structural semantics inspired by the Greimas semiotic square. By interpreting unitary spider diagrams as semantic configurations, we have shown how the core relations of the square of contrariety, negation, and implication may be realised as transformations governed by a system of inference rules. In this setting, the arrows of the semiotic square are no longer treated as informal relations, but as derivations witnessed by proof trees in a formally specified calculus.

Our main result establishes that the Greimasian meta-terms arise uniformly from a proof-theoretic construction. Given a conjunctive pair of semic configurations, a fixed derivation schema produces a new diagram containing a witness in the meta-level domain $M$. This provides a precise interpretation of the operation `+’ not as a logical or set-theoretic combination, but as a derivational lift in which a new semantic configuration is constructed from simpler components. The resulting meta-term does not collapse the incompatibility of its inputs, but mediates their joint representation at a higher level of the diagrammatic hierarchy.

A key feature of the framework is that diagrammatic negation is not identified with Boolean complement. Rather, it corresponds to a restricted semantic counter-position determined by the structure of admissible zones in the diagram. In this way, the system captures the non-classical character of opposition in structural semiotics, where negation introduces indeterminacy rather than exhaustive exclusion. Diagrammatic implication, in turn, resolves this indeterminacy through successive habitat refinements, yielding determinate configurations from virtualised ones.

The resulting proof system thus provides a compositional and generative account of the semiotic square. It makes explicit the procedures by which semantic configurations are constructed and transformed, while preserving the hierarchical organisation of semes and sememes within a structured semantic universe $S_1, S_2 \subseteq M \subseteq X$. Moreover, the derivational framework extends beyond the four canonical meta-terms, yielding a broader class of admissible configurations and thereby situating the traditional square within a larger space of diagrammatically definable structures.

More broadly, this work illustrates how diagrammatic reasoning systems can be used to formalise non-standard semantic operations within a unified inferential setting. By combining the expressive resources of spider diagrams with the conceptual apparatus of structural semiotics, we obtain a framework in which meaning is not merely represented, but constructed through rule-governed transformations. This suggests further avenues for the application of diagrammatic logic to the analysis of semantic structures in linguistics, narrative theory, and related domains.

\backmatter


\bibliography{GreimasSquare}

\end{document}